\let\latexaddtocontents\addtocontents
\documentclass[a4paper,hidelinks,accepted=2023-12-01]{quantumarticle}
\let\addtocontents\latexaddtocontents
\pdfoutput=1
\usepackage{amsmath}
\usepackage{amssymb}
\usepackage{eufrak}
\usepackage{stackrel}
\usepackage{tikz-cd}
\usepackage{tikz}
\usepackage[cm]{fullpage}
\usepackage{amsthm}
\usepackage{mathtools}
\usepackage{enumitem}
\usepackage{bbold}
\usepackage[normalem]{ulem}
\usepackage{hyperref}
\usepackage{cleveref}
\usepackage{appendix}
\usepackage{comment}
\usepackage{subcaption}

\allowdisplaybreaks

\usepackage[braket, qm]{qcircuit}

\usepackage[backend=bibtex,sorting=none,citestyle=numeric-comp]{biblatex}
\title{The Min-Entropy of Classical-Quantum Combs for Measurement-Based Applications}
\author{Isaac D. Smith}
\email{isaac.smith@uibk.ac.at}
\author{Marius Krumm}
\author{Lukas J. Fiderer}
\author{Hendrik Poulsen Nautrup}
\author{Hans J. Briegel}
\affiliation{Institute for Theoretical Physics, UIBK, 6020 Innsbruck, Austria}

\newtheorem{Theorem}{Theorem}[section]

\newtheorem{Lemma}[Theorem]{Lemma}
\newtheorem{Proposition}[Theorem]{Proposition}

\theoremstyle{definition}
\newtheorem{Definition}{Definition}[section]

\theoremstyle{definition}

\theoremstyle{definition}

\theoremstyle{definition}

\theoremstyle{definition}

\DeclareMathOperator{\Odd}{Odd}
\DeclareMathOperator{\Tr}{Tr}

\DeclareMathOperator{\diag}{diag}

\DeclareMathOperator{\Od}{Odd}

\DeclareMathOperator{\Pa}{Pa}

\DeclareMathOperator{\BQC}{BQC}
\DeclareMathOperator{\MBQC}{MBQC}
\DeclareMathOperator{\client}{client}
\DeclareMathOperator{\Comb}{Comb}

\renewcommand{\H}{\mathcal{H}}
\renewcommand{\L}{\mathcal{L}}

\newcommand{\M}{\mathcal{M}}

\begin{document}

\maketitle

\begin{abstract}
Learning a hidden property of a quantum system typically requires a series of interactions. In this work, we formalise such multi-round learning processes using a generalisation of classical-quantum states, called classical-quantum combs. Here, ``classical'' refers to a random variable encoding the hidden property to be learnt, and ``quantum'' refers to the quantum comb describing the behaviour of the system. The optimal strategy for learning the hidden property can be quantified by applying the comb min-entropy (Chiribella and Ebler, \textit{NJP}, $2016$) to classical-quantum combs. To demonstrate the power of this approach, we focus attention on an array of problems derived from measurement-based quantum computation (MBQC) and related applications. Specifically, we describe a known blind quantum computation (BQC) protocol using the combs formalism and thereby leverage the min-entropy to provide a proof of single-shot security for multiple rounds of the protocol, extending the existing result in the literature. Furthermore, we consider a range of operationally motivated examples related to the verification of a partially unknown MBQC device. These examples involve learning the features of the device necessary for its correct use, including learning its internal reference frame for measurement calibration. We also introduce a novel connection between MBQC and quantum causal models that arises in this context.
\end{abstract}


\section{Introduction} \label{sec:intro}

With the rapid development of quantum technology, a plethora of increasingly complex quantum devices has become available. Already, a range of noisy-intermediate scale quantum computers (NISQ) \cite{NISQ} are accessible via the internet. Year on year, these devices increase in size and quality, as evidenced by the growing number of addressable qubits, the lengthening coherence times and improving gate fidelity.

An important parallel development to that of the devices themselves, is in how they interconnect. The current and future progress of quantum communication networks \cite{QuantumRepeaters, QuantumNetworkBook, QuantumInternetArrived, Satellite, DiamondNetwork, NetworkEntanglement, QKD307km, 2000kmFibre, QKDreview,cacciapuoti2019quantum} aims to create a global quantum internet of networked quantum computers \cite{QuantumInternetAlliance, Roadmap}, with the attendant benefits in communication and cryptography. Current implementations of such networks largely focus on protocols for quantum key distribution, however a variety of other significant applications appear to be feasible in the near-term, including distributed quantum computation \cite{EfficientDistributedQuantum, DistributedQuantum}, quantum position verification \cite{QuantumPositionVerification}, and blind quantum computation \cite{fitzsimons2017private}. 

With the increasing sophistication of quantum networks, a full description of the underlying physical systems and dynamics quickly becomes intractable. In many cases, a quantum information-theoretic description is both expressive enough to capture the relevant characteristics of the network and manageable enough to make progress on questions surrounding verification and functionality. Since quantum networks consist of a number of communication channels connecting different nodes, it is natural to consider their information-theoretic description as a concatenation of quantum channels connecting different Hilbert spaces. The results operators are commonly referred to as \textit{quantum combs} \cite{chiribella2009theoretical,chiribella2008quantum}.

Just as density matrices and quantum operators are versatile representations of quantum states and single time-step quantum dynamics respectively \cite{nielsen_chuang_2010}, quantum combs are a widely applicable representation of dynamics that occur over multiple time-steps especially in the presence of correlations with an environment. They are represented by operators acting on the Hilbert spaces corresponding to controllable systems only, which constitutes a distinct advantage of their use: any forwarding of information between systems mediated by the environment is modelled by an operator of the same size. Accordingly, quantum combs have become a fruitful tool to describe e.g., non-Markovian noise in quantum systems \cite{pollock2018non} and multi-round quantum information processing protocols \cite{chiribella2009theoretical,chiribella2008quantum}. 

A crucial discovery within the field of theoretical quantum networks is that the optimisation of multi-round quantum protocols with respect to different performance measures can be framed as \textit{semi-definite programs} \cite{chiribella2016optimal}. For a given protocol, the maximal ability to produce quantum correlations with its output leads to a notion of \textit{min-entropy} for quantum combs, a generalisation of the conditional min-entropy of quantum states \cite{konig2009operational,renner2008security}. In this work, we use the comb min-entropy to quantify how well correlations can be generated with an unknown classical parameter underlying the quantum networks via an optimal strategy, thereby quantifying the ultimate limits on what can be learnt about the parameter. These networks are modelled as classical-quantum combs, which are generalisations of classical-quantum states \cite{wilde_2017}, where the classical variable indexes as family of quantum combs rather than quantum states. Due to the broad applicability of both the combs formalism and entropic quantities in quantum information theory, this classical-quantum comb min-entropy framework represents a powerful tool for analysing a range of information processing tasks. In particular, this framework allows for the analysis of the structural properties of a protocol, which cannot be represented using classical-quantum states. We demonstrate the utility of this approach via a range of operationally motivated examples based on measurement-based quantum computation.

Measurement-based quantum computing (MBQC) \cite{raussendorf2001one, briegel2009measurement, raussendorf2003measurement, raussendorf2001computational, jozsa2006introduction, gross2007novel} is a well-known paradigm for quantum computation distinct from the circuit model formalism, in which computation is driven by measurements rather than unitary evolutions. It provides an exemplary test-bed for the application of classical-quantum combs and the comb min-entropy for a number of reasons. Firstly, the computation is inherently sequential: single qubit measurements are performed on graph states \cite{briegel2001persistent,hein2006entanglement,hein2006ent,hein2004multiparty, mhalla2011graph} following a specified order \cite{raussendorf2001computational, browne2007generalized, danos2006determinism, markham2014entanglement, danos2007measurement}. Secondly, due to connections with stabiliser quantum mechanics \cite{van2004graphical,van2005local}, investigations into certain aspects of MBQC remain tractable despite its multi-partite, high-dimensional nature. Finally, there are pragmatic reasons to consider examples based on MBQC: as a universal computation paradigm, it represents one pathway to fully-fledged quantum computation \cite{walther2005experimental,raussendorf2007topological,Tame_07} and forms the basis of a range of cryptographic computation protocols collectively known as blind quantum computation (BQC) \cite{broadbent2009universal,mantri2017flow,morimae2013blind,morimae2014verification, fitzsimons2017private}.

The first classical-quantum comb we consider models a specific BQC protocol, that of Mantri et al. \cite{mantri2017flow}. This protocol consists of entirely classical communication between a client and a (quantum) server, allowing the client to obtain the results of a measurement-based quantum computation while maintaining no quantum capabilities. The classical parameter in this context encodes the secret choice of computation by the client which an untrustworthy server may attempt to discover. Due to the restriction to classical communication, the quantum comb part of the classical-quantum comb, the part that models a round of the protocol, is entirely classical (i.e. we consider a classical-classical comb). The utility of representing the protocol as a comb arises from the ability to apply the min-entropy in the security analysis of the protocol. In contemporary quantum cryptography, the min-entropy is typically taken as the gold standard for statements of security \cite{portmann2022security}. Consequently, by using the comb min-entropy we provide here a stronger security analysis for a single round of the protocol than the original paper \cite{mantri2017flow}. Furthermore, we provide an analysis of the security of the protocol under multiple rounds of computation, indicating that more information is leaked than the in the single-round case. This multi-round analysis, as well as the application of the comb min-entropy for BQC security analysis in general, are novel in this work.

Not all measurement based quantum computation must be performed as part of a cryptographic protocol, it could also be implemented via a device in a laboratory. In such a case, the device must be verified to check that it is functioning correctly. We investigate a variety of scenarios where some aspect of an MBQC device is unknown and whose details are required in order to perform correct computations. For example, the specific measurements applied to perform the computation must be calibrated to the internal state preparation of the device. In all scenarios we consider, the classical-quantum combs are built from a comb based upon the method of measurement adaptation required for MBQC \cite{browne2007generalized,danos2006determinism}. Aside from the pragmatic considerations related to device verification, we also establish a novel connection between MBQC and quantum causal models \cite{barrett2019quantum,costa2016quantum,allen2017quantum,ried2015quantum,fitzsimons2015quantum,chiribella2019quantum,kubler2018two} which may be of foundational interest.

The remainder of this paper is structured as follows. The next subsection provides a brief and informal presentation of the key results in this work. Quantum combs are presented in \Cref{sec:Q_combs}, including classical-quantum combs that are or main focus. The comb min-entropy is presented in \Cref{sec:Q_comb_min_ent}, including some results specific to classical-quantum combs. \Cref{sec:MBQC} introduces the background on MBQC and the BQC protocol of \cite{mantri2017flow} required for example cases we investigate. \Cref{sec:CBQC} introduces the combs for the BQC protocol (\Cref{subsec:D_client}) and provides the security analysis (\Cref{subsec:BQC_results}). \Cref{sec:grey_MBQC} considers a partially unknown MBQC device. Its description in terms of combs is given in \Cref{subsec:gflow_qcomb} with examples related to learning the unknown properties given in \Cref{subsec:meas_planes,subsec:noise,subsec:meas_calibration}. A connection between MBQC and quantum causal models is discussed in \Cref{subsec:caus_inf}. \Cref{sec:discussions} discusses future directions and concludes.

\subsection{Our Contributions}

The foremost contribution of this paper consists in presenting the combination of classical-quantum combs and the comb min-entropy as an apposite framework for analysing quantum information processing protocols. Many protocols of interest are sequential and dependent on unknown parameters, some of which may represent features of the protocol in its entirety. Accordingly, these protocols are aptly represented by classical-quantum combs, with the parameters only able to be learnt through sequential interaction. The use of the comb min-entropy as we propose here provides knowledge of how best to learn these parameters, which can inform how to probe these protocols in practice. The classical-quantum comb min-entropy framework thus constitutes a methodology that is both analytically and numerically applicable in a diverse array of scenarios.

The following is an informal summary of some of the specific technical contributions of this paper, in their order of appearance. 
\begin{itemize}
	\item We established upper and lower bounds for the min-entropy for classical-quantum combs where the quantum combs are all diagonal in the same basis (\Cref{prop:bayesian_classical_comb}). 
	\item Using these bounds, we provide a security analysis of a single as well as multiple rounds of the BQC protocol of \cite{mantri2017flow} (\Cref{thm:single_shot_entropy} and \Cref{thm:multi_round}). For a single-round, our framework raises the original security analysis to modern cryptographic standards which use the min-entropy. The multi-round analysis is new in this work and indicates that the security of the protocol is compromised under multiple rounds.
	\item We provide an array of results quantifying the ability to verify different aspects of an MBQC device, which range from being able to learn the required set of measurements with certainty (\Cref{subsec:meas_planes}) to not being able to learn the structure of the device at all (\Cref{prop:no_causal_learning}). In the latter context, we demonstrate that noise in the device is in fact beneficial for learning its structure (\Cref{subsec:noise}).
	\item We establish a novel connection between MBQC and quantum causal models (\Cref{prop:MBQC_QCM}), and provide examples pertaining to learning an unknown parameter representing causal structure. This connection also has broader consequences e.g., for investigations into the structure of unitary transformations. 
\end{itemize}


\section{Quantum Combs Formalism} \label{sec:Q_combs}

In quantum information theory \cite{nielsen_chuang_2010}, states and evolutions of quantum systems are typically represented as operators on Hilbert spaces. Throughout this work, (finite dimensional) Hilbert spaces are denoted by $\H$ with the associated space of linear operators written $\L(\H)$. When multiple systems are under consideration, subscripts are used to remove ambiguity: i.e. $\H_{A}$ denotes the Hilbert space corresponding to system $A$. The state of system $A$ is represented by a density matrix $\rho \in \L(\H_{A})$ which is both positive semi-definite (PSD) and of unit trace. A quantum channel from system $A$ to system $B$ is a completely-positive and trace-preserving (CPTP) linear map from $\L(\H_{A})$ to $\L(H_{B})$, which can be equivalently represented via its Choi-Jamio\l{}kowski (CJ) state \cite{choi1975completely,jamiolkowski1972linear}: an element $D \in \L(\H_{A} \otimes \H_{B})$ which is PSD and satisfies $\Tr_{B}[D] = I_{A}$, with $\Tr_{B}$ denoting the partial trace over system $B$ and $I_{A}$ the identity operator on system $A$ (for completeness, further details of this equivalence are given in \Cref{subsec:app_qcombs}).

It is a fact that every quantum system is an open quantum system \cite{breuer2002theory}, which is to say, one that interacts with an `environment'. When using CPTP maps to describe the evolution of such a system, a number of assumptions are being made, often implicitly. Namely, it is assumed that the initial joint state of the system and environment is a product state and furthermore that the state of the environment and the joint evolution are time-independent. In some cases, these assumptions are reasonable approximations to the real state of affairs, however both experimental and theoretical evidence indicates that these assumptions are often not valid, and in such cases, neither is the application of the CPTP formalism \cite{pechukas_94,Alicki_95,pechukas_95,Royer_96} (see also \cite{milz2017introduction} for an overview). 

One resolution of this problem is to extend the formalism to allow for the removal of the above assumptions. The formalism of quantum combs \cite{chiribella2008quantum,chiribella2009theoretical} provides one such extension. A quantum comb is a linear operator on a number of Hilbert spaces which represents a sequence of channels, where some channels may take inputs contingent on (partial) outputs of previous ones. Accordingly, these operators provide a fruitful representation of quantum networks, with potentially complex connectivity. The formalism of quantum combs is closely related to that of process tensors in the modelling of non-Markovian open quantum systems \cite{pollock2018non,milz2017introduction}, that of quantum games \cite{gutoski2007toward} and the treatment of quantum causality \cite{oreshkov2012quantum,allen2017quantum,costa2016quantum,barrett2019quantum,barrett2021cyclic,d2018causality}. The definition of quantum comb used throughout this work is as follows; see \Cref{fig:combs_with_link} for the labelling of Hilbert spaces.

\begin{Definition}[\cite{chiribella2008quantum,chiribella2009theoretical}] \label{def:qcomb} A \textbf{quantum comb} is a positive semi-definite operator $D \in \L\left(\bigotimes_{j = 1}^{n} \H_{A_{j}} \otimes \H_{B_{j}}\right)$ for which there exists a sequence of positive semi-definite operators $D_{k} \in \L\left(\bigotimes_{j = 1}^{k} \H_{A_{j}} \otimes \H_{B_{j}}\right)$, $k = 0, ..., n$ which satisfy
\begin{align}
\Tr_{B_{k}}\left[D_{k} \right] = I_{A_{k}} \otimes D_{k-1} \quad \quad \forall k \in \{1, ..., n\}, \label{eq:comb_trace_conds}
\end{align}

with $D_{n} = D$ and $D_{0} = 1$.
\end{Definition}

\begin{figure*}
\includegraphics[width=\textwidth]{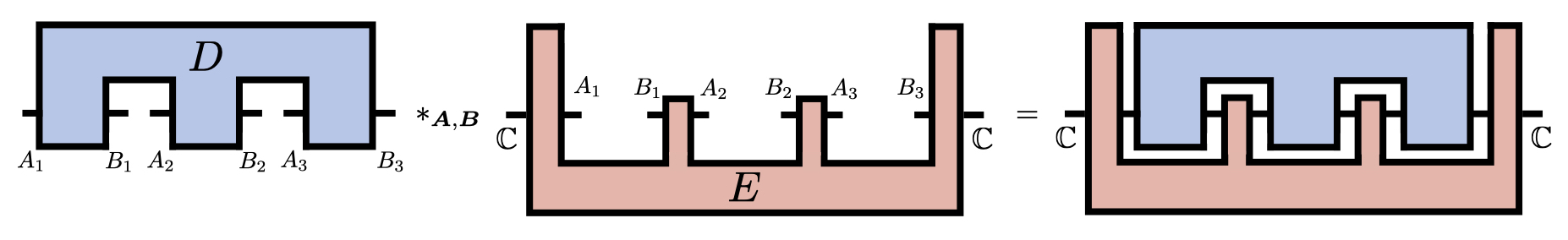}
\centering
\caption{Quantum combs are operators representing the evolution of a system over a series of time steps. The comb $D$ above receives an input initially at time $A_{1}$, produces an output at $B_{1}$, receives an input again at $A_{2}$, and so on from left to right. The operator $E$ is dual to the comb $D$ and represents an element of a generalised instrument, which is analogous to an element of a positive operator-valued measure. The operators $D$ and $E$ can be contracted using the link product, denoted using the $\ast_{\boldsymbol{A},\boldsymbol{B}}$ above, to produce a probability value (represented on the right as the contraction over all non-trivial spaces of the two operators).}
\label{fig:combs_with_link}
\end{figure*}

The requirement of positivity and the sequence of partial traces are inherited from the CPTP maps that underlie the comb (see \Cref{subsec:app_qcombs} for more details). To the extent possible, we use the notation $A_{i}$ for ``incoming'' Hilbert spaces in the comb definition and $B_{j}$ for all ``outgoing'' ones, as in \Cref{fig:combs_with_link}, where ``ingoing'' and ``outgoing'' identify which spaces are present in the partial trace conditions of \Cref{eq:comb_trace_conds}. For additional clarity, we denote the set of combs with these conditions as $\Comb(A_{1} \rightarrow B_{1},...,A_{n} \rightarrow B_{n})$, where the $\rightarrow$ notation mirrors the typical notation of a map $A \rightarrow B$ from input space $A$ to output space $B$. In any case where we wish to emphasise that a Hilbert space is one-dimensional, we replace the corresponding $A_{i}$ or $B_{j}$ with $\mathbb{C}$.

As a sequence of channels, quantum combs represent multi-time evolutions of quantum systems that are deterministic, as witnessed by the normalisation to unity in the above definition. Just as measurement channels are an important sub-class of quantum channels, it will be important for the coming treatment of the comb min-entropy (\Cref{sec:Q_comb_min_ent}) to have an analogous notion for quantum combs. The following definition of a generalised instrument parallels that of a positive operator-valued measure (POVM). 

\begin{Definition}[\cite{chiribella2009theoretical}] \label{def:gen_instrument} A \textbf{probabilistic comb} is a PSD operator $E$ such that $E \leq F$ for some (deterministic) comb $F$. A \textbf{generalised instrument} is a family of probabilistic combs $\{E_{i}\}$ such that $E = \sum_{i} E_{i}$ is a (deterministic) comb.
\end{Definition} 

In this work, we will consider generalised instruments that are dual to combs $D$ (see the operator labelled $E$ in \Cref{fig:combs_with_link}). Due to this duality, the combs $E$ have input and output spaces labelled differently to, but compatibly with, the combs $D$. For example, for $D \in \Comb(A_{1} \rightarrow B_{1}, ..., A_{n} \rightarrow B_{n})$ one can consider combs $E \in \Comb(\mathbb{C} \rightarrow A_{1}, B_{1} \rightarrow A_{2}, ..., B_{n} \rightarrow \mathbb{C})$. Note that for any probabilistic comb $E_{i}$, there exists a generalised instrument $E$ that contains it \cite[Section IV.D,][]{chiribella2009theoretical}.

To denote the contraction of a comb with another operator, such as an element of a generalised instrument (see again \Cref{fig:combs_with_link}) or even just the CJ state of a channel in $\L(\H_{B_{i}} \otimes \H_{A_{i+1}})$, it is convenient to have the following notation:

\begin{Definition}[\cite{chiribella2009theoretical}] \label{def:link_product}
For two operators $M \in \L(\H_{A} \otimes \H_{B})$ and $N \in \L(\H_{B} \otimes \H_{C})$, the \textbf{link product} is defined as:
\begin{align}
M*_{B} N := \Tr_{B}\left[(M^{\top_{B}}\otimes I_{C}) \cdot (I_{A} \otimes N)\right],
\end{align}

where the subscript on $\ast$ indicated the common subsystem $B$ over which $M$ and $N$ contract and where $\top_{B}$ denotes the partial transpose over the system $B$.
\end{Definition}

To conclude this section, we introduce two special types of comb that are pertinent to the applications considered below: classical-quantum combs and classical combs. The former represent families of quantum combs indexed by a (finite, discrete) random variable $X$, including a prior distribution over $X$. This is a direct extension of classical-quantum states such as those used in quantum hypothesis testing \cite{holevo2011probabilistic,holevo73,helstrom1967detection,holevo1973statistical}. The latter type represent combs that are diagonal in a given basis.

\begin{Definition} Let $X$ be a finite and discrete random variable and let $\{\sigma\}_{x \in X}$ denote a family of combs with $\sigma_{x} \in \Comb(A_{1} \rightarrow B_{1}, ..., A_{n} \rightarrow B_{n})$. A \textbf{classical-quantum comb} is a PSD operator
\begin{align}
D = \sum_{x \in X} P(x)\ket{x}\!\!\bra{x} \otimes \sigma_{x}  \label{eq:X_indexed_comb}
\end{align}

where $\{\ket{x}\}_{x}$ is an orthonormal basis of $\H_{X}$ and $P(x)$ denotes the prior distribution over $X$. 
\end{Definition}

A diagrammatic depiction of $D$ is shown in \Cref{fig:classical_quantum_comb}. It is easily verifiable that $D$ is indeed a comb, and in fact can be written as a comb in two different ways, as summarised by the following proposition (see \Cref{subsec:app_qcombs} for the proof):

\begin{Proposition} \label{prop:X_A_independence} The operator $D$ as above is an element of both $\Comb(A_{1} \rightarrow B_{1}, ..., A_{n} \rightarrow B_{n}, \mathbb{C} \rightarrow X)$ and $\Comb(\mathbb{C} \rightarrow X, A_{1} \rightarrow B_{1}, ..., A_{n} \rightarrow B_{n})$.
\end{Proposition}

The latter fact disallows signalling from the inputs to the $\sigma_{x}$ combs to the output at $X$ (cf. \cite[Section III.C,][]{chiribella2009theoretical}). This is significant for our purposes as it establishes the independence of $X$ from any strategy used to learn about its value; the value of $X$ is an ``objective'' feature of the system under investigation.

In some of the applications we consider below (see \Cref{sec:CBQC}), the combs $\sigma_{x}$ can all be diagonalised in the same basis. Accordingly, we can modify the criteria of the general definition of combs given above to account for this case: 

\begin{figure*}
\begin{subfigure}{0.5\textwidth}
\centering
\includegraphics[width=0.87\textwidth]{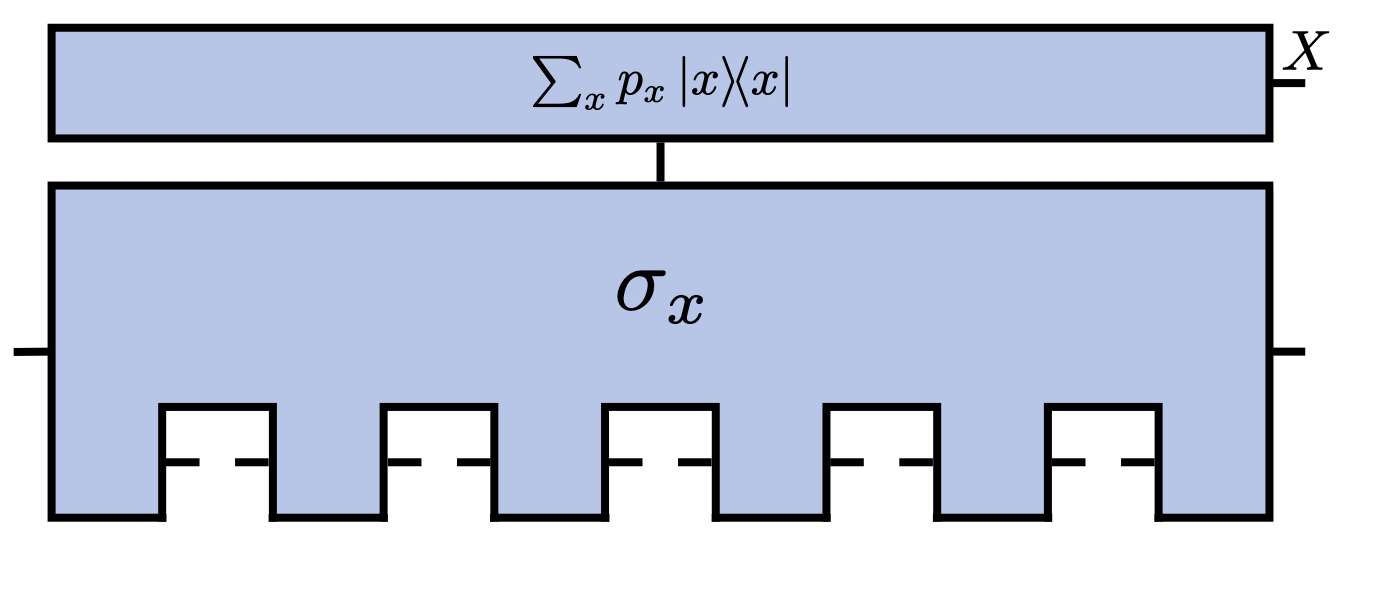}
\caption{Classical-quantum comb $D$}
\label{fig:classical_quantum_comb}
\end{subfigure}
\hfill
\begin{subfigure}{0.5\textwidth}
\centering
\includegraphics[width=\textwidth]{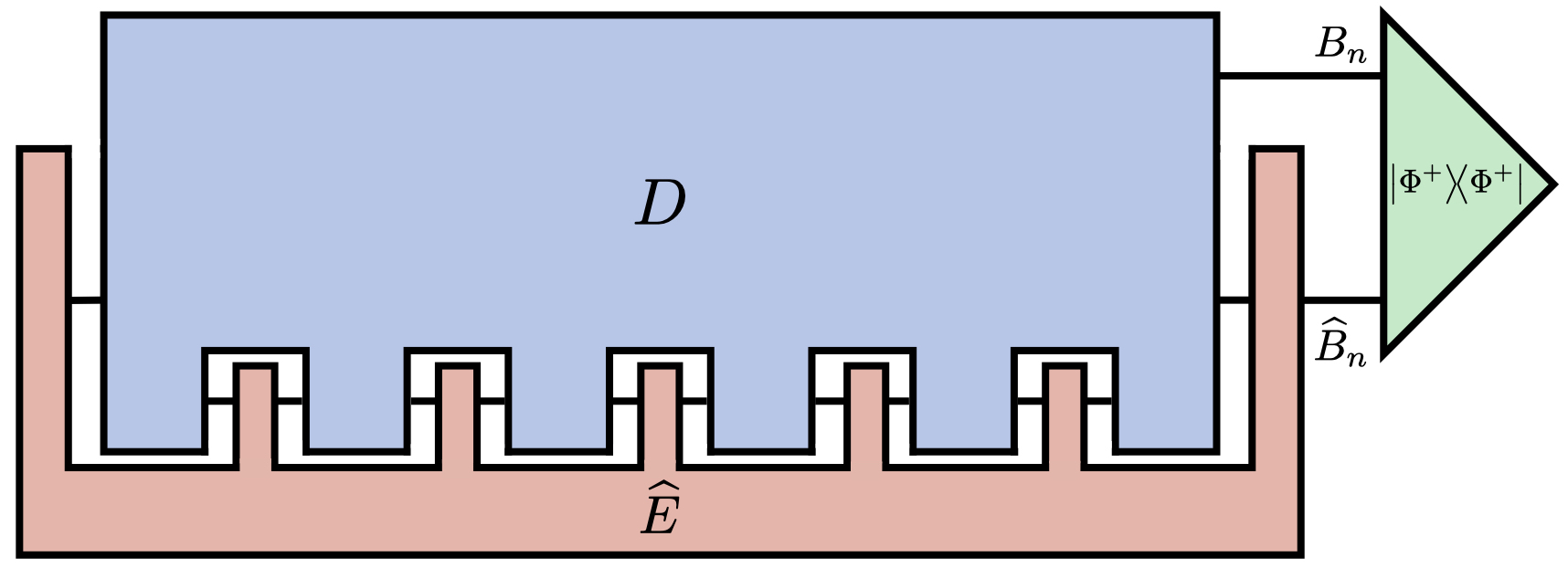}
\caption{The comb min-entropy for a general comb $D$}
\label{fig:comb_min_ent}
\end{subfigure}
\caption{In this work, we consider an extension to classical-quantum states, called classical-quantum combs, where a classical random variable $X$ indexes a set of quantum combs $\sigma_{x}$ (rather than quantum states). (a) The variable $X$ is considered to be unknown and as such is represented as an inaccessible system (upper part). The combs $\sigma_{x}$ are contingent on the value of $X$ and describe the dynamics of a system which can be interacted with (lower part) over a series of time-steps. By interacting with the accessible system, updated knowledge of $X$ can be obtained. (b) To calculate the comb min-entropy of a comb $D$, one optimises over all possible ways to interact with $D$, denoted by $\widehat{E}$ to maximise correlations with (and hence knowledge of) the output space of $D$. If $D$ is a classical-quantum comb as in (a), then the output space is the unknown classical variable $X$ and the comb min-entropy quantifies how much can be known about $X$ when using the optimal strategy.}
\label{fig:c_q_combs}
\end{figure*}

\begin{Definition} \label{def:classical_comb} A positive semi-definite operator $\sigma \in \L\left(\bigotimes_{j = 1}^{n} \H_{A_{j}} \otimes \H_{B_{j}}\right)$ is a \textbf{classical comb} if it can be written as
\begin{align}
\sigma = \sum_{\boldsymbol{a}, \boldsymbol{b}} P(\boldsymbol{b}|\boldsymbol{a})\ket{\boldsymbol{a},\boldsymbol{b}}\!\!\bra{\boldsymbol{a},\boldsymbol{b}}_{\boldsymbol{A},\boldsymbol{B}}
\end{align}

where $\{\ket{\boldsymbol{a},\boldsymbol{b}} \}$ is an orthonormal basis of $\H_{\boldsymbol{A},\boldsymbol{B}} := \bigotimes_{j = 1}^{n} \H_{A_{j}} \otimes \H_{B_{j}}$ and where $P(\boldsymbol{b}|\boldsymbol{a})$ is a conditional probability kernel for which there exists a sequence of marginalised conditional distributions $P^{(k)}(b_{1},...,b_{k}|a_{1},...,a_{k})$, $k = 0, ..., n$ that satisfy, for all $j = 2, ..., n$:
\begin{align}
\sum_{b_{j}}P^{(j)}&(b_{1},...,b_{j}|a_{1},...,a_{j}) \nonumber\\
&= P^{(j-1)}(b_{1},...,b_{j-1}|a_{1},...,a_{j-1}) \label{eq:marginalisation_independence}
\end{align}

with $P^{(n)} = P$ and $P^{(0)} := \sum_{b_{1}}P^{(1)}(b_{1}| a_{1}) = 1$.
\end{Definition}

The conditions in \Cref{eq:marginalisation_independence} are a straightforward consequence of the partial trace conditions in the general combs definition. It is worth noting that other notions of classical comb exists, such as that presented in \cite{Milz_20} which defines classicality via indistinguishability under contraction with identity channels or de-phasing channels, which is more general than diagonal combs in the case where $\H_{B_{j}} \cong \H_{A_{j+1}}$ for each $j$. Since we consider classical combs for which this restriction on the Hilbert spaces is not possible, the above definition is most appropriate for our purposes.


\section{Comb Min-Entropy} \label{sec:Q_comb_min_ent}

In this section, we define and motivate the primary tool of analysis used in the remainder of this work: the comb min-entropy \cite{chiribella2016optimal}. The comb min-entropy is an extension of the conditional min-entropy for quantum states \cite{renner2008security, konig2009operational, tomamichel2009fully, tomamichel2015quantum} which has found application in e.g., quantum cryptography \cite{renner2008security,portmann2022security}, hypothesis testing \cite{holevo1973statistical,yuen75} and quantum metrology \cite{meyer2023quantum}, due to its interpretation as measuring the distinguishability of the states in question. After defining the comb min-entropy, we consider its operational meaning with particular emphasis on the cases relevant for this work. In the following subsection, we collect a number of results that support the analysis of the specific combs considered in later sections.

To develop some intuition for the comb min-entropy, consider first a probability distribution $P(B)$ over a discrete random variable $B$. Whereas the Shannon entropy measures the average surprisal of the variable $B$, i.e. $H(B) := -\sum_{b} P(B=b)\log(P(B=b))$, the \textit{min-entropy} instead measures the minimal surprisal: $H_{\min}(B) := -\log(\max_{b} P(B=b))$. If $B$ is conditioned on another variable, say $A$, then the \textit{conditional} min-entropy measures the minimal surprisal when optimising over the conditioned variable as well: $H_{\min}(B|A) := -\log(\max_{b,a}P(B=b|A=a))$. The \textit{comb} min-entropy is a generalisation of the conditional min-entropy to the case where the negative logarithm is taken over probabilities arising from quantum operators that are conditioned on multiple inputs:

\begin{Definition}[\cite{chiribella2016optimal}] \label{def:comb_min_entropy} Let $D \in \Comb(A_{1} \rightarrow B_{1}, ..., A_{n} \rightarrow B_{n})$. The \textbf{comb min-entropy} of $B_{n}$ given $A_{1},B_{1},...,A_{n-1},B_{n-1}$ for $D$ is
\begin{align}
H_{\min}&(B_{n}|A_{1},B_{1},...,A_{n-1},B_{n-1})_{D} := -\log \left[\max_{E} D\ast E\right] \label{eq:comb_min_entropy_max}
\end{align}

where the link product is over all systems related to $D$, i.e. $A_{1},B_{1}, ..., A_{n},B_{n}$, and where the maximisation is over probabilistic combs, i.e. over the set 
\begin{align}
\{E : \exists F \text{ s.t. } E \leq F\}, 
\end{align}

where $F \in \Comb(\mathbb{C} \!\rightarrow \! A_{1}, B_{1}\! \rightarrow \! A_{2}, ..., B_{n-1}\! \rightarrow \! A_{n} \otimes B_{n})\}$.
\end{Definition}

In other words, the maximisation is over all probabilistic combs that combine with $D$ to produce a probability value (the term inside the logarithm in \Cref{eq:comb_min_entropy_max} is readily identifiable as the generalised Born rule \cite{chiribella2009theoretical}). Equivalently, we can write the term inside the logarithm in \Cref{eq:comb_min_entropy_max} \cite{konig2009operational} as 
\begin{align}
(\dim B_{n}) \max_{\widehat{E}} \Tr\left[D \ast \widehat{E} \ket{\Phi^{+}}\!\!\bra{\Phi^{+}}_{B_{n},\widehat{B}_{n}} \right] \label{eq:maximum_pairs}
\end{align}

where the maximum is now over the set $\{\widehat{E} : \widehat{E} \leq F, F \in \Comb(\mathbb{C} \rightarrow A_{1}, B_{1} \rightarrow A_{2}, ..., B_{n-1} \rightarrow A_{n} \otimes \widehat{B}_{n})\}$ with $\H_{\widehat{B}_{n}} \cong \H_{B_{n}}$. The link product is over the systems $A_{1},B_{1},...,A_{n-1},B_{n-1}$ and $\ket{\Phi^{+}}$ is the maximally entangled state given by 
\begin{align}
\ket{\Phi^{+}}_{B_{n},\widehat{B}_{n}} := \frac{1}{\sqrt{\dim B_{n}}} \sum_{i} \ket{ii}_{B_{n},\widehat{B}_{n}}
\end{align}

where $\{\ket{i}_{B_{n}}\}$ ($\{\ket{i}\}_{\widehat{B}_{n}}$) is an orthonormal basis of $\H_{B_{n}}$ ($\H_{\widehat{B}_{n}}$).

Despite the equivalence, we find that the expression in \Cref{eq:maximum_pairs} contains greater conceptual clarity, especially for our later purposes. The comb in focus, $D$, represents the multi-step evolution of the system we are interested in with final output at $B_{n}$. The optimisation over probabilistic combs $\widehat{E}$ is an optimisation over all possible strategies for interacting with $D$. The distinction between optimising over $E$ in the original definition and over $\widehat{E}$ above, is that the latter highlights that we are maximising the amount of information about the output of $D$ that is stored in a \textit{separate} system $\widehat{B}_{n}$. This quantity can be expressed, as it is above, via the fidelity with a maximally entangled state on the two spaces, which represents the maximal possible correlations produced between the two spaces by a strategy $\widehat{E}$. For the conditional min-entropy of quantum states, this is termed the maximal singlet fraction \cite{konig2009operational}. 

This perspective of the min-entropy as optimising over strategies that maximise correlations is particularly appropriate for the classical-quantum combs that we consider for the remainder of this work. If $D$ is a classical-quantum comb, then the output space represents the random variable $X$ describing the family of combs $\{\sigma_{x}\}$ underlying $D$. The comb min-entropy for such a comb, denoted $H_{\min}(X|A_{1},B_{1},...,A_{n},B_{n})_{D}$, quantifies the minimum information that remains to be learnt about $X$, when the most informative strategy has been implemented (note, this strategy may be probabilistic and the min-entropy relates to the most informative \textit{branch} of the strategy in a single shot - this contrasts to the average case scenario quantified by the Shannon or von Neumann entropies). Equivalently, it is a measure of how distinguishable the $\sigma_{x}$ are. Since $X$ represents a random variable and not a physical system, the expression in \Cref{eq:maximum_pairs} is preferred since it highlights that the implemented strategy $\widehat{E}$ is aiming to maximise correlations with, and hence knowledge of, $X$. For classical-quantum combs, the term inside the logarithm in the min-entropy definition is called the guessing probability, since it represents the maximum certainty with which $X$ can be guessed (the terminology follows that for the min-entropy for classical-quantum states \cite{konig2009operational}). \Cref{fig:comb_min_ent} depicts the min-entropy for both states and combs from the perspective of \Cref{eq:maximum_pairs}.


\subsection{Results: Classical-Quantum Combs} \label{subsec:comb_min_ent_results}

In this subsection, we collect results for general classical-quantum and classical-classical combs that support the analysis of the specific combs related to blind quantum computing and measurement-based quantum computing in later sections.

In certain circumstances, it is natural to consider a multi-step evolution of a system which occurs in a series of independent rounds. For example, one could envisage an experimental scenario where a quantum system experiences non-Markovian noise during $n$-runs of the experiment, before being reset or recalibrated prior to a second set of $n$-runs. In later sections, we consider such round structure for the $X$-indexed combs of a classical-quantum comb and so establish the notation here. For $m$ rounds, we denote the classical-quantum comb as
\begin{align}
D^{(m)} := \sum_{x \in X}P(x)\ket{x}\!\!\bra{x} \otimes \bigotimes_{j = 1}^{m}\sigma_{x}^{(j)}  \label{eq:m_round_D_comb}
\end{align}

which is an operator in $\Comb\left(A_{1}^{(1)} \!\rightarrow\! B_{1}^{(1)}, ..., A_{n}^{(m)}\!\rightarrow\! B_{n}^{(m)}, \mathbb{C}\! \rightarrow\! X\right)$, with each $\sigma_{x}^{(j)}$ in $\Comb(A_{1}^{(j)} \!\rightarrow\! B_{1}^{(j)}, ..., A_{n}^{(j)}\!\rightarrow\! B_{n}^{(j)})$. Superscripts in brackets are reserved for indicating round number. For combs with such a structure, it is possible to show that the comb min-entropy never increases from round to round (see \Cref{sec:app_min_entropy} for an explicit proof):

\begin{Lemma} \label{lem:min_entropy_non_increasing} For $D^{(m)}$ and $D^{(l)}$ of the form given in \Cref{eq:m_round_D_comb}, where $m \geq l$, 
\begin{align*}
H_{\min}(X|&A_{1}^{(1)},B_{1}^{(1)}, ..., A_{n}^{(l)}, B_{n}^{(l)})_{D^{(l)}} \\
&\geq H_{\min}(X|A_{1}^{(1)},B_{1}^{(1)}, ..., A_{n}^{(m)},B_{n}^{(m)})_{D^{(m)}}.
\end{align*}
\end{Lemma}

Suppose now that we consider a classical-quantum comb where each $\sigma_{x}$ is diagonal in the same basis $\{\ket{\boldsymbol{a},\boldsymbol{b}}\}$ of $\H_{\boldsymbol{A},\boldsymbol{B}}$. That is, for each $x$ we can write
\begin{align}
\sigma_{x} = \sum_{\boldsymbol{a},\boldsymbol{b}}P(\boldsymbol{b}|x,\boldsymbol{a})\ket{\boldsymbol{a},\boldsymbol{b}}\!\!\bra{\boldsymbol{a},\boldsymbol{b}}_{\boldsymbol{A},\boldsymbol{B}}.
\end{align}

where $P(\boldsymbol{b}|x,\boldsymbol{a})$ satisfies the required marginalisation conditions (\Cref{eq:marginalisation_independence}) for each $x$. Accordingly, we can write the classical-quantum comb $D$ as
\begin{align}
D = \sum_{x, \boldsymbol{a},\boldsymbol{b}} P(x)P(\boldsymbol{b}|x, \boldsymbol{a})\ket{x,\boldsymbol{a},\boldsymbol{b}}\!\!\bra{x,\boldsymbol{a},\boldsymbol{b}}_{X,\boldsymbol{A},\boldsymbol{B}}.
\end{align}

As a consequence of \Cref{prop:X_A_independence}, we know that the variable $X$ is independent of all the inputs to $D$ (the $\boldsymbol{a}$). In conjunction with the conditional version of Bayes' rule, this allows us to rewrite $D$ as follows:
\begin{align}
D = \sum_{x, \boldsymbol{a}, \boldsymbol{b}} P(x|\boldsymbol{a},\boldsymbol{b})P(\boldsymbol{b}|\boldsymbol{a})\ket{x,\boldsymbol{a},\boldsymbol{b}}\!\!\bra{x,\boldsymbol{a},\boldsymbol{b}}_{X,\boldsymbol{A},\boldsymbol{B}}.
\end{align}

With $D$ in this form, we can establish the following bounds for its comb min-entropy.

\begin{Proposition} \label{prop:bayesian_classical_comb} Let $D$ be as above. Then
\begin{align}
-\log(\eta) \leq H_{\text{min}}(X|A_{1},B_{1},...,A_{n},B_{n})_{D} \leq -\log(\gamma)
\end{align}

where
\begin{align}
\eta &:= \sum_{\boldsymbol{b}} \max_{x, \boldsymbol{a}}P(x|\boldsymbol{a},\boldsymbol{b}) P(\boldsymbol{b}|\boldsymbol{a}), \\
\gamma &:= \frac{1}{|\boldsymbol{A}|}\sum_{\substack{\boldsymbol{a},\boldsymbol{b}}} \max_{x}P(x|\boldsymbol{a},\boldsymbol{b}) P(\boldsymbol{b}|\boldsymbol{a}).
\end{align}

with $|\boldsymbol{A}| := \prod_{i=1}^{n} \dim \H_{A_{i}}$.
\end{Proposition}

The proof is given in \Cref{sec:app_min_entropy}.

If the classical-quantum $D$ happens to be both mutli-round and with all $\sigma_{x}$ diagonal in the same basis, then 
\begin{align}
D = \sum_{\substack{x,\\ \boldsymbol{a}^{(1:m)},\\\boldsymbol{b}^{(1:m)}}} P(x|\boldsymbol{a}^{(1:m)},\boldsymbol{b}^{(1:m)}) \prod_{j = 1}^{m} P(\boldsymbol{b}^{(j)}|\boldsymbol{a}^{(j)}) \ket{x\boldsymbol{a}\boldsymbol{b}}\!\!\bra{x\boldsymbol{a}\boldsymbol{b}}
\end{align}

where $\boldsymbol{a}^{(1:m)}$ in the conditional probability is shorthand for $a^{(1)},...,a^{(m)}$ (similarly for $\boldsymbol{b}^{(1:m)}$; note the superscripts are dropped inside the bra and ket, as is the subscript indicating the respective Hilbert spaces), and the above result applies with $\eta$ and $\gamma$ replaced by
\begin{align}
\eta^{(m)} &:= \sum_{\boldsymbol{b}^{1:m}} \max_{x, \boldsymbol{a}^{(1:m)}}P(x|\boldsymbol{a}^{(1:m)},\boldsymbol{b}^{(1:m)}) \prod_{j = 1}^{m} P(\boldsymbol{b}^{(j)}|\boldsymbol{a}^{(j)}), \\
\gamma^{(m)} &:= \frac{1}{|\boldsymbol{A}|^{m}}\sum_{\substack{\boldsymbol{a}^{(1:m)},\\\boldsymbol{b}^{(1:m)}}} \max_{x}P(x|\boldsymbol{a}^{(1:m)},\boldsymbol{b}^{(1:m)}) \prod_{j = 1}^{m} P(\boldsymbol{b}^{(j)}|\boldsymbol{a}^{(j)}).
\end{align}


\section{Measurement-Based and Blind Quantum Computation} \label{sec:MBQC}

The previous sections introduced quantum combs and the comb min-entropy at a certain level of generality. In subsequent sections, we will consider specific classical-quantum combs representing certain quantum information processing tasks, and their min-entropy as a way of quantifying how well an unknown parameter can be learnt. This section introduces the required background for the chosen examples treated in this work, which relate to measurement-based quantum computation (MBQC) \cite{raussendorf2001one, briegel2009measurement, raussendorf2003measurement, raussendorf2001computational, jozsa2006introduction} and blind quantum computing (BQC) \cite{fitzsimons2017private,broadbent2009universal,mantri2017flow,morimae2013blind,morimae2014verification}. 


\subsection{Measurement-Based Quantum Computation} \label{subsec:MBQC}

\begin{figure*}
\begin{subfigure}{0.5\textwidth}
\centering
\includegraphics[width=0.9\textwidth]{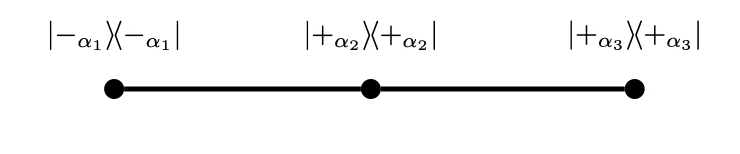}
\caption{}
\label{fig:graph_state_minus}
\end{subfigure}
\hfill
\begin{subfigure}{0.5\textwidth}
\centering
\includegraphics[width=0.9\textwidth]{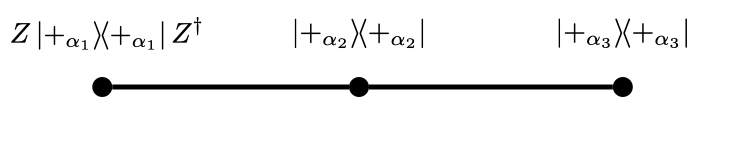}
\caption{}
\label{fig:graph_state_Z_plus_outcome}
\end{subfigure}
\\
\vspace{0.5cm}
\\
\begin{subfigure}{0.5\textwidth}
\centering
\includegraphics[width=0.9\textwidth]{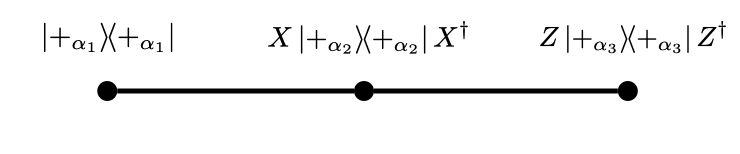}
\caption{}
\label{fig:graph_state_X_Zs}
\end{subfigure}
\hfill
\begin{subfigure}{0.5\textwidth}
\centering
\includegraphics[width=0.9\textwidth]{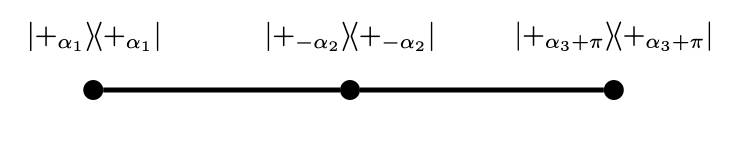}
\caption{}
\label{fig:graph_state_adjusted_angles}
\end{subfigure}
\caption{The properties of graph states allows for the correction of wrong measurement outcomes during measurement based computation. Figures (a)-(d) depict a three qubit graph state with measurements on all qubits in the $XY$-plane, and are all equivalent. Figure (a) depicts a negative measurement outcome obtaining on qubit $1$. For measurements in the $XY$-plane, this is related to a positive measurement outcome by $Z$ operators, shown in (b). Since a $Z$ on qubit $1$ is part of the stabiliser $K_{2} = Z_{1}X_{2}Z_{3}$, one can equivalently conjugate qubits $2$ and $3$ by $X$ and $Z$ respectively, as in (c). Instead of acting on the qubits, we can act on the planed measurements instead, which results in the change of measurement angle shown in (d).}
\label{fig:meas_corrections_explained}
\end{figure*}

Measurement-based quantum computation \cite{raussendorf2001one, briegel2009measurement, raussendorf2003measurement, raussendorf2001computational, jozsa2006introduction} is quantum computing paradigm in which computation is driven by single-qubit measurements on certain types of highly entangled states. Since measurements are inherently indeterministic, any deterministic computation in this framework requires an adaptive correction method. This correction method, termed gflow \cite{browne2007generalized, danos2006determinism}, plays a key role in the following sections, so much of this subsection is devoted to its explication.

MBQC is performed on entangled, multi-partite quantum states called graph states, which take their name from the connection to mathematical graphs. Let $G$ be a (simple, connected) graph on vertex set $V = \{1, ..., n\}$ and edge set $E$. One can define a graph state, denoted $\ket{G}$, as
\begin{align}
\ket{G} := \prod_{(i,j) \in E} CZ_{ij}\ket{+}^{\otimes n} \label{eq:graph_state_CZs}
\end{align}

where each vertex of the graph is assigned a qubit in the $\ket{+}$ state and where each edge $(i,j)$ in the graph is associated to a controlled Pauli-$Z$ gate, denoted $CZ_{ij}$, between the corresponding qubits. Equivalently, $\ket{G}$ can be specified as the unique stabiliser state \cite{gottesman1997stabilizer} of the set of stabilisers generated by
\begin{align}
\left\{ K_{v} := X_{v}  \bigotimes_{v' \in N_{v}^{G}} Z_{v'} \vert v \in V \right\}
\end{align}

where $N_{v}^{G}$ denotes the set of neighbours of $v$ in $G$. We will use the notation $\rho_{G}$ for $\ket{G}\!\!\bra{G}$.

A desired computation to be performed on $\ket{G}$ is specified by the projection operator corresponding to the positive outcomes of a list of single-qubit measurements. As measurements are indeterminate, the chances of obtaining only the desired outcomes are vanishingly small. However, in certain cases the symmetries of the graph state $\ket{G}$, namely the stabilisers $K_{v}$, can be leveraged to \textit{effectively} obtain the desired projection even when undesired measurement outcomes obtain. The catch is that this process only works for specific types of single-qubit measurements and requires adaptation of some measurements conditioned on the outcomes of others (introducing a time-ordering).

To understand how this works, consider the stabilisers $K_{v}$ and products thereof. Clearly, these are product operators consisting of single qubit unitaries $Z$, $X$ and their product $XZ$ (up to a sign). Now consider the states $\ket{\pm_{\alpha}}_{mp}$ where the subscript $mp$ stands for ``measurement plane'' which are the $XY$-, $XZ$- and $YZ$-planes of the Bloch sphere (i.e., $mp \in \{XY, XZ, YZ\}$) and where $\alpha$ indicates the angle of rotation from one of the axes of the corresponding plane ($\alpha \in [0, 2\pi)$). These states are defined as
\begin{align}
\ket{\pm_{\alpha}}_{XY} &:= \frac{1}{\sqrt{2}}(\ket{0} \pm e^{-i\alpha}\ket{1}), \\
\ket{\pm_{\alpha}}_{XZ} &:= \frac{1}{2}((1 \pm e^{-i\alpha})\ket{0} + i(1 \mp e^{-i\alpha})\ket{1}), \\
\ket{\pm_{\alpha}}_{YZ} &:= \frac{1}{2}((1 \pm e^{-i\alpha})\ket{0} + (1 \mp e^{-i\alpha})\ket{1}). \\
\end{align}

These states satisfy the following relations involving the single qubit unitaries identified above:
\begin{align}
\begin{split}
\ket{+_{\alpha}}_{XY} &= Z\ket{-_{\alpha}}_{XY},  \label{eq:meas_plane_syms}\\
\ket{+_{\alpha}}_{XZ} &= iXZ\ket{-_{\alpha}}_{XZ}, \\
\ket{+_{\alpha}}_{YZ} &= X\ket{-_{\alpha}}_{YZ}.
\end{split}
\end{align}

As a consequence, the projections $\ket{-_{\alpha}}\!\!\bra{-_{\alpha}}_{mp}$ for any measurement plane $mp$ is related to $\ket{+_{\alpha}}\!\!\bra{+_{\alpha}}_{mp}$ via conjugation by a ``piece'' of a stabiliser: a piece corresponding to a $Z$ operator for measurements in the $XY$-plane; a piece corresponding to the product of $X$ and $Z$, which arises from a product of stabilisers, for measurements in the $XZ$-plane (note the $i$ factor drops out in the conjugation); and a piece corresponding to an $X$ operator for measurements in the $YZ$-plane. The undesired outcome can thus be effectively corrected by ``completing'' the remainder of the stabiliser (provided one exists), as we now explain with an example. 

Consider a linear graph state on four qubits, such as if the three-qubit graph states depicted in \Cref{fig:meas_corrections_explained} were extended by one qubit to the right. Suppose that computational measurements are performed in the $XY$-plane on qubits $1$, $2$ and $3$ using measurement angles $\alpha_{1},\alpha_{2},\alpha_{3}$ respectively The desired output state on the final qubit is given by
\begin{align}
\rho_{\text{output}} = \Tr_{1,2,3}\left[\left(\bigotimes_{i = 1}^{3}\ket{+_{\alpha_{i}}}\!\!\bra{+_{\alpha_{i}}} \right)\rho_{G}\right].
\end{align}

where the $XY$ subscripts on the measurements have been dropped for simplicity. If instead, the first measurement produces a negative outcome, we have the following situation:
\begin{align}
\Tr_{1,2,3}&\left[\left(\ket{-_{\alpha_{1}}}\!\!\bra{-_{\alpha_{1}}} \otimes \bigotimes_{i=2}^{3} \ket{+_{\alpha_{i}}}\!\!\bra{+_{\alpha_{i}}} \right) \rho_{G}\right] \\
&= \Tr_{1,2,3}\left[\left(\ket{-_{\alpha_{1}}}\!\!\bra{-_{\alpha_{1}}} \otimes \bigotimes_{i=2}^{3} \ket{+_{\alpha_{i}}}\!\!\bra{+_{\alpha_{i}}} \right) K_{2} \rho_{G}K_{2}^{\dagger}\right] \\
&= \Tr_{1,2,3}\left[\left(\bigotimes_{i = 1}^{3}\ket{+_{\alpha_{i}}}\!\!\bra{+_{\alpha_{i}}}\right) (X_{2}Z_{3})\rho_{G}(X_{2}Z_{3})^{\dagger}\right].
\end{align}

We see that, via the stabiliser $K_{2} = Z_{1}X_{2}Z_{3}$, the desired outcome is recovered by applying $X$ and $Z$ unitaries on other qubits in the graph state. In fact, these unitaries can be absorbed into a change of angle for the corresponding qubits via the following relations:
\begin{align}
\begin{split} 
X^{\dagger}\ket{\pm_{\alpha}}\!\!\bra{\pm_{\alpha}}_{XY}X &\equiv \ket{\pm_{-\alpha \bmod 2\pi}}\!\!\bra{\pm_{-\alpha \bmod 2\pi}}_{XY}; \label{eq:XY_angle_adapts} \\
Z^{\dagger}\ket{\pm_{\alpha}}\!\!\bra{\pm_{\alpha}}_{XY}Z &\equiv \ket{\pm_{\alpha + \pi \bmod 2\pi}}\!\!\bra{\pm_{\alpha + \pi \bmod 2\pi}}_{XY}.
\end{split}
\end{align}

In order to make use of these relations in the example above, the measurements on qubits $2$ and $3$ must occur \textit{after} the outcome of measurement $1$ is known. The need to correct undesired outcomes and the above method for doing so induces a time-ordering in the computation. The above relations play a key role in the blind quantum computing protocol discussed below. Similar relations exist for measurements in the $XZ$- and $YZ$-planes.

This correction method relies on the existence of an appropriate stabiliser for each measurement and a compatible order of measurements so that all undesired outcomes can be accounted for. Such a set of stabilisers and order of measurements does not exist in all cases and it turns out that existence can be determined solely from the mathematical properties of the graph underlying the graph state. The literature on gflow \cite{browne2007generalized, danos2006determinism}, which stands for the `generalised flow' of corrections through the graph state, develops a complete methodology to specify when such a correction method exists and how it should proceed. Intuitively, gflow specifies \textit{how} a measurement outcome can be corrected (i.e. via which stabiliser), \textit{where} it can be corrected (i.e. which qubits receive an update), and \textit{when} a qubit should be measured (i.e. which other measurements must precede it). The formal definition is the following.

\begin{Definition}[\cite{browne2007generalized}] \label{def:gflow} Let $G = (V, E)$ be a graph, $I$ and $O$ be input and output subsets of $V$ respectively, and $\omega: O^{c} \rightarrow \{XY, XZ, YZ\}$ be a map assigning measurement planes to qubits (the superscript $c$ denotes set complement). The tuple $(G,I,O, \omega)$ has \textbf{gflow} if there exists a map $g: O^{c} \rightarrow \mathcal{P}(I^{c})$, where $\mathcal{P}$ denotes the powerset, and a partial order over $V$ such that the following hold for all $v \in O^{c}$:
\begin{enumerate}
	\item if $v' \in g(v)$ and $v' \neq v$, then $v < v'$;
	\item if $v' \in \Odd(g(v))$ and $v' \neq v$, then $v < v'$;
	\item if $\omega(v) = XY$, then $v \notin g(v)$ and $v \in \Odd(g(v))$;
	\item if $\omega(v) = XZ$, then $v \in g(v)$ and $v \in \Odd(g(v))$;
	\item if $\omega(v) = YZ$, then $v \in g(v)$ and $v \notin \Odd(g(v))$;
\end{enumerate}

where $\Odd(K) := \{\tilde{v} \in V : |N_{\tilde{v}}^{G} \cap K| = 1 \bmod 2 \}$ for any $K \subseteq V$. 
\end{Definition}

It is known that the existence of gflow is both necessary and sufficient for deterministic MBQC \cite[Theorems 2 and 3,][]{browne2007generalized}. Furthermore, polynomial time algorithms exist for determining whether a given tuple $(G,I,O, \omega)$ supports gflow \cite{mhalla2008finding, de2008finding}. Many distinct gflows for a given $(G,I,O, \omega)$ can exist, and characterising or counting all gflows for a given graph (as the input and output sets vary) in general remains an open problem.

We can understand this definition in light of the discussion preceding it as follows. The map $g$ assigns a stabiliser to each qubit being measured (the `\textit{how}' of the correction method): $g(v)$ is a subset of $V$ and identifies the stabiliser
\begin{align*}
K_{g(v)} := \prod_{w \in g(v)} K_{w}.
\end{align*}

Every element in the set $g(v)$ receives an $X$-correction from the above product and every element of $\Odd(g(v))$ receives a $Z$-correction; the vertices in their intersection receive both. Together $g(\cdot)$ and $\Odd(\cdot)$ specify the `\textit{where}' of the correction method. The partial order and the first two conditions enforce that every correction conditioned on the measurement at $v$ happens in the future of that measurement (the `\textit{when}' of the method). The remaining three conditions enforce that the component of the product that acts on the qubit $v$ is precisely the required symmetry associated to the measurement plane $\omega(v)$, as in \Cref{eq:meas_plane_syms}.

Derived from the definition of gflow, it is useful to define, for each $v \in V$, the set of vertices whose corrections induce an $X$-operation on $v$ and the set whose correction induce a $Z$-operation:
\begin{align}
\begin{split}
\mathcal{X}_{v} &:= \{v' \in V : v \in g(v')\setminus\{v\}\}; \label{eq:X_Z_correction_sets} \\
\mathcal{Z}_{v} &:= \{v' \in V: v \in \Odd(g(v'))\setminus\{v\} \}.
\end{split}
\end{align}

We allow for $\mathcal{X}_{v}$ or $\mathcal{Z}_{v}$ to be empty (such as when $v \in I$ for example).


\subsection{Blind Quantum Computation} \label{subsec:BQC}

Blind quantum computing (BQC) refers to an array of cryptographic quantum computational protocols (see e.g. \cite{broadbent2009universal,mantri2017flow,morimae2013blind,morimae2014verification} and the review \cite{fitzsimons2017private}), many of which build upon aspects derived from MBQC. In very broad terms, a BQC protocol consists of a client, who has limited computational power, interacting with a server (multi-party variants also exist), who has quantum computational capabilities, in order to carry out a desired computation in such a way that the latter is ``blind'' to the details. A range of theoretical tools for modelling quantum cryptographic protocols and analysing their security exist (see e.g. \cite{portmann2017causal,maurer2005abstract} and the review \cite{portmann2022security}). We provide an outline of some aspects of cryptographic security analysis based on the min-entropy at the start of the next section.

In this work, we consider a specific BQC protocol due to Mantri et al. \cite{mantri2017flow} and the remainder of this subsection is devoted to its description. This protocol is distinct from many other BQC proposals since it considers purely classical communication between client and server. At its core, the protocol leverages properties of gflow and MBQC outlined above, including the non-uniqueness of gflows for a given graph state and the ability to effectively implement a unitary by adapting measurement angles (recall \Cref{eq:XY_angle_adapts}). The protocol also uses standard cryptographic primitives such as one-time pads, which are a string of single-use bits drawn at random to obscure the true value of the secret data.

The protocol parameters consist of a choice of graph $G$, a total order on the vertices of $G$, and a discrete set of angles $\mathcal{A}$ which satisfies the following property:
\begin{align}
\mathcal{A} = \{(-1)^{x}\alpha + z\pi \bmod 2\pi: \alpha \in \mathcal{A}; x, z \in \mathbb{Z}_{2}\}. \label{eq:angle_set}
\end{align}

All measurements are made in the $XY$-plane (which is no restriction on the universality of the resulting computations - see \cite{mantri2017universality}) and so the above property enforces that $\mathcal{A}$ is closed under the angle transformations given in \Cref{eq:XY_angle_adapts}. The specific graph, total order and angle set are agreed upon by both the client and server prior to the commencement of a specific computation. In secret, the client also chooses their desired computation, that is, a list of measurement angles $\boldsymbol{\alpha} \in \mathcal{A}^{n}$ (where $n$ is the number of vertices in $G$) and designated input and output sets $I$ and $O$, a bit-string one-time pad $\boldsymbol{r} \in \mathbb{Z}_{2}^{n}$ uniformly at random, and a gflow compatible with $(G, I, O)$ and the total order.

One round of computation proceeds as follows:
\begin{enumerate}
	\item The server initialises the graph state $\rho_{G}$.
	\item For $i = 1, ..., n$ according to the total order, the following sequence is repeated:
	\begin{enumerate}
		\item The user reports a measurement angle $\alpha'_{i}$ to the server, where
		\begin{align}
		\alpha'_{i} := (-1)^{\bigoplus_{j \in \mathcal{X}_{i}} c_{j}}\alpha_{i} + \left(r_{i} \oplus \bigoplus_{j \in \mathcal{Z}_{i}} c_{j}\right)\pi \bmod 2\pi \label{eq:adapt_meas_angles}
		\end{align}

		with $c_{j}$ denoting the measurement outcome for qubit $j < i$ recorded by the user based on the outcome reported by the server and with $\mathcal{X}_{i}$ and $\mathcal{Z}_{i}$ denoting the corrections sets defined by the gflow (recall \Cref{eq:X_Z_correction_sets}).
		\item The server measures $\mathcal{M}_{\alpha'_{i}}$ and reports $c'_{i} = 0$ for a positive outcome and $c'_{i} = 1$ for a negative outcome.
		\item The user records $c_{i} = c'_{i} \oplus r_{i}$.
	\end{enumerate}
	\item The outcomes pertaining to the output qubits (which are known only to the user) are processed to obtain the results of the computation.
\end{enumerate}

There are two things worth noting about the use of the bit-strings $\boldsymbol{r}$. Firstly, their utility for obscuring angles is a direct consequence of the $Z$-relation given in \Cref{eq:XY_angle_adapts} (computationally, a positive measurement outcome for angle $\alpha$ and negative measurement outcome for $\alpha + \pi$ are equivalent). Secondly, the presence of $r_{i}$ in both the equation for $\alpha'_{i}$ as well as for any $\alpha'_{k}$ for which $i \in \mathcal{X}_{k}$ or $i \in \mathcal{Z}_{k}$ places constraints on the set of possible reported angles $\boldsymbol{\alpha'}$ for a given choice of true angles $\boldsymbol{\alpha}$. This latter fact indicates that some amount of information is leaked during the protocol.

The above protocol is known to be correct \cite[Theorem 1,][]{mantri2017flow}: if both the client and server behave accordingly, the correct computation is performed. However, it also known that this protocol is not verifiable: the client has no way of knowing whether the server actually prepares a graph state, measures according to the reported angles, and communicates the actual measurement outcomes. It is thus the security of the protocol that is of key concern, which we take up in the coming section.

To conclude, we remark on some minor differences between the analysis in \cite{mantri2017flow} and the one below. Mantri et al. place a further condition on the definition of gflow, largely for the purpose of simplifying a counting argument (see \cite[Theorem 3,][]{mantri2017flow}), which effectively singles out one gflow for every pair $(I,O)$ given $G$ and they thus identify a choice of computation with a choice of angles and choice of gflow. Here, we work with the unconstrained definition of gflow and so identify a choice of computation as a choice of $\boldsymbol{\alpha}$ and $(I,O)$ (again for fixed $G$), for which any of the compatible gflows can be chosen. Moreover, since the protocol is entirely classical, the client can only prepare an input state on the qubits in $I$ via measurement-based state preparation, which is thus indistinguishable from the part of the measurement sequence implementing the unitary. As such, computations are specified here purely by angles $\boldsymbol{\alpha}$ and a choice of output set $O$ for which $(G,I,O)$ supports gflow for some $I$.


\section{Analysing Classically Driven Blind Quantum Computation} \label{sec:CBQC}

In quantum cryptography, it is typical to analyse the security of a protocol in terms of information-theoretic quantities (see \cite{portmann2022security} for an accessible review of security in quantum cryptography). In this section, we provide a security analysis of the blind quantum computing protocol outlined above via the use of the comb min-entropy. This necessitates representing the protocol as a comb, which we establish in the next subsection. Prior to doing so, we outline the use of the conditional min-entropy for cryptographic protocols involving quantum states as well as review the existing security analysis of the protocol given in \cite{mantri2017flow} which is based on the Shannon entropy. We also briefly comment on some of the benefits of using the min-entropy as opposed to the Shannon entropy for security analysis.

Cryptographic protocols are usually analysed for both their correctness and their security. The two terms refer to how well the protocol achieves its intended aim in the absence and presence of malicious intent respectively. For example, in quantum key distribution (QKD) protocols \cite{bennett2014quantum,ekert_91}, where the aim is to establish a secure key between two parties, correctness refers to whether both parties receive the same key in the absence of an eavesdropper, and security refers to a measure of how much information an eavesdropper can learn about the key by intervening on the communication between the parties. Since the BQC protocol under consideration is known to be correct, we focus on security in the following.

The most commonly used contemporary notion of security is based on distinguishability, that is, the ability to discriminate between the ``true'' value of the secret information and the other possible values. The distinguishability of two quantum states is related to their trace distance \cite{nielsen_chuang_2010,watrous2018theory} and accordingly this distance is often used in statements of security \cite{renner2008security,ben2005universal,renner2005universally}. For more than two states, distinguishability is given by the conditional state min-entropy \cite{konig2009operational,holevo1973statistical,yuen75,watrous2018theory}, which provides some motivation for our use of the comb min-entropy below.

To provide a concrete example, consider a QKD scenario. Let $X$ denote the variable representing the bit-strings received by one party and $E$ represent the information obtained by the eavesdropper (which could be quantum in nature). The security of $X$ in the presence of an eavesdropper with access to $E$ is given as a lower bound to $H_{\min}(X|E)$. In this setting, $H_{\min}(X|E)$ can essentially be thought of as the minimum number of bits (deterministically) obtainable from $X$ that are uncorrelated to $E$: the greater this number, the less information the eavesdropper has about $X$. Note that in many proofs of security, the smooth version of $H_{\min}(X|E)$ is used instead to allow for some degree of error in the protocol (see e.g., \cite{renner2008security} for details).

The security (blindness) analysis of the BQC protocol provided by Mantri et al., \cite{mantri2017flow} also establishes bounds on the amount of information shared between the secret information belonging to the client and the information communicated to the server. These bounds, provided below, are based on mutual information and the conditional Shannon entropy rather than the conditional min-entropy. Prior to stating the result, we establish some notation based on the description of the protocol in \Cref{subsec:BQC}. Let $\boldsymbol{A}$ denote the random variable for the angles, which takes values in $\mathcal{A}^{n}$ (recall \Cref{eq:angle_set}), and let $\boldsymbol{F}$ denote the random variable which takes values in the set of restricted gflows (recall the comments made at the end of the previous section). The random variables for the reported angles and measurement outcomes are denoted $\boldsymbol{A'}$ and $\boldsymbol{C'}$ respectively. With this notation, the blindness theorem is:

\begin{Theorem}[Theorem 2, \cite{mantri2017flow}] \label{thm:Mantri} In a single instance of the protocol, the mutual information between the client's secret input $\{\boldsymbol{\alpha}, \boldsymbol{f}\}$ and the information received by the server is bounded by
\begin{align}
I(\boldsymbol{C'}, \boldsymbol{A'}; \boldsymbol{A}, \boldsymbol{F}) \leq H(\boldsymbol{A'})
\end{align}

where $I(\cdot;\cdot)$ denotes the mutual information and $H(\cdot)$ denotes the Shannon entropy.
\end{Theorem}

Using part of the proof of the above theorem (\cite[Lemma 4,][]{mantri2017flow}) and the properties of the mutual information and Shannon entropy, it is possible to lower bound the conditional Shannon entropy of the secret information $\boldsymbol{A}, \boldsymbol{F}$ given access to the shared information $\boldsymbol{C'}, \boldsymbol{A'}$:
\begin{align}
H(\boldsymbol{A}, \boldsymbol{F}| \boldsymbol{C'}, \boldsymbol{A'}) \geq \log_{2} |\boldsymbol{F}| > 0. \label{eq:Shannon_entropy_inequality}
\end{align}

There are a number of reasons why statements involving the min-entropy provide stronger statements of security than those using the Shannon entropy. As alluded to above, the min-entropy is more commonly used to analyse protocols such as QKD. This is partly due to its better behaviour when making statements about the security of protocols composed of a number of parts; there are known examples where Shannon entropic quantities indicate insecurity of a protocol built up from perfectly secure ones (see e.g., \cite[III.C.1,][]{portmann2022security} and references therein). Furthermore, the conditional min-entropy is always a lower bound for the corresponding conditional Shannon entropy, so bounding the former above some non-zero value also bounds the latter (the converse does not hold in general). Finally, Shannon entropies are effectively asymptotic quantities whereas the min-entropy is a truly single-shot quantity and thus better suited to statements regarding the security of a single instance of the BQC protocol considered here. See e.g., \cite{konig2009operational} and \cite{tomamichel2015quantum} for further information on the distinctions between the min-entropy and Shannon entropy.

To the best of our knowledge, the security analysis presented here is the first time the comb min-entropy has been used to analyse a BQC protocol.


\subsection{The Classical BQC Protocol as a Comb} \label{subsec:D_client}

In order to be able to apply the comb min-entropy to the protocol under investigation, we need to write down the corresponding comb. We begin by fixing a graph $G$ on $n$ vertices and a set of allowed angles $\mathcal{A}$ satisfying \Cref{eq:angle_set}. Let $1 < 2 < ... < n$ denote a choice of total order on the vertices. Since all communication in the protocol is classical, we represent the reported angles and measurement outcomes as basis states in corresponding Hilbert spaces: $\ket{\alpha'_{i}}\!\!\bra{\alpha'_{i}} \in \H_{A'_{i}}$ for the reported angle at step $i$ and $\ket{c'_{i}}\!\!\bra{c'_{i}} \in \H_{C'_{i}}$ for the reported measurement outcome, where $\dim \H_{A'_{i}} = |\mathcal{A}|$ and $\dim \H_{C'_{i}} = 2$.

Prior to the start of the protocol, the client secretly selects their desired computation. This consists of a choice of output vertices $O$ in $G$ and the ``true'' angles for the computation $\boldsymbol{\alpha} \in \mathcal{A}^{n}$. Since $G$ and $\mathcal{A}$ are fixed and discrete, so are $\mathcal{O}$ and $\mathcal{A}^{n}$, the set of all possible choices of output sets and all possible angle sequences respectively. We denote by $\boldsymbol{O}$ and $\boldsymbol{A}$ the random variables which jointly represent the choice of computation, which together make up the classical variable part of the classical-quantum combs describing the protocol (i.e. the variable labelled $X$ in the general treatment of classical-quantum combs). The prior knowledge that the server has about the choice of computation is denoted by $P(\boldsymbol{\alpha},O)$ and the full classical-quantum comb representing the choice of computation and ensuing protocol is denoted
\begin{align}
D_{\client} = \sum_{\boldsymbol{\alpha}, O} P(\boldsymbol{\alpha}, O)\ket{\boldsymbol{\alpha}, O}\!\!\bra{\boldsymbol{\alpha}, O} \otimes \sigma_{\boldsymbol{\alpha}, O}.
\end{align}

Each $\sigma_{\boldsymbol{\alpha}, O}$ represents the protocol for the computation $\boldsymbol{\alpha}, O$ and the remainder of this subsection constructs these combs from the details of the protocol. For later reference, note that each $\sigma_{\boldsymbol{\alpha}, O} \in \Comb(\mathbb{C} \rightarrow A'_{1}, C'_{1} \rightarrow A'_{2}, ..., C'_{n} \rightarrow \mathbb{C})$ and $D_{\client} \in \Comb(\mathbb{C} \rightarrow A'_{1}, C'_{1} \rightarrow A'_{2}, ..., C'_{n} \rightarrow \mathbb{C}, \mathbb{C} \rightarrow \boldsymbol{A} \times \boldsymbol{O})$. We use the notation $\H_{\boldsymbol{A'},\boldsymbol{C'}} := \bigotimes_{i=1}^{n} \H_{A'_{i}} \otimes \H_{C'_{i}}$ for the composite Hilbert space related to the $\sigma_{\boldsymbol{\alpha}, O}$.

Along with $\boldsymbol{\alpha}$ and $O$, the client secretly chooses a bit-string $\boldsymbol{r} \in \mathbb{Z}_{2}^{n}$ and gflow $g$. Despite being unknown to the server, this information does not feature in the classical variable in $D_{\client}$ since we are only concerned with how well the protocol protects information regarding the computation. The gflow $g$ must be compatible with the graph, total order and choice of output set, all of which we denote by the shorthand $g \sim O$. The comb corresponding to a fixed choice of $\boldsymbol{\alpha}, \boldsymbol{r}$ and $g$ is
\begin{align}
\sigma_{\BQC}^{\boldsymbol{\alpha}, \boldsymbol{r}, g} := \sum_{\substack{\boldsymbol{\alpha'} \in \mathcal{A}^{n}, \\ \boldsymbol{c'} \in \mathbb{Z}_{2}^{n}}} P(\boldsymbol{\alpha'}| \boldsymbol{c'}, \boldsymbol{\alpha}, \boldsymbol{r}, g)\ket{\boldsymbol{\alpha'}\boldsymbol{c'}}\!\!\bra{\boldsymbol{\alpha'}\boldsymbol{c'}}_{\boldsymbol{A'},\boldsymbol{C'}}
\end{align}

and is an element of the same comb set as stipulated for $\sigma_{\boldsymbol{\alpha}, O}$ above. The distribution $P(\boldsymbol{\alpha'}| \boldsymbol{c'}, \boldsymbol{\alpha}, \boldsymbol{r}, g)$ is deterministic and simply encodes the angle adaptations:
\begin{align}
P(\boldsymbol{\alpha'}| \boldsymbol{c'}, \boldsymbol{\alpha}, \boldsymbol{r}, g) = \begin{cases} 1, \text{ if \Cref{eq:adapt_meas_angles} holds for all $i$} \\ 0, \text{ otherwise } \end{cases}.
\end{align}

Note that \Cref{eq:adapt_meas_angles} contains the notation $c_{i}$, but since $c_{i} := c'_{i} \oplus r_{i}$, we use only the $\boldsymbol{c'}$ and $\boldsymbol{r}$ notation instead.

Since any choice of $g \sim O$ and $\boldsymbol{r}$ produce the desired result once $\boldsymbol{\alpha}$ and $O$ are fixed, we define $\sigma_{\boldsymbol{\alpha}, O}$ by averaging over all possible choices of compatible gflow $g$ and bit-string $\boldsymbol{r}$:
\begin{align}
\sigma_{\boldsymbol{\alpha}, O} := \sum_{g \sim O, \boldsymbol{r}} P(g | O)P(\boldsymbol{r}) \sigma_{\BQC}^{\boldsymbol{\alpha}, \boldsymbol{r}, g} \label{eq:BQC_comb_not_all_comps}
\end{align}

The distribution $P(g|O)$ and $P(\boldsymbol{r})$ are the probabilities of choosing $g$ and $\boldsymbol{r}$ respectively. According to the protocol, the latter distribution is taken to be uniform, $P(\boldsymbol{r}) = \frac{1}{2^{n}}$, and independent of all other variables (see \cite[Lemma 4,][]{mantri2017flow}). Furthermore, since the specific choice of gflow is irrelevant for the computation once $O$ is fixed, we are also justified in taking $P(g|O)$ to be uniform and independent of $\boldsymbol{\alpha}$. It is in fact possible to write $\sigma_{\boldsymbol{\alpha}, O}$ in a more compact form as
\begin{align}
\sigma_{\boldsymbol{\alpha},O} = \sum_{\boldsymbol{\alpha'}, \boldsymbol{c'}}  P(\boldsymbol{\alpha'}| \boldsymbol{c'}, \boldsymbol{\alpha}, O) \ket{\boldsymbol{\alpha'}\boldsymbol{c'}}\!\!\bra{\boldsymbol{\alpha'}\boldsymbol{c'}}_{\boldsymbol{A'},\boldsymbol{C'}}.
\end{align}

The details are given in \Cref{subsec:app_BQC_single_round_proof}.


\subsection{Results: Security Analysis} \label{subsec:BQC_results}

This subsection presents the results of an improved security analysis of the BQC protocol first presented in \cite{mantri2017flow} via use of the comb min-entropy. This aligns more closely with typical quantum cryptographic practices, as outlined at the start of this section. Specifically, we establish a non-zero lower bound for the comb min-entropy for $D_{\client}$ defined above for a single round of the protocol and provide an example that obtains this bound. Next, we establish a non-zero lower bound for any (finite) number of rounds of the protocol, discuss its relation to the security of the protocol and provide an example where the min-entropy strictly decreases between two rounds of the protocol indicating some leakage of information.


\subsubsection{Blindness in a Single Round} \label{subsec:single_round_thm}

Here we present the analogous result to \Cref{thm:Mantri} using the comb min-entropy. This result assumes that the client selects the computation uniformly at random (which is also assumed in \cite{mantri2017flow}).

\begin{Theorem} \label{thm:single_shot_entropy} For $D_{\client}$ as above, and under the given assumptions, the following holds for any choice of graph $G$ and angle set $\mathcal{A}$ that satisfy the required conditions:
\begin{align}
H_{\min}(\boldsymbol{A}, \boldsymbol{O}| \boldsymbol{A'}, \boldsymbol{C'})_{D_{\client}} \geq n + \log_{2}(|\mathcal{O}|). \label{eq:thm_single_round}
\end{align}
\end{Theorem}

The proof of this theorem is given in \Cref{subsec:app_BQC_single_round_proof}. The proof makes use of \Cref{prop:bayesian_classical_comb} applied to the comb $D_{\client}$. Due to the specific form of $D_{\client}$, this bound has an interpretation in terms of the number of different ways each $\boldsymbol{\alpha'}$ can be reported for a given classical message (that is, for different gflows $g$ and bit-strings $\boldsymbol{r}$) with a larger number corresponding to a lower min-entropy bound (this is explained further after the proof of the theorem in \Cref{subsec:app_BQC_single_round_proof}). 

One can rightfully ask how good this bound is in practice. Consider the following example for which the gflows are characterised and for which the min-entropy obtains the bound in \Cref{eq:thm_single_round}. Let $G$ be the three-vertex graph as shown in \Cref{fig:min_BQC_example} with total order given by the natural order on vertex labels. As demonstrated in \Cref{subsec:app_BQC_min_example}, there is exactly one non-trivial output set ($O = \{2,3\}$) that supports gflow for this example, and exactly two corresponding gflows exist. It can further be shown that the upper and lower bounds of \Cref{prop:bayesian_classical_comb} for $D_{\client}$ coincide, which establishes $n + \log_{2}(|\mathcal{O}|)$ as the true min-entropy value. For this example, $n = 3$ and $|\mathcal{O}| = 1$, so $H_{\min}(\boldsymbol{A}, \boldsymbol{O}| \boldsymbol{A'}, \boldsymbol{C'})_{D_{\client}} = 3$, which is corroborated by the numerical results (for two choices of $\mathcal{A}$, one such that $|\mathcal{A}| = 4$ and the other such that $|\mathcal{A}| = 8$).

\begin{figure}
\includegraphics[width=4cm]{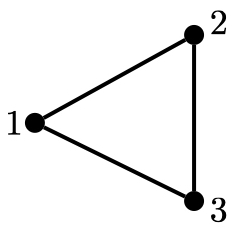}
\centering
\caption{An example graph state for which the security for a single round of the BQC protocol considered here is maximally bad (that is, the min-entropy for the associated quantum comb obtains the lower bound given in \Cref{thm:single_shot_entropy}). All potential gflows for the protocol are required to be compatible with the order $1<2<3$ on vertices. The only output set that supports such gflows is $O = \{2,3\}$.}
\label{fig:min_BQC_example}
\end{figure}


\subsubsection{Blindness in Multiple Rounds} \label{subsec:multi_round_blindness}

Many quantum algorithms produce the correct output with high, but non-unit, probability. In some cases, it may be necessary to collect measurement statistics of the outputs of the algorithm, for example when estimating eigenvalues to finite precision in quantum phase estimation \cite{kitaev1995quantum,nielsen_chuang_2010}. If such an algorithm was to be performed in a BQC context, many runs of the protocol would be required. We now extend the blindness analysis to such a situation, where the client and server engage in multiple rounds of the protocol. Prior to this work, no knowledge of the security of the protocol under multiple rounds was known as such a case was not treated in the original work \cite{mantri2017flow}. As an extension of \Cref{thm:single_shot_entropy} above, we demonstrate that the min-entropy is bounded away from zero for all rounds.

We use similar assumptions to \Cref{thm:single_shot_entropy}, namely that $P(g|O)$ and $P(\boldsymbol{r})$ are uniform, however we make a slightly weaker assumption on $P(\boldsymbol{\alpha}, O)$: we take $P(\boldsymbol{\alpha}, O) = \frac{P(O)}{|\mathcal{A}|^{n}}$ for $P(O)$ an arbitrary distribution over the output sets. Furthermore, we make the additional assumptions that the gflow and one-time pads are chosen independently for each round, allowing the multi-round protocol to be represented as
\begin{align}
D_{\client}^{(m)} = \sum_{\boldsymbol{\alpha}, O}\frac{P(O)}{|\mathcal{A}|^{n}}\ket{\boldsymbol{\alpha}, O}\!\!\bra{\boldsymbol{\alpha},O} \bigotimes_{j = 1}^{m} \sigma_{\boldsymbol{\alpha}, O}^{(j)}.
\end{align}

where each $\sigma_{\boldsymbol{\alpha}, O}^{(j)} \in \Comb(\mathbb{C} \rightarrow (A'_{1})^{(j)}, (C'_{1})^{(j)} \rightarrow (A'_{2})^{(j)}, ..., (C'_{n})^{(j)} \rightarrow \mathbb{C})$. We have the following theorem:

\begin{Theorem} \label{thm:multi_round} For any $m$, the following holds:
\begin{align}
H_{\min}(\boldsymbol{A}, \boldsymbol{O}| \boldsymbol{A'}^{(1)}, \boldsymbol{C'}^{(1)}, ..., &\boldsymbol{A'}^{(m)}, \boldsymbol{C'}^{(m)})_{D_{\client}^{(m)}} \nonumber \\ 
&\geq -\log\left(\sum_{O \in \mathcal{O}} \frac{P(O)}{2^{|O|}} \right).
\end{align}
\end{Theorem}

A key point of the proof, given in \Cref{subsec:app_multi_round_blind}, consists of the observation that, for any choice of output set $O$, the bit-string values assigned to the qubits in the output set cannot be learnt by the server since they do not appear in the adaptation of any other angle (in contrast to the bits of $\boldsymbol{r}$ assigned to non-output qubits). The consequence is that the server can, at best, only learn the output angles of $\boldsymbol{\alpha}$ up to a $\pi$ phase. However, the server may be able to learn all other measurement angles with certainty, which correspond to the computation proper, in which case the only remaining uncertainty consists of not knowing how to interpret the measurement outcomes on the output state. This constitutes a large amount of leakage in the protocol, which would likely be deemed unsatisfactory in practice without further modifications to the protocol.

The above theorem does not say anything about whether the min-entropy does in fact change from round to round, but it does in fact do so for the minimal example presented above. As outlined in \Cref{subsec:app_BQC_min_example}, for a specific choice of angle set $\mathcal{A}$, the two round entropy is upper bounded by $-\log_{2}(0.140625) \approx 2.830(0)$ and thus the min-entropy value does indeed decrease from the single round to the two round case.


\section{Grey Box MBQC} \label{sec:grey_MBQC}

In the previous section, the comb min-entropy was used to quantify the distinguishability of different quantum computations in a cryptographic setting and make statements about security. In this section, the comb min-entropy is used to quantify how well an unknown property or parameter of a quantum system can be learnt. This aligns well with the use of the min-entropy in areas such as quantum hypothesis testing and quantum metrology, which have recently seen application of the combs formalism \cite{bavaresco2023designing}. 

Our system of interest throughout this section will be a ``grey box MBQC device''. This device is to be understood as a sealed measurement-based quantum computing device about whose internal functioning we only have partial knowledge. Specifically, the device prepares a specific graph state and performs all measurement corrections for us, all we need to do is receive qubits one-by-one and measure them. Throughout, we assume the graph structure of the graph state prepared inside the device is known, but the details of the exact state and correction method are not.

The scenarios we consider below are motivated by practical considerations and can be understood as components of a verification process for the device. The first example is related to establishing the correct functioning of the device. For certain graph states, there exist a number of correction methods (many gflows, cf \Cref{subsec:MBQC}) which are only compatible with certain types of single-qubit measurements (i.e., in a specific plane of the Bloch sphere). If a different measurement is used, the output of the device is no longer the correct output of the computation. Accordingly, in order to use the device properly, the correct allowed measurement planes must be known. \Cref{subsec:meas_planes} investigates how well this can be achieved for a specific example.

Interestingly, even if all measurement planes for the device are known, which restricts the set of possible gflows that the device could be implementing, there is no guarantee that any further information can be learnt about which specific gflow it is. In \Cref{subsec:noise}, we demonstrate an example where it is in fact impossible to learn anything more about which correction method is being used once the measurement planes are known. The situation changes, however, if the graph state preparation is imperfect. Using a realistic noise model, we demonstrate that, even though clearly harmful for any computation, the presence of noise is actually beneficial for learning other features of the grey box device, such as the specific details of the correction method.

If all measurement planes are known and there is no noise in the device, we still must be sure that our measurement angles are compatible with the basis in which the graph state is being prepared in order to obtain useful computational results. That is, we need to calibrate our single qubit measurements to the device. Stated another way, we need to ensure that the internal quantum reference frame \cite{bartlett2007reference} of the device aligns with the external one of the user. An example investigating how well this can be done to a given precision is treated in \Cref{subsec:meas_calibration}.

The final subsection introduces a connection between MBQC and quantum causal models \cite{barrett2019quantum,allen2017quantum,costa2016quantum,oreshkov2012quantum, d2018causality} that may be of interest to the quantum foundations community. Quantum causal models can be considered to be a refinement of quantum combs to include more structure: while combs stipulate a linear sequence of input and output spaces, quantum causal models further stipulate which input spaces can influence which other output spaces and which do not. In \Cref{subsec:caus_inf}, we establish for the first time the connection between MBQC and quantum causal models and discuss some consequences of this connection. In keeping with the rest of this work, we interpret the comb min-entropy in this context as a type of causal hypothesis testing and investigate some examples again using the grey box MBQC device.


\subsection{Gflow Quantum Combs} \label{subsec:gflow_qcomb} 

\begin{figure*}
\begin{subfigure}{0.5\textwidth}
\centering
\includegraphics[width=0.62\textwidth]{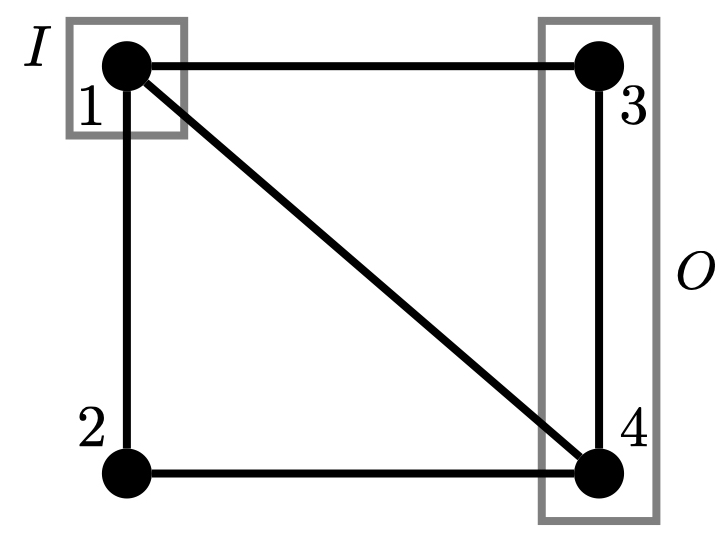}
\caption{$(G, I, O)$}
\label{fig:grey_box_graph}
\end{subfigure}
\hfill
\begin{subfigure}{0.5\textwidth}
\centering
\includegraphics[width=0.9\textwidth]{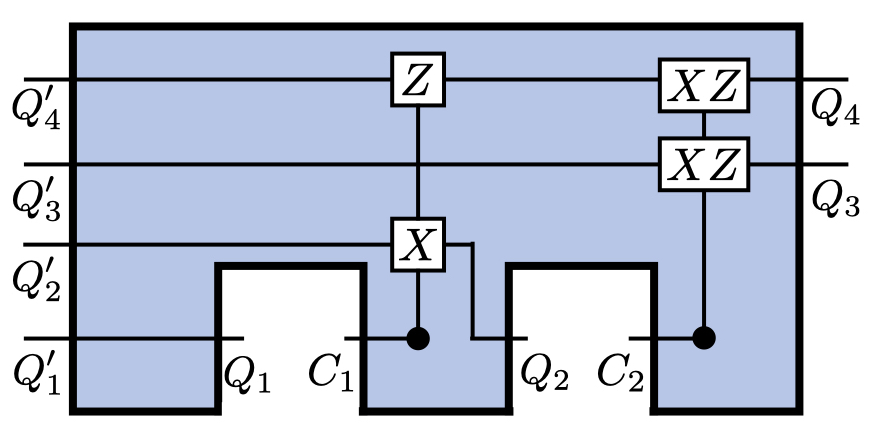}
\caption{$\sigma_{\MBQC}^{g_{1}}$}
\label{fig:sigma_MBQC_g}
\end{subfigure}
\caption{The classical-quantum combs defined here for the consideration of an MBQC device have quantum comb component based on the flow of corrections required for measurement-based computation (known as gflow). (a) In most of the examples, we consider a four qubit graph $G$ and choice of input set $I = \{1\}$ and output set $O = \{3,4\}$. There are $15$ gflows compatible with this choice of $(G,I,O)$ which are catalogued in \Cref{subsec:app_gflow_catalogue}, including the corresponding directed acyclic graphs and correction operators. The quantum comb corresponding to one gflow, $g_{1}$ (see the catalogue for labelling), for this example is shown in (b). The graph state is input on the left and passed through the conditional correction channels to each of the output spaces on the right. Measurement channels performing single qubit measurements can be inserted into the gaps (such as between $Q_{1}$ and $C_{1}$). The graph state preparation and corrections are all internal to the grey box device, so the user can only interact via the space $Q_{i}$ and $C_{j}$.}
\label{fig:grey_box_combs}
\end{figure*}

Similarly to the treatment of the BQC protocol, we begin by constructing the quantum comb from which the classical-quantum combs for the upcoming examples can be constructed. As stated above, the grey box device contains both a graph state preparation and the correction method corresponding to a gflow. The graph state will be denoted $\rho_{G}$ and the operator implementing the corrections for gflow $g$ will be denoted $\sigma_{\MBQC}^{g}$. In the coming subsections, the classical-quantum combs for each example will be constructed from $\rho_{G}$ and $\sigma_{\MBQC}^{g}$ depending on the specific parameters to be learnt.

Let us fix a graph $G$ on $n$ vertices, input and output sets $I$ and $O$, and measurement planes for each qubit via the map $\omega$ (recall the definition of gflow \Cref{def:gflow}). Let $g$ denote a gflow compatible with $(G,I,O,\omega)$. The comb $\sigma_{\MBQC}^{g}$ will be constructed out of conditional single-qubit unitaries (the ``completion'' of the stabiliser - \Cref{subsec:MBQC}) as specified by $g$. Recall from \Cref{eq:X_Z_correction_sets} that $g$ defines the correction sets $\mathcal{X}_{i}$ and $\mathcal{Z}_{i}$ for each $i \in V$. We then define for each $i \in V$ 
\begin{align}
U_{\text{corr}(\boldsymbol{c}),i} := X_{Q_{i}}^{\bigoplus_{j \in \mathcal{X}_{i}} c_{j}}Z_{Q_{i}}^{\bigoplus_{j \in \mathcal{Z}_{i}} c_{j}}. \label{eq:U_corr_i}
\end{align}

where the subscript $Q_{i}$ identifies the Hilbert space for the qubit corresponding to vertex $i$. For notational simplicity, we also define
\begin{align}
U_{\text{corr}(\boldsymbol{c})} := \bigotimes_{i = 1}^{n}U_{\text{corr}(\boldsymbol{c}),i}. \label{eq:U_corr}
\end{align}

We can then define $\sigma_{\MBQC}^{g}$ as
\begin{align}
\sigma_{\MBQC}^{g} := \sum_{\boldsymbol{a}, \boldsymbol{b}, \boldsymbol{c}} U_{\text{corr}(\boldsymbol{c})} \ket{\boldsymbol{a}}\!\!\bra{\boldsymbol{b}}_{\boldsymbol{Q}}U_{\text{corr}(\boldsymbol{c})}^{\dagger} \otimes \ket{\boldsymbol{c},\boldsymbol{a}}\!\!\bra{\boldsymbol{c},\boldsymbol{b}}_{\boldsymbol{CQ'}}  \label{eq:sigma_MBQC}
\end{align} 

where the subscripts $\boldsymbol{Q}$, $\boldsymbol{C}$ and $\boldsymbol{Q'}$ denote that the corresponding state is in $\H_{\boldsymbol{Q}} : = \bigotimes_{i = 1}^{n}\H_{Q_{i}}$, $\H_{\boldsymbol{C}} := \bigotimes_{i = 1}^{n}\H_{C_{i}}$ or $\H_{\boldsymbol{Q'}} := \bigotimes_{i = 1}^{n}\H_{Q'_{i}}$ respectively, and the sum is over the computational basis of these spaces. Note that $\sigma_{\MBQC}^{g}$ is an operator on $\bigotimes_{i = 1}^{n} \H_{Q_{i}} \otimes \H_{C_{i}} \otimes \H_{Q'_{i}}$, with $\H_{Q_{i}} \cong \H_{C_{i}} \cong \H_{Q'_{i}} \cong \mathbb{C}^{2}$ for each $i$. Moreover, we endow $\sigma_{\MBQC}^{g}$ with a choice of total order for which the partial order given by the gflow is compatible. Any such total order is sufficient and we leave it implicit in the labelling of the Hilbert spaces. Explicitly, we consider $\sigma_{\MBQC}^{g} \in \Comb(\boldsymbol{Q'} \rightarrow Q_{1}, C_{1} \rightarrow Q_{2}, ..., C_{|V\setminus O|} \rightarrow \bigotimes_{i \in O} Q_{i})$. The graph state $\rho_{G}$ will be prepared on the space $\H_{\boldsymbol{Q'}}$ and ultimately we will build classical-quantum combs from $\sigma_{\MBQC}^{g} \ast_{\boldsymbol{Q'}} \rho_{G}$ where $\ast_{\boldsymbol{Q'}}$ denotes the link product over $\H_{\boldsymbol{Q'}}$. See \Cref{fig:sigma_MBQC_g} for a depiction of $\sigma_{\MBQC}^{g}$ for one of the gflows considered in the examples below.

For completeness, a proof that $\sigma_{\MBQC}^{g} \ast_{\boldsymbol{Q'}} \rho_{G}$ does indeed perform a correct MBQC is given in \Cref{subsec:app_sigma_MBQC_correctness}. This proof is related to the existing theorem establishing sufficiency of gflow for deterministic computation \cite[Theorem 2,][]{browne2007generalized} and also provides justification for why the ordering of operators in \Cref{eq:U_corr_i} is valid.


\subsection{Results: Learning Measurement Planes} \label{subsec:meas_planes} 

Included in a given correction method for MBQC is the requirement that each qubit be measured in the correct measurement plane, i.e. in either the $XY$-, $XZ$- or $YZ$-plane of the Bloch sphere. If the exact correction method is unknown, then an incorrect choice of measurement leads to incorrect operation of the device. In this subsection, we use the comb min-entropy to quantify how well the correct measurement planes can be learnt for a concrete example.

We consider a fixed four-vertex graph $G$ and choice of input and output sets $I$ and $O$ as depicted in \Cref{fig:grey_box_graph}. For this $(G,I,O)$ there are $15$ different possible gflows, which are catalogued in \Cref{sec:app_grey_box_examples} by stipulating the map $g$ and corresponding DAG for each (an example is depicted in \Cref{fig:grey_box_gflow_example} for $g_{1}$; the labelling is given in the appendices). These $15$ gflows can be grouped into three different groups depending on which measurement plane they specify for qubit $2$; all three measurements planes are possible and there are $5$ gflows corresponding to each plane. Note that, due to the definition of gflow, any input qubit must be measured in the $XY$-plane and we allow for arbitrary measurements on the output qubits, so there is only ambiguity surrounding the measurement plane for qubit $2$.

Fortunately, to use the device correctly, the user does not need to know which exact gflow the device implements, but solely which measurement plane it assigns to qubit $2$. As such, we consider a classical-quantum comb with random variable $X$ whose values are the possible measurement plane assignments for the second qubit. The comb for this example is
\begin{align}
D_{\text{planes}} := \sum_{\text{mp}}P(\text{mp})&\ket{\text{mp}}\!\!\bra{\text{mp}} \nonumber \\
&\otimes \sum_{g \sim \text{mp}} P(g|\text{mp}) \sigma_{\MBQC}^{g} \ast_{\boldsymbol{Q'}} \rho_{G}
\end{align}

where the first sum is over the elements of $\{XY, XZ, YZ\}$. The second sum averages over all gflows $g$ that assign measurement plane $\text{mp}$ to qubit $2$ (the shorthand notation for which is $g \sim \text{mp}$) as this information is irrelevant for the user. $D_{\text{planes}}$ is an element of $\Comb(\mathbb{C} \rightarrow Q_{1}, C_{1} \rightarrow Q_{2}, C_{2} \rightarrow Q_{3,4}, \mathbb{C} \rightarrow X)$.

Assuming the distributions are uniform, i.e. $P(\text{mp}) = \frac{1}{3}$ and $P(g|\text{mp}) = \frac{1}{5}$ for each $\text{mp}$, the comb min-entropy for $D_{\text{mp}}$ reaches $0.000(0)$ indicating that the measurement planes can in fact be determined with certainty after just a single use of the device.


\subsection{Results: Learning in the Presence of Noise} \label{subsec:noise} 

Continuing with the same example, suppose now that the measurement plane is known and now we wish to learn the exact details of the gflow that the device implements. For concreteness, suppose that qubit $2$ is to be measured in the $XY$-plane (the analysis for the other measurement planes is the same). Suppose further that we only need to distinguish between those gflows that are also compatible with the order of measurements where qubit $1$ is measured before qubit $2$ (we discuss the reasons behind this in \Cref{subsec:caus_inf}). For the specific graph under consideration, there are four such gflows: $g_{1}$, $g_{2}$, $g_{4}$ and $g_{5}$ (see \Cref{sec:app_grey_box_examples} for the details of these gflows). The corresponding random variable $X$ is thus four-valued and the classical-quantum comb defined from it is
\begin{align}
D_{XY, 1<2} := \sum_{j \in J} P(g_{j})\ket{j}\!\!\bra{j} \otimes \sigma_{\MBQC}^{g_{j}} \ast_{\boldsymbol{Q'}} \rho_{G}
\end{align}

where $J = \{1,2,4,5\}$. Assuming a uniform prior, i.e. $P(g_{j}) = \frac{1}{4}$, the comb min-entropy for $D_{XY, 1<2}$ is $-\log_{2}(0.250)$. That is, nothing more can be learnt about the gflows, even under the optimal strategy. This result can be understood by considering the following proposition.

\begin{Proposition} \label{prop:no_causal_learning} Let $(G, I, O, \omega)$ be given and let $g, g'$ be gflows compatible be with $(G, I, O, \omega)$ and moreover have mutually compatible partial orders. Then, for any $c_{j} \in \{0,1\}$ for $j = 1, ..., |V\setminus O|$,
\begin{align}
\left(\ket{\boldsymbol{c}_{\setminus O}}\!\!\bra{\boldsymbol{c}_{\setminus O}}_{\boldsymbol{C}_{\setminus O}} \right) &\ast_{\boldsymbol{C}_{\setminus O}} \sigma_{\MBQC}^{g} \ast_{\boldsymbol{Q'}} \rho_{G} \nonumber\\
&= \left(\ket{\boldsymbol{c}_{\setminus O}}\!\!\bra{\boldsymbol{c}_{\setminus O}}_{\boldsymbol{C}_{\setminus O}} \right) \ast_{\boldsymbol{C}_{\setminus O}} \sigma_{\MBQC}^{g'} \ast_{\boldsymbol{Q'}} \rho_{G}
\end{align}

where $\boldsymbol{c}_{\setminus O}$ is shorthand for the measurement outcomes for all qubits not in $O$ and $\boldsymbol{C}_{\setminus O}$ labels the corresponding Hilbert spaces.
\end{Proposition}

The proof is given in \Cref{subsec:app_gflow_causal_equivalence}. The key step in the proof rests on an observation that (to the best of our knowledge) has not previously been made in the literature on gflow: for any given measurement outcomes, the correction operators induced by each $g \sim (G,I,O,\omega)$ are related to each other by a stabiliser of $\ket{G}$ and hence are in fact equivalent when applied to the graph state. Furthermore, this provides insight into the result of the previous subsection: gflows for different measurement planes are \textit{not} related by a stabiliser, and consequently are more readily distinguishable.

The indistinguishability of the gflows above relies on the graph state being a stabiliser state. Since no pure state can be prepared in practice, it is much more likely that the state preparation within the device is imperfect. In the remainder of this subsection, we investigate the consequences for learning the previously indistinguishable gflows when the pure graph state $\rho_{G}$ is replaced with a noisy version $\tilde{\rho}_{G}$. 

Consider a noise channel $\mathcal{F}_{\theta_{2},\theta_{3}}$ consisting of independent amplitude damping channels \cite{nielsen_chuang_2010} on qubits $2$ and $3$ of the graph state. That is,
\begin{align}
\mathcal{F}_{\theta_{2}, \theta_{3}}(\rho_{G}) = \sum_{i,j=0}^{1} F_{i,\theta_{2}}\otimes F_{j,\theta_{3}}  \rho_{G} F_{i,\theta_{2}}^{\dagger} \otimes F_{j,\theta_{3}}^{\dagger}
\end{align}

where
\begin{align}
F_{0, \theta} &:= \begin{bmatrix} 1 &0 \\ 0 &\cos(\theta) \end{bmatrix},\\
F_{1,\theta} &:= \begin{bmatrix} 0 &\sin(\theta) \\0 &0 \end{bmatrix},
\end{align}

and the subscript on the $\theta$ indicates the qubit of $G$ that $F_{i, \theta}$ acts upon. Taking $\tilde{\rho}_{G}(\theta_{2},\theta_{3}) := \mathcal{F}_{\theta_{2}, \theta_{3}}(\rho_{G})$, we consider the min-entropy of combs
\begin{align}
D_{XY, 1<2, \theta_{2},\theta_{3}} := \sum_{j \in J} P(g_{j})\ket{j}\!\!\bra{j} \otimes \sigma_{\MBQC}^{g_{j}} \ast_{\boldsymbol{Q'}} \tilde{\rho}_{G}(\theta_{2},\theta_{3})
\end{align}

for varying values of $\theta_{2}, \theta_{3}$. As demonstrated in \Cref{fig:noisy_graph_results}, the distinguishability of the gflows does indeed increase when the graph state is noisy (please note that the figure plots the angles against the guessing probability; recall from \Cref{sec:Q_comb_min_ent} that the min-entropy is the negative logarithm of the guessing probability). In fact, the gflows can be perfectly distinguished when $\theta_{2} = \theta_{3} = \frac{\pi}{2}$.

\begin{figure}
\centering
\includegraphics[width=0.9\columnwidth]{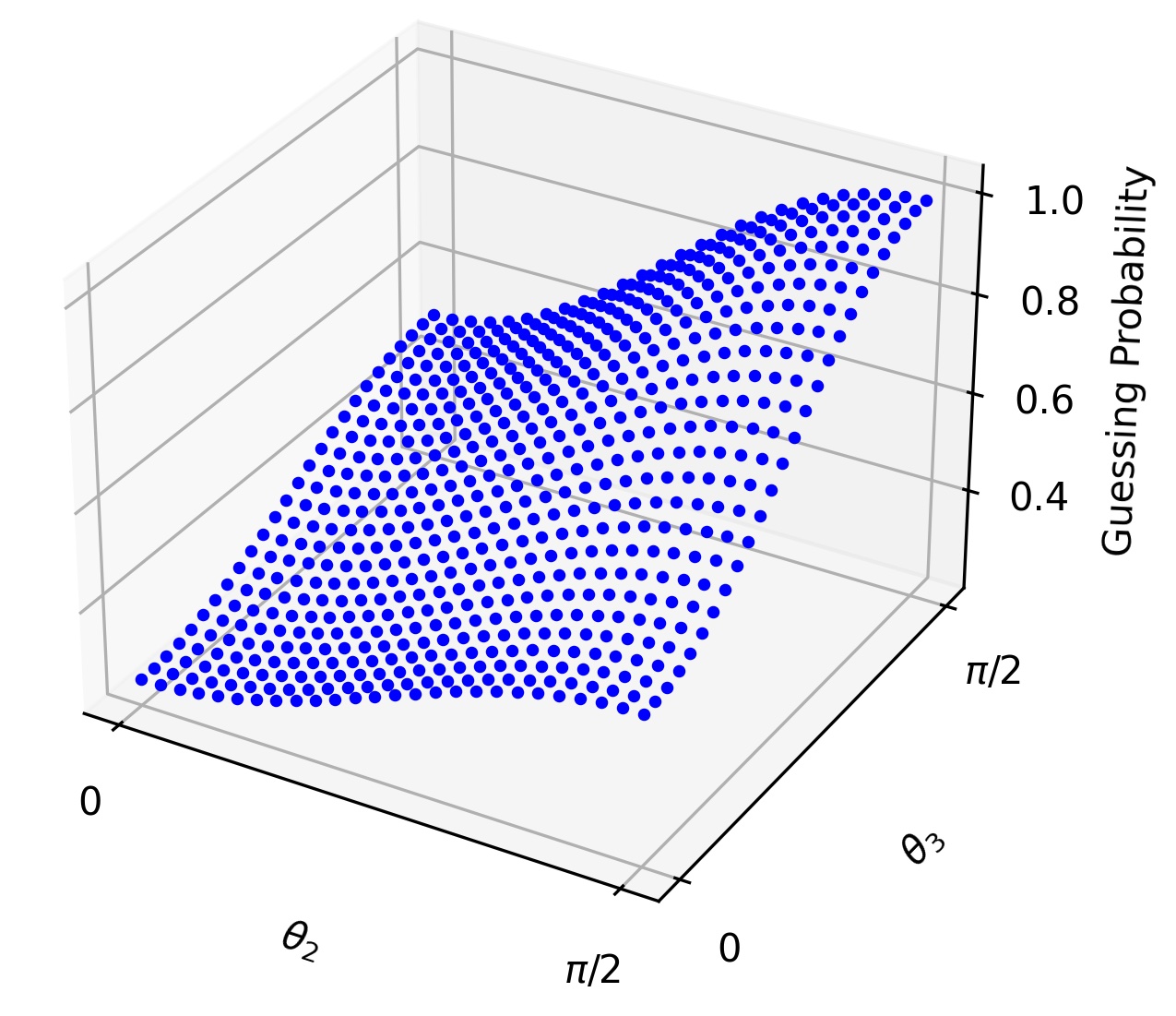}
\caption{The guessing probability for $D_{XY, 1<2}$ where the state preparation is imperfect and produces a state related to a graph state acted upon by an amplitude damping channel on qubits $2$ and $3$, as a function of the damping parameters $\theta_{2}$ and $\theta_{3}$.}
\label{fig:noisy_graph_results}
\end{figure}


\subsection{Results: Calibrating Measurements} \label{subsec:meas_calibration} 

With perfect state preparation and knowledge of the measurement planes, we can get the grey box device to function correctly, that is, to run computations deterministically. However, it is not yet enough to ensure the correctness of a \textit{specific} computation: we also need to ensure that the measurement \textit{angles} are correct, not just the measurement planes. To explain this further, consider a computation where all measurements are in the $XY$-plane, meaning that the computation is specified by the angles $\alpha_{v}$ in the projectors
\begin{align}
\ket{\pm_{\alpha_{v}}}\!\!\bra{\pm_{\alpha_{v}}} = R_{Z}(\alpha_{v})\ket{\pm}\!\!\bra{\pm}R_{Z}(\alpha)^{\dagger} \label{eq:rotated_proj}
\end{align}

for each $v$. The $\ket{\pm}$ denote the eigenvectors of the $X$ operator and when used to define single qubit measurements in MBQC, there is an implicit assumption that the $\ket{+}$ in the projector is the same as the $\ket{+}$ used in the graph state preparation (recall the definition of graph state in \Cref{eq:graph_state_CZs}). In the present context of the grey box device, the preparation occurs within the device and the measurements outside of it, so there is no \textit{a priori} guarantee that this assumption holds true.

To see why this is a problem for the computation, suppose that the $\ket{+}$ of the projection operators and the $\ket{+}$ of the graph state differ by an angle of $\theta$ in the $XY$-plane. That is,
\begin{align}
\ket{+_{\text{meas}}}\!\!\bra{+_{\text{meas}}} = R_{Z}(\theta)\ket{+_{\text{graph}}}\!\!\bra{+_{\text{graph}}}R_{Z}^{\dagger}(\theta). \label{eq:theta_adjusted}
\end{align}

If the user then measures at an angle of $\alpha$ with respect to $\ket{+_{\text{meas}}}$, then effectively they have measured at an angle of $\alpha + \theta$. The cumulative effects of the $\theta$ shift of these measurements could be catastrophic for the output computation. Accordingly, the user would need to calibrate their measurements to the device at hand to be able to receive sensible computational output. This calibration can be understood as a problem of correlating two quantum reference frames \cite{bartlett2007reference}, the one inside the device with the one outside of it.

\begin{figure}
\centering
\includegraphics[width=1.1\columnwidth]{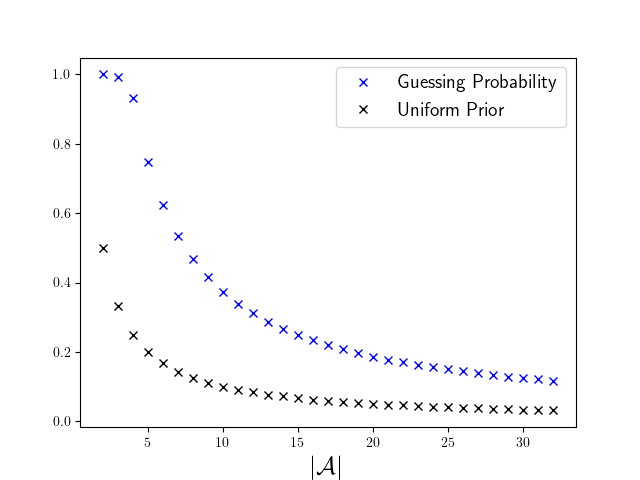}
\caption{The guessing probability for $D_{\text{calibr}}$ in a single round as the size of possible angles set $\mathcal{A}$ varies (shown in blue). The angles of $\mathcal{A}$ are taken to be evenly spaced. The probability of guessing the correct value for the angle off-set $\theta$ prior to interacting with the device is shown in black.}
\label{fig:calibr_meas_results}
\end{figure}

We now consider a simple example along the lines of the explanation above; we consider a scenario where the positive $Z$-orientation of the qubits within the device is known, but the positive $X$-axis is not (we assume a right-handed frame meaning that the positive $Y$-axis is determined if the positive $X$- and $Z$-axes are). In other words, we aim to learn $\theta$ in \Cref{eq:theta_adjusted} above. We assume $\theta$ can take values in a discrete set of angles which are evenly spaced between $0$ and $2\pi$ (this is in accordance with common methodology in finite precision quantum metrology \cite{meyer2023quantum}). We consider here a different example graph state to the previous subsection: we take a simple three-vertex linear graph (i.e. with vertices $V = \{1,2,3\}$ and edges $\{(1,2), (2,3)\}$), along with a single choice of gflow $g$ defined as $1 \mapsto \{2\}$ and $2 \mapsto \{3\}$, which gives the correction sets (only the non-empty such sets are shown):
\begin{align}
\mathcal{X}_{2} &= \{1\}; \\
\mathcal{X}_{3} &= \{2\};\\
\mathcal{Z}_{3} &= \{1\}.
\end{align}

For the purposes of writing down the comb $D_{\text{calibr}}$ for this example, we take the basis $\ket{\pm_{\text{meas}}}\!\!\bra{\pm_{\text{meas}}}$ as reference and write the graph state $\rho_{G}$ and the correction operators $U_{\text{corr}(\boldsymbol{c}),i}$ (recall \Cref{eq:U_corr_i}) with respect to this basis:
\begin{align}
\rho_{G}^{\theta} &:= R_{Z}^{\otimes n}(-\theta) \rho_{G}\left(R_{Z}^{\otimes n}(-\theta)\right)^{\dagger}; \\
U_{\text{corr}(\boldsymbol{c}),i}^{\theta} &:= \left(R_{Z}(-\theta)X_{Q_{i}}^{\bigoplus_{j \in \mathcal{X}_{i}} c_{j}}R_{Z}(-\theta)^{\dagger}\right) Z_{Q_{i}}^{\bigoplus_{j \in \mathcal{Z}_{i}} c_{j}}.
\end{align}

Denoting the comb defined from the $U_{\text{corr}(\boldsymbol{c}),i}^{\theta}$ in analogy to \Cref{eq:sigma_MBQC} as $\sigma_{\MBQC}^{\theta}$ and the (discrete) set of possible (regularly spaced) values for $\theta$ by $\mathcal{A}$, the comb of interest for this example is
\begin{align}
D_{\text{calibr}} := \sum_{\theta \in \mathcal{A}}P(\theta)\ket{\theta}\!\!\bra{\theta} \otimes \sigma_{\MBQC}^{\theta} \ast_{\boldsymbol{Q'}} \rho_{G}^{\theta}.
\end{align}

As per usual, we take the prior distribution over $\theta$ to be uniform: $P(\theta) = \frac{1}{|\mathcal{A}|}$. \Cref{fig:calibr_meas_results} presents the results for different sizes of $\mathcal{A}$, ranging from $|\mathcal{A}| = 2$ to $|\mathcal{A}| = 32$, where again the guessing probability is plotted instead of the min-entropy. For $|\mathcal{A}| = 2$, the correct direction can be known with certainty in a single round, and almost so for the case where $|\mathcal{A}| = 3$ ($P_{\text{guess}} \approx 0.992(5)$), but otherwise there is a steep decrease in the guessing probability as the size of $\mathcal{A}$ increases.


\subsection{Results: MBQC and Quantum Causal Models} \label{subsec:caus_inf}

\begin{figure*}
\begin{subfigure}{0.5\textwidth}
\centering
\includegraphics[width=0.9\textwidth]{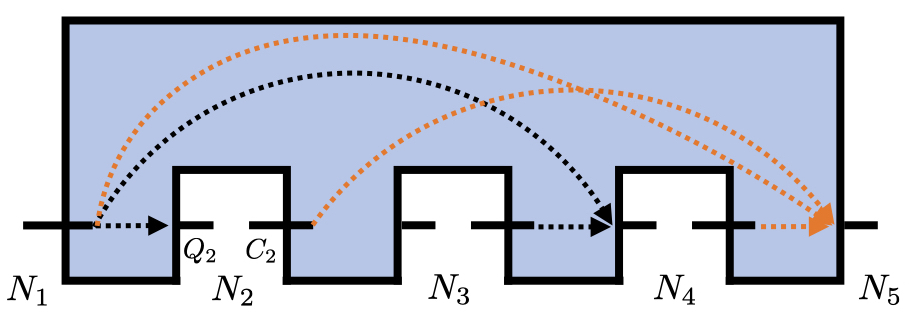}
\caption{Quantum Causal Model}
\label{fig:QCM_demo}
\end{subfigure}
\hfill
\begin{subfigure}{0.5\textwidth}
\centering
\includegraphics[width=0.5\textwidth]{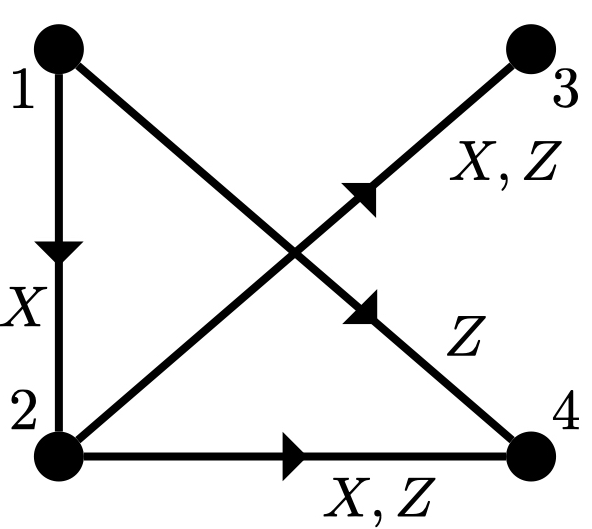}
\caption{DAG for $g_{1}$}
\label{fig:grey_box_gflow_example}
\end{subfigure}
\caption{Quantum causal models are similar to quantum combs but with extra structure. (a) The dotted arrows indicate which nodes $N_{i} = (\H_{Q_{i}},\H_{C_{i}})$ can influence which other nodes. A lack of an arrow connecting two nodes indicates the impossibility of influence, such as between $N_{2}$ and $N_{3}$. Any node at the tail of an arrow is called a parent of the node at its head. For example, the parents of $N_{5}$ are $N_{1}$, $N_{2}$ and $N_{4}$ as indicated by the orange arrows, and the corresponding map between them is denoted $\rho_{Q_{5}|C_{1}C_{2}C_{4}}$. The structure of the causal model is not made apparent by its representation as a comb alone (i.e. the blue operator without arrows). (b) Each QCM comes equipped with a directed acyclic graph, which is determined by the gflow in the examples considered here. The DAG for gflow $g_{1}$ is shown above, which measures qubits $1$ and $2$ in the $XY$-plane following the order $1<2$. The labels on the directed arrows depict the conditional correction operators, with the head and tail of the arrow denoting the target and control of the operation respectively.}
\label{fig:QCM_explanation}
\end{figure*}

To this point, we have represented each correction method specified by a gflow $g$ as a comb $\sigma_{\MBQC}^{g}$. However, this ignores a certain amount of structural information contained in the details of the gflow. For example, consider again the four qubit graph state shown in \Cref{fig:grey_box_graph}. We have already seen that this graph supports a number of different gflows, but for now let us consider just two of them: $g_{1}$, whose DAG of corrections is shown in \Cref{fig:grey_box_gflow_example}, and $g_{2}$, whose DAG of corrections is given in the appendices (see \Cref{fig:g2}). Both $\sigma_{\MBQC}^{g_{1}}$ and $\sigma_{\MBQC}^{g_{2}}$ are elements of the same set of combs, namely $\Comb(\mathbb{C} \rightarrow Q_{1}, C_{1} \rightarrow Q_{2}, C_{2} \rightarrow Q_{3,4})$, but internally there are a number of differences between them, as witnessed by their respective DAGs of corrections. For example, for $\sigma_{\MBQC}^{g_{1}}$, the input on $\H_{C_{1}}$ can have an influence on the output at $\H_{Q_{2}}$ but not $\H_{Q_{3}}$, where as in $\sigma_{\MBQC}^{g_{2}}$ the opposite is true. In our grey box device, these differences occur as different channels between different spaces and in certain circumstances, it is beneficial to represent these differences at the level of the operators $\sigma_{\MBQC}^{g_{i}}$.

To represent this extra structure, we can leverage the formalism of quantum causal models \cite{allen2017quantum,barrett2019quantum, costa2016quantum,oreshkov2012quantum, d2018causality}. Quantum causal models (QCMs) can be viewed as an extension of classical Bayesian networks within the field of classical causal modelling (see e.g. \cite{pearl2009causality, spirtes2000causation, scholkopf2012causal}). Where classical Bayesian networks consist of a set of interconnected conditional probability distributions, quantum causal models consist of a network of quantum channels. The motivation for developing a quantum version was due, at least in part, to the realisation that classical causal models are insufficient for treating quantum correlations, such as those that violate a Bell inequality \cite{wood2015lesson}.

A formal definition of a QCM is the following, where we have adapted the notation of the definition in \cite{barrett2019quantum} for consistency with the notation used throughout this section. Let $N_{i}$ denote a pair of Hilbert spaces $(\H_{Q_{i}}, \H_{C_{i}})$, with the possibility that one of the two Hilbert spaces be trivial. We refer to $N_{i}$ as a quantum node, and denote the set of all $Q$-labels and the set of all $C$-labels for all nodes as $\boldsymbol{Q}$ and $\boldsymbol{C}$ respectively. \Cref{fig:QCM_demo} demonstrates the distinction between a quantum causal model and a quantum comb by indicating the possibility of influence between different nodes.

\begin{Definition}[Definition 3.3, \cite{barrett2019quantum}]  A \textbf{quantum causal model} is given by a directed acyclic graph over quantum nodes $N_{1}, ..., N_{n}$ and for each node $N_{i}$, a quantum channel $\rho_{Q_{i}|\Pa(i)} \in \L(\H_{Q_{i}} \otimes \H_{\Pa(i)})$, (where $\Pa(i) \subseteq \boldsymbol{C}$ denotes the possibly empty set of $C$-labels of parent nodes to $N_{i}$ in the DAG) such that all channels mutually commute. This defines an operator
\begin{align*}
\sigma_{N_{1},...,N_{n}} := \prod_{i=1}^{n} \rho_{Q_{i}|\Pa(i)}.
\end{align*}
\end{Definition}

The operator $\sigma_{N_{1},...,N_{n}}$ is an element of $\Comb(\mathbb{C} \rightarrow Q_{\pi(1)}, C_{\pi(1)} \rightarrow Q_{\pi(2)}, ..., C_{\pi(n-1)} \rightarrow Q_{\pi(n)})$ for any permutation $\pi$ of $\{1,..., n\}$ such that $\pi(1) < \pi(2) < ... < \pi(n)$ is compatible with the topological ordering of the DAG. 

We now establish that the correction operators corresponding to a gflow are indeed QCMs before discussing some of the consequences of this fact.

\begin{Proposition} \label{prop:MBQC_QCM} For each gflow $g$, $\sigma_{\MBQC}^{g}$ is a quantum causal model.
\end{Proposition}

As explicitly demonstrated in \Cref{prop:acyclic}, each gflow $g$ (including the partial order) induces a DAG on the vertices of the corresponding graph $G$ (we have already seen examples of such DAGs, such as in \Cref{fig:grey_box_gflow_example}). For each $i \in V$, the parents of $i$ are given by the union of the correction sets defined in \Cref{eq:X_Z_correction_sets}:
\begin{align}
\Pa(i) := \mathcal{X}_{i} \cup \mathcal{Z}_{i}.
\end{align}

We then define:
\begin{align}
\rho_{Q_{i}|C_{\Pa(i)}, Q'_{i}} := \sum_{\substack{a_{i}, b_{i} \\ \boldsymbol{c}_{\Pa(i)}}}& U_{\text{corr}(\boldsymbol{c}_{\Pa(i)}), i} \ket{a_{i}}\!\!\bra{b_{i}}_{Q_{i}}U_{\text{corr}(\boldsymbol{c}_{\Pa(i)}), i}^{\dagger} \nonumber\\
&\quad \quad \otimes \ket{\boldsymbol{c}_{\Pa(i)}a_{i}}\!\!\bra{\boldsymbol{c}_{\Pa(i)}b_{i}}_{C_{\Pa(i)},Q'_{i}} \label{eq:struc_qchannel}
\end{align}

with $\boldsymbol{c}_{\Pa(i)}$ denoting the classical measurement outcomes pertaining to the elements of $\Pa(i)$. The channel $\rho_{Q_{i}|C_{\Pa(i)}, Q'_{i}}$ has a very similar structure to \Cref{eq:sigma_MBQC}, but contains only the relevant spaces and corrections pertaining to $Q_{i}$. It is immediate to see that the $\rho_{Q_{i}|C_{\Pa(i)}, Q'_{i}}$ commute since the only overlap between any two such channels occurs on (products of) Hilbert spaces $\H_{C_{i}}$, where the operators are diagonal. The operator $\sigma_{\MBQC}^{g}$ as defined in \Cref{eq:sigma_MBQC} is readily seen as equivalent to the product of channels $\rho_{Q_{i}|C_{j : j \in \Pa(i)}, Q'_{i}}$
\begin{align}
\sigma_{\MBQC}^{g} = \prod_{i = 1}^{n} \rho_{Q_{i}|C_{\Pa(i)}, Q'_{i}}.
\end{align}

This result is significant since the connection between MBQC and QCMs has potential utility for both fields. For example, ongoing research within the QCM community is concerned with investigating the causal structure of unitary channels \cite{lorenz2021causal,ormrod2023causal}. Since MBQC provides a distinct perspective of unitary transformations than that of the circuit model, the above causal interpretation of gflow may prove beneficial for progress on questions surrounding causal representations of such transformations. These considerations are disjoint from the comb min-entropy approach taken in this work, and so we pursue this perspective elsewhere. Instead we turn to a final example which includes both the min-entropy approach and causal perspective developed here.


\subsubsection{Causal Discovery Example} \label{subsec:caus_inf_example}

As stated previously, there are $15$ different gflows for the four-qubit graph state considered throughout this section. Using the above perspective, that means that there are $15$ different causal structures that might be implemented within the grey box device. In \Cref{subsec:noise}, we considered a restricted example where the aim was to distinguish between just four of the $15$ gflows corresponding to a given measurement plane and partial order. Here, we consider the more general case and discuss the subtleties regarding incompatible orders that were skipped above.

In \Cref{sec:app_grey_box_examples}, the $15$ different gflows are catalogued and grouped according to which measurement plane is assigned to the second qubit (i.e. three groups, $5$ gflows in each). Within each group, four gflows are compatible with the partial order where the first qubit is measured before the second, and one gflow is not. For example, the gflow $g_{3}$ (see \Cref{sec:app_grey_box_examples} for labelling) is incompatible with the partial order $1<2$ and consequently, the operator $\sigma_{\MBQC}^{g_{3}}$ on $\L(\H_{Q_{1}} \otimes \H_{C_{1}} \otimes \H_{Q_{2}} \H_{C_{2}} \otimes \H_{Q_{3,4}})$ is \textit{not} an element of $\Comb(\mathbb{C} \rightarrow Q_{1}, C_{1} \rightarrow Q_{2}, C_{2} \rightarrow Q_{3,4})$ since the output at $Q_{1}$ depends on the input at $C_{2}$. Thus, if we consider the operator
\begin{align}
D_{\text{gflow}} := \sum_{j = 1}^{15} P(g_{j})\ket{j}\!\!\bra{j} \otimes \sigma_{\MBQC}^{g_{j}} \ast_{\boldsymbol{Q'}} \rho_{G}.
\end{align}

where no modifications have been made to the $\sigma_{\MBQC}^{g_{j}}$ corresponding to the incompatible gflows, then $D_{\text{gflow}}$ also fails to be a comb. To rectify this, we consider
\begin{align}
\widehat{D}_{\text{gflow}} := \sum_{j = 1}^{15} P(g_{j})\ket{j}\!\!\bra{j} \otimes \widehat{\sigma}_{\MBQC}^{g_{j}} \ast_{\boldsymbol{Q'}} \rho_{G}
\end{align}

where
\begin{align}
\widehat{\sigma}_{\MBQC}^{g_{j}} = \begin{cases} \sigma_{\MBQC}^{g_{j}}, \text{ if } j \neq 3,8,12 \\ \text{SWAP}_{Q'_{1},Q'_{2}} \sigma_{\MBQC}^{g_{j}} \text{SWAP}_{Q'_{1},Q'_{2}}^{\dagger}, \text{ otherwise} \end{cases}.
\end{align}

The notation $\text{SWAP}_{Q'_{1},Q'_{2}}$ denotes a swap operation between incoming graph states qubits $Q'_{1}$ and $Q'_{2}$. Accordingly, $\widehat{D}_{\text{gflow}}$ is a valid comb, and obtains the min-entropy value of around $-\log(0.373) \approx 1.423$ bits.

If we remove the incompatible gflows ($g_{3}, g_{8}, g_{12}$) and consider instead the operator (which is a valid comb)
\begin{align}
D_{\text{gflow}, 1<2} := \sum_{j \neq 3,8,12} P(g_{j})\ket{j}\!\!\bra{j} \otimes \sigma_{\MBQC}^{g_{j}} \ast_{\boldsymbol{Q'}} \rho_{G}
\end{align}

then we obtain the min-entropy value of $-\log(0.250) = 2$ bits, which is consistent with the results for $D_{\text{planes}}$ and $D_{XY, 1<2}$ above (i.e. that it is possible to learn the measurement plane with certainty but nothing more in the noiseless case).

To conclude, we place the above analysis into context within current literature on quantum causal models and more broadly within quantum information theory. Previously, investigations have been conducted into the distinguishability of quantum causal structures, however these focused almost entirely on structures on two qubits only \cite{ried2015quantum,fitzsimons2015quantum,kubler2018two,hu2018discrimination}. One exception to this is Ref. \cite{chiribella2019quantum} which investigated causal scenarios involving more systems but focused on a question of identification of which system from a list was influenced by a given cause. To the best of our knowledge, the above example is the first related to distinguishing between truly multi-partite quantum causal structures. 

The above example can also be seen from a slightly different perspective, one of discriminating between quantum networks. Recently, Ref. \cite{hirche2023quantum} established a number of results regarding the discrimination of networks, i.e. combs, in the context of quantum hypothesis testing. These results are phrased in terms of a variety of divergences suitably generalised to treat quantum combs. Since the min-entropy is closed related to the max divergence (see e.g., \cite{tomamichel2015quantum}), the above example can also be understood from the point of view of symmetric hypothesis testing (we develop this view elsewhere).


\section{Discussion} \label{sec:discussions}

By interacting with a system, possibly via some complicated sequence of actions, it is possible to learn an unknown property that determines the evolution of the system. By modelling such a situation as classical-quantum combs, it is possible to leverage the comb min-entropy \cite{chiribella2016optimal} to quantify how much can be learnt about the unknown property in the best case scenario. Due to the general, and also natural, modelling of both quantum and classical interactions as combs of the form considered above, the methodology showcased here has broad applicability.

In this work, we restricted our attention to a novel set of combs defined by the paradigm of measurement-based quantum computation. In particular, we defined a classical comb which models a specific BQC protocol \cite{mantri2017flow} and by so doing, gave a proof of partial security based on the comb min-entropy for both a single round of the protocol as well as the multi-round case, extending the security analysis in the existing literature.

We further defined a series of combs that model interactions with an MBQC device under varying levels of knowledge regarding the inner working of the device. For different choices of classical variable representing different aspects related to the functionality of the device, we investigated how well a user could learn the necessary information to properly use the device, including in the presence of noise. Additionally, we established a novel connection between MBQC and quantum causality, which allowed the comb min-entropy to be used to quantify the optimal causal discovery in the associated examples.


\subsection{Limitations of the Min-Entropy Approach} \label{subsec:limitations}

Despite the broad applicability and operational meaning of the methodology used in this work, there are certain limitations of which one should be made aware. Since quantum combs are operators on the tensor product of many Hilbert spaces, their dimensionality becomes quite large. Consequently, when using a numerical SDP solver to compute the comb min-entropy, size issues quickly start to play a role. For example, in the three-qubit graph state example for the BQC protocol in \Cref{subsec:single_round_thm}, as a square matrix, the operator $D_{\client}^{(m)}$ has dimension $|\mathcal{A}|^{3 + 3m}2^{3m}$ where $m$ is the number of rounds. For the minimal possible choice of angle set, $|\mathcal{A}| = 4$ and for a single round, this already equals $32768$. During the numerical analysis, it was observed that much of the memory and time cost involved in running the SDP solver arose from establishing the necessary comb constraints defining the problem (see \Cref{sec:app_SDP_implementation} for more details). It is likely that there exist improvements to these costs from optimising the software implementation and representation of these constraints.


\subsection{Future Work}

As mentioned above, the generality of the combs framework ensures that the methodology introduced in this work can be applied in a wide variety of contexts. We conclude by outlining an example avenue for future investigation. One feature of many SDP solvers that was not utilised in this work, regards the possibility to return an optimal solution instance (and its dual) along with the optimal value (e.g. the min-entropy) for the problem at hand. Since either the primal or dual solution (depending on the implementation) will be a matrix representing a strategy for interacting with the system in question, it would be interesting to analyse these matrices in order to infer what the corresponding strategy would be. For example, the optimal solutions of the SDP for the varying values of $\theta$ in the measurement calibration example (\Cref{subsec:meas_calibration}) could help identify a type of quantum instrument or series of instruments that are optimal for this type of parameter estimation in general.


\section{Acknowledgements}

We would like to thank Atul Mantri for useful discussions about the security of the BQC protocol of \cite{mantri2017flow}, Simon Milz for insightful comments regarding quantum combs, Joshua Morris for invaluable suggestions regarding semi-definite programming resources, and Sofiene Jerbi for useful conversations regarding gflow and MBQC. We would also like to thank an anonymous reviewer for many helpful suggestions that improved the manuscript. We acknowledge support from the Austrian Science Fund (FWF) through DK-ALM: W$1259$-N$27$ and SFB BeyondC F$7102$, and from the University of Innsbruck. This work was also co-funded by the European Union (ERC, QuantAI, Project No. $101055129$). Views and opinions expressed are however those of the author(s) only and do not necessarily reflect those of the European Union or the European Research Council. Neither the European Union nor the granting authority can be held responsible for them.

\printbibliography

\onecolumn

\newpage

\appendix


\section{Quantum Combs Supplementary} \label{subsec:app_qcombs} 

In this section, we motivate the definition of quantum comb given in the main text (\Cref{def:qcomb}), including why this is an appropriate representation of quantum networks and what the conditions of the definition correspond to. We provide these explanations here for completeness; they are drawn from references such as \cite{chiribella2008quantum, chiribella2009theoretical} which should be referred to for further details.

The definition of quantum comb given earlier is as an operator on a tensor product of Hilbert spaces, which represents a series of connected quantum channels. This correspondence between operators on a tensor product of spaces and maps between these spaces, that is, between elements of $\L(\H_{A} \otimes \H_{B})$ and $\L(\L(\H_{A}), \L(\H_{B}))$ respectively, makes use of the Choi-Jamiolkowski isomorphism \cite{choi1975completely,jamiolkowski1972linear}. Since we are most interested in the operator form, we make the following definition.

\begin{Definition} \label{def:Choi} For a linear map $\mathcal{D} : \L(\H_{A}) \rightarrow \L(\H_{B})$, its \textbf{Choi operator} is defined as:
\begin{equation}
D := (\dim A)(\mathcal{D} \otimes I_{\H_{A}})( \ket{\Phi^{+}}\!\!\bra{\Phi^{+}}_{A'A}).
\end{equation}
Here, $\H_{A'} \simeq \H_{A}$ is a copy of the input space, and $\ket{\Phi^{+}} := \frac{1}{\sqrt{\dim A}}\sum_{i = 0}^{d_{A} -1} \ket{i i}_{A'A}$ is a maximally entangled state on $\H_{A'} \otimes \H_{A}$.
\end{Definition}

The defining properties of CPTP maps correspond to properties of the Choi operator:

\begin{Lemma} \label{lem:ChoiVersusChannel} The Choi operator $D$ of a linear map $\mathcal{D}$ satisfies the following properties:
\begin{enumerate}
	\item $D$ is positive semi-definite, denoted $D \geq 0$, iff $\mathcal{D}$ is completely positive;
	\item $D$ is Hermitian iff $\mathcal{D}$ is Hermitian-preserving;
	\item $\Tr_{B}\left[D \right] \leq I_{A}$ iff $\mathcal{D}$ is trace non-increasing;
	\item $\Tr_{B}\left[D \right] = I_{A}$ iff $\mathcal{D}$ is trace-preserving. 
\end{enumerate}
\end{Lemma}

One notices that the first and last items of the above lemma feature in \Cref{def:qcomb}, however iteratively for the latter item. To arrive at the iterative constraints, consider a simple example of two CPTP maps $\mathcal{D}_{1}: \L(\H_{A_{1}}) \rightarrow \L(\H_{B_{1}} \otimes \H_{C})$ and $\mathcal{D}_{2}: \L(\H_{A_{2}} \otimes \H_{C}) \rightarrow \H_{B_{2}}$ with Choi operators $D_{1}$ and $D_{2}$ respectively. The composition of the two maps over the space $H_{C}$, denoted $\mathcal{D}_{2} \circ_{C} \mathcal{D}_{1}$ is a linear operator from $\L(\H_{A_{1}} \otimes \H_{A_{2}})$ to $\L(\H_{B_{1}} \otimes \H_{B_{2}})$, and the corresponding Choi operator is given by
\begin{align}
D = D_{1} \ast_{C} D_{2}
\end{align}

where $\ast_{C}$ denotes the link product over the space $C$, which is the analogue of composition in the Choi operator picture. By tracing over $B_{2}$, we get
\begin{align}
\Tr_{B_{2}}\left[D\right] &= \Tr_{B_{2}}\left[D_{1} \ast_{C} D_{2}\right] \\
&= D_{1} \ast_{C} \Tr_{B_{2}}\left[D_{2}\right] \\
&= D_{1} \ast_{C} (I_{C} \otimes I_{A_{2}}) \\
&= I_{A_{2}} \otimes \Tr_{C}\left[D_{1}\right]
\end{align} 

where we have used the trace-preservation criterion of \Cref{lem:ChoiVersusChannel}. Noting that $\Tr_{C}\left[D_{1}\right]$ is a positive semi-definite operator, we see that the last line above is indeed of the form of the constraints in \Cref{def:qcomb}. By tracing also over $B_{1}$, we can use the trace-preservation again to obtain the terminal constraint of normalisation to $1$. This reasoning extends to any number of composed maps, which helps demonstrate the conciseness of the comb notation.

The remainder of this appendix subsection provides the proof of \Cref{prop:X_A_independence} in the main text. The proof considers a classical-quantum comb $D$ as in \Cref{eq:X_indexed_comb}.

\begin{proof}[Proof of \Cref{prop:X_A_independence}] The proof consists of showing that the sequence of partial trace constraints (\Cref{def:qcomb}) are satisfied for both the sequence where $X$ is traced over first then $B_{n}, ..., B_{1}$ and the sequence $B_{n}, ..., B_{1}, X$. Starting with the former, tracing over $X$ gives
\begin{align}
\Tr_{X}[D] = \sum_{x} P(x)\sigma_{x}  \equiv I_{\mathbb{C}} \otimes \sum_{x} P(x)\sigma_{x}.
\end{align}

Since the $\sigma_{x}$ are normalised combs, and since the set of normalised combs is convex, $\sum_{x} P(x)\sigma_{x}$ is also a normalised comb and so the remaining trace conditions are satisfied, thus showing that $D \in \Comb(A_{1} \rightarrow B_{1}, ..., A_{n} \rightarrow B_{n}, \mathbb{C} \rightarrow X)$. 

For the other sequence, define a sequence of operators $D_{k}$, $k = 0, ..., n+1$ via
\begin{align}
D_{k} &:= \sum_{x} P(x)\ket{x}\!\!\bra{x} \otimes C_{x, k-1} \quad \quad \forall k = 1, ..., n+1 \\
D_{0} &:= \sum_{x} P(x)
\end{align}

where $C_{x, 0}, ..., C_{x, n}$ denote the positive semi-definite operators that satisfy the comb conditions for $\sigma_{x}$, that is:
\begin{align}
\sigma_{x} &= C_{x, n}; \\
\Tr_{B_{j}}[C_{x, j}] &= I_{A_{j}} \otimes C_{x, j-1} \quad \quad \forall j = 1, ..., n; \\
C_{x, 0} &= 1.
\end{align}

It follows that $D_{n+1} = D$ and that 
\begin{align}
\Tr_{B_{k-1}}\left[D_{k}\right] &= \sum_{x} P(x)\ket{x}\!\!\bra{x} \otimes \Tr_{B_{k-1}}[C_{x, k-1}] \\
&= \sum_{x} P(x)\ket{x}\!\!\bra{x} \otimes I_{A_{k-1}} \otimes C_{x, k-2} \\
&= I_{A_{k-1}} \otimes D_{k-1}
\end{align}

for all $k = 2, ..., n+1$. For the remaining cases, we have that  
\begin{align}
\Tr_{X}[D_{1}] &= \Tr_{X}\left[\sum_{x} P(x)\ket{x}\!\!\bra{x} \otimes 1  \right] \\
&= \sum_{x} P(x) \\
&= I_{\mathbb{C}} \otimes D_{0}
\end{align}

and that $D_{0} = 1$. Thus, $D$ is also in $\Comb(\mathbb{C} \rightarrow X, A_{1} \rightarrow B_{1}, ..., A_{n} \rightarrow B_{n})$.
\end{proof}


\section{Comb Min-Entropy Supplementary} \label{sec:app_min_entropy}

In \Cref{subsec:comb_min_ent_results}, we presented some results related to the min-entropy of classical-quantum combs, namely \Cref{lem:min_entropy_non_increasing} and \Cref{prop:bayesian_classical_comb}. The proofs of these results, and some supporting discussion, use a different (but equivalent) form for the comb min-entropy, which we now establish.

In the main text, the comb min-entropy for $D \in \Comb(A_{1} \rightarrow B_{1}, ... A_{n} \rightarrow B_{n})$ was given as
\begin{align}
H_{\min}(B_{n}|A_{1},B_{1},...,A_{n-1},B_{n-1})_{D} &:= -\log \left[ \max_{E} D \ast E \right].
\end{align}

where the link product is over all $A_{1},B_{1},...,A_{n},B_{n}$. However, due to the strong duality of this SDP, it is possible to write the min-entropy equivalently as \cite{chiribella2016optimal}:
\begin{align}
H_{\min}(B_{n}|A_{1},B_{1},...,A_{n-1},B_{n-1})_{D} &= -\log \left[ \min_{\Gamma} \min \{\lambda \in \mathbb{R} : I_{A_{n}B_{n}} \otimes \lambda\Gamma \geq D \} \right] \label{eq:comb_min_entropy_min} 
\end{align}

where $\Gamma \in \Comb(A_{1} \rightarrow B_{1}, ..., A_{n-1} \rightarrow B_{n-1})$. Since $\lambda \leq 1$, it is sometimes convenient to consider the $\lambda$ and $\Gamma$ together as a single operator, namely, as a sub-normalised comb. In this form, we write the min-entropy as
\begin{align}
H_{\min}(B_{n}|A_{1},B_{1},...,A_{n-1},B_{n-1})_{D} &=  -\log\left[ \min_{\widehat{\Gamma} \text{ s.t. } I_{A_{n}B_{n}} \otimes \widehat{\Gamma} \geq D} \frac{1}{\prod_{j=1}^{n-1}\dim A_{j}} \Tr[\widehat{\Gamma}] \right] \label{eq:min_ent_unnormed}
\end{align}

where $\widehat{\Gamma}$ is a sub-normalised comb on the same space as $\Gamma$. This latter form was found to be particularly amenable to implementation in code - see \cite{Smith_Min_Entropy_and_MBQC}.

The following is the proof of \Cref{lem:min_entropy_non_increasing} which states that the min-entropy for multi-round combs, i.e. of the form $D^{(m)}$, is non-increasing as the number of rounds $m$ increases.

\begin{proof}[Proof of \Cref{lem:min_entropy_non_increasing}] Let $\lambda \in \mathbb{R}$ and $\Gamma \in \Comb\left(A_{1}^{(1)} \rightarrow B_{1}^{(1)}, ..., A_{n}^{(m)} \rightarrow B_{n}^{(m)}  \right)$ be such that $I_{X} \otimes \lambda \Gamma \geq D^{(m)}$ and that $H_{\min}(X|A_{1}^{(1)},B_{1}^{(1)},...,A_{n}^{(m)},B_{n}^{(m)})_{D^{(m)}} = -\log(\lambda)$. Let $\Gamma' \in \Comb\left(A_{1}^{(1)} \rightarrow B_{1}^{(1)}, ..., A_{n}^{(l)} \rightarrow B_{n}^{(l)}  \right)$ be given by
\begin{align}
\Gamma' := \frac{1}{\prod_{t = l+1}^{m} \dim \boldsymbol{A}^{(t)}} \Tr_{\boldsymbol{A}^{(l+1)}, \boldsymbol{B}^{(l+1)}...\boldsymbol{A}^{(m)}, \boldsymbol{B}^{(m)}}\left[\Gamma \right]
\end{align}

where $\dim \boldsymbol{A}^{(t)} := \prod_{k = 1}^{n}\dim A_{k}^{(t)}$ and the subscripts in the trace indicate that the trace is over all subspace related to $\sigma_{x}^{(j)}$ for $j > l$, i.e. $\H_{A_{1}}^{(l+1)}, \H_{B_{1}}^{(l+1)}, ..., \H_{B_{n}}^{(l+1)}, \H_{A_{1}}^{(l+2)},...,\H_{B_{n}}^{(m)}$. 

It remains only to show that $I_{X} \otimes \lambda \Gamma' \geq D^{(l)}$. It suffices to show that 
\begin{align}
\frac{1}{\prod_{t = l+1}^{m} \dim \boldsymbol{A}^{(t)}} \Tr_{\boldsymbol{A}^{(l+1)}, \boldsymbol{B}^{(l+1)}...\boldsymbol{A}^{(m)}, \boldsymbol{B}^{(m)}}\left[D^{(m)}\right] = D^{(l)}.
\end{align}

Starting from the left-hand side:
\begin{align}
\frac{1}{\prod_{t = l+1}^{m} \dim \boldsymbol{A}^{(t)}} \Tr_{\boldsymbol{A}^{(l+1)}, \boldsymbol{B}^{(l+1)}...\boldsymbol{A}^{(m)}, \boldsymbol{B}^{(m)}}\left[ D^{(m)}\right] &= \frac{1}{\prod_{t = l+1}^{m} \dim \boldsymbol{A}^{(t)}} \sum_{x} P(x)\ket{x}\!\!\bra{x} \otimes \Tr_{\boldsymbol{A}^{(l+1)}, \boldsymbol{B}^{(l+1)}...\boldsymbol{A}^{(m)}, \boldsymbol{B}^{(m)}}\left[ \bigotimes_{j = 1}^{m} \sigma_{x}^{(j)} \right] \\
&= \frac{1}{\prod_{t = l+1}^{m} \dim \boldsymbol{A}^{(t)}} \sum_{x} P(x)\ket{x}\!\!\bra{x} \otimes \left( \prod_{t = l+1}^{m} \dim \boldsymbol{A}^{(t)} \right) \bigotimes_{j = 1}^{l} \sigma_{x}^{(j)} \\
&= D^{(l)}
\end{align}

where we have used the fact that each $\sigma_{x}^{(j)}$ is a comb, so $\Tr[\sigma_{x}^{(j)}]$ is equal to the product of input space dimensions by definition. So, it holds that $I_{X} \otimes \lambda \Gamma' \geq D^{(l)}$, which entails that
\begin{align}
H_{\min}(X|A_{1}^{(1)},B_{1}^{(1)},...,A_{n}^{(l)},B_{n}^{(l)})_{D^{(l)}} \geq -\log(\lambda) = H_{\min}(X|A_{1}^{(1)},B_{1}^{(1)},...,A_{n}^{(m)},B_{n}^{(m)})_{D^{(m)}}
\end{align}

thus proving that the min-entropy is non-increasing with increasing round number.
\end{proof}

In \Cref{subsec:comb_min_ent_results}, we also considered combs of the form
\begin{align}
D = \sum_{x}P(x)\ket{x}\!\!\bra{x} \otimes \sigma_{x}
\end{align}

where all the $\sigma_{x}$ are diagonal in the same basis. The ultimate aim was to provide an interpretation of the guessing probability, $P_{\text{guess}}(X| \boldsymbol{A}, \boldsymbol{B})$, in terms of Bayesian updating since such an interpretation exists for the state min-entropy (see eg., \cite[Section 6.1.4,][]{tomamichel2015quantum}). For completeness, we present this interpretation here and discuss the similarities and differences with the combs case.

Let $D$ be a classical-quantum state on $\H_{X} \otimes \H_{Y}$ (i.e. $D \in \Comb(\mathbb{C} \rightarrow Y, \mathbb{C} \rightarrow X)$) such that each $\sigma_{x} \in \L(\H_{Y})$ is diagonal in the same basis $\{\ket{y}\}$. We can then write
\begin{align}
\sigma_{x} = \sum_{y} P(y|x)\ket{y}\!\!\bra{y}_{Y}
\end{align}

where $P(y|x)$ is a conditional probability distribution. We can thus write
\begin{align}
D = \sum_{x,y}P(x)P(y|x)\ket{xy}\!\!\bra{xy}_{XY}.
\end{align}

Applying Bayes' rule, it follows that
\begin{align}
D = \sum_{x,y}P(x|y)P(y)\ket{xy}\!\!\bra{xy}_{XY}.
\end{align}

By maximising over $x$ for each $y$, we obtain the inequality:
\begin{align}
D &\leq \sum_{x,y} \left[\max_{\tilde{x}}P(\tilde{x}|y)P(y) \right]\ket{xy}\!\!\bra{xy}_{XY} \\
&= I_{X} \otimes \sum_{y}\max_{\tilde{x}}P(\tilde{x}|y)P(y)\ket{y}\!\!\bra{y}_{Y}. \label{eq:classical_unnormed_state}
\end{align}

The second tensor factor above is an (in general) unnormalised state on $Y$, and moreover it can be shown that this unnormalised state is a minimal such state $\rho_{Y}$ for which $D \leq I_{X} \otimes \rho_{Y}$ holds (this follows from the proof of the left-hand inequality of \Cref{prop:bayesian_classical_comb}). Thus, using the unnormalised version of the min-entropy (recall \Cref{eq:min_ent_unnormed}), we arrive at
\begin{align}
P_{\text{guess}}(X|Y)_{D} = \sum_{y} \max_{\tilde{x}}P(\tilde{x}|y)P(y)
\end{align}

and hence also at the desired interpretation of the guessing probability in terms of Bayesian updating: the guessing probability is the maximal Bayesian update for each interaction (denoted by $y$) averaged over all possible interactions. Clearly, for $P_{\text{guess}}$ to take value $1$, we must have perfect updates for \textit{every} interaction (in the support of $P(y)$).

It is worthwhile emphasising here that, in the present case, the states $\sigma_{x}$ can be considered as combs with trivial input spaces, ie. $\sigma_{x} \in \Comb(\mathbb{C} \rightarrow Y)$, and hence we are guaranteed that $\sum_{y}\max_{x}P(x|y)P(y)\ket{y}\!\!\bra{y}$ is an unnormalised state. For the general case where the $\sigma_{x}$ are classical combs with non-trivial input spaces, we have no analogous guarantee as we will now discuss.

In the main text, we used the conditional Bayes' rule and an independence condition to write the classical-classical comb $D$ as
\begin{align}
D = \sum_{x,\boldsymbol{a},\boldsymbol{b}} P(x|\boldsymbol{a},\boldsymbol{b}) P(\boldsymbol{b}|\boldsymbol{a})\ket{x,\boldsymbol{a},\boldsymbol{b}}\!\!\bra{x,\boldsymbol{a},\boldsymbol{b}}_{X,\boldsymbol{A},\boldsymbol{B}}
\end{align}

Similarly to above, we can obtain the inequality:
\begin{align}
D \leq I_{X} \otimes \sum_{\boldsymbol{a}, \boldsymbol{b}} \max_{x}P(x|\boldsymbol{a},\boldsymbol{b})P(\boldsymbol{b}|\boldsymbol{a}) \ket{\boldsymbol{a},\boldsymbol{b}}\!\!\bra{\boldsymbol{a},\boldsymbol{b}}{\boldsymbol{A},\boldsymbol{B}}
\end{align}

However, unlike above, we are not guaranteed that the second tensor factor is an unnormalised classical comb: $\sum_{\boldsymbol{a}, \boldsymbol{b}} \max_{x}P(x|\boldsymbol{a},\boldsymbol{b})P(\boldsymbol{b}|\boldsymbol{a}) \ket{\boldsymbol{a},\boldsymbol{b}}\!\!\bra{\boldsymbol{a},\boldsymbol{b}}$ may fail the required marginalisation conditions due to the dependence on the inputs $\boldsymbol{a}$ that persists in the distribution $P(x|\boldsymbol{a},\boldsymbol{b})$ in the maximum. The trace of this operator (appropriately normalised by the dimension of the input spaces) still provides a lower bound for the guessing probability, just as for the state case above, but the assurance that this bound can be reached is lacking. 

To obtain an upper bound, we can construct an operator that removes the dependence on the inputs by also maximising over the $\boldsymbol{a}$, which then trivially satisfies the required marginalisation conditions. These bounds are established more formally in the following proof of \Cref{prop:bayesian_classical_comb}. Before presenting the proof of the proposition, we prove two useful lemmas which establish that, for any classical comb $D$, we need only (un)normalised classical combs $\Gamma$ in the minimum formulations of the min-entropy since if some non-classical comb achieves the minimum, then there exists a related classical comb that does also.

\begin{Lemma} \label{lem:diag_comb_inequality} Let $A$ and $B$ be operators on $\H$, where $A$ is diagonal in a specific basis and $B$ is an arbitrary positive semi-definite operator. If $B - A \geq 0$ then $\diag(B) - A \geq 0$.
\end{Lemma}

\begin{proof} Let $\mathcal{E}: \L(\H) \rightarrow \L(\H)$ be a decohering channel in the basis for which $A$ is diagonal. Since this is in particular a positive map and $B-A$ is positive semi-definite by assumption, it follows that $\mathcal{E}(B-A)$ is also positive semi-definite Linearity of $\mathcal{E}$ ensures that $\mathcal{E}(B-A) = \mathcal{E}(B) - \mathcal{E}(A) = \diag(B) - A$ giving the result.
\end{proof}

\begin{Lemma} \label{lem:classical_comb_if_quantum_comb} If $\Gamma$ is a normalised (resp. unnormalised) quantum comb, then $\diag(\Gamma)$ is a normalised (resp. unnormalised) classical comb.
\end{Lemma}

\begin{proof} The proof is essentially immediate by definition but the details are spelt out nonetheless. Let $\Gamma$ be a normalised (resp. unnormalised) quantum comb in $\Comb(A_{1} \rightarrow B_{1}, ..., A_{n} \rightarrow B_{n})$ and let $C^{(k)}$, $k = 0, ..., n$ be positive semi-definite operators that satisfy $C^{(n)} = \Gamma$,
\begin{align}
\Tr_{B_{k}}\left[C^{(k)}\right] = I_{A_{k}} \otimes C^{(k-1)}
\end{align}

for $k = 1, ..., n$ and $C^{(0)} = 1$ (resp. $C^{(0)} > 0$). Denoting the $\boldsymbol{a}\boldsymbol{b}^{\text{th}}$ diagonal element of $\Gamma$ by $\alpha_{\boldsymbol{a}\boldsymbol{b}}$, we define $f$ via $f(\boldsymbol{b},\boldsymbol{a}) = \alpha_{\boldsymbol{a}\boldsymbol{b}}$. Thus, we can write
\begin{align}
\diag(\Gamma) = \sum_{\boldsymbol{a}, \boldsymbol{b}} f(\boldsymbol{b},\boldsymbol{a})\ket{\boldsymbol{a}\boldsymbol{b}}\!\!\bra{\boldsymbol{a}\boldsymbol{b}}_{\boldsymbol{A},\boldsymbol{B}}.
\end{align}

Since $\Gamma$ is positive semi-definite every $\alpha_{\boldsymbol{a}\boldsymbol{b}}$, and hence every $f(\boldsymbol{b},\boldsymbol{a})$, is non-negative. By denoting the diagonal elements of $C^{(k)}$ similarly as $\alpha_{\boldsymbol{a}_{1:k}\boldsymbol{b}_{1:k}}$, where $\boldsymbol{a}_{1:k} := a_{1}a_{2}...a_{k}$ and similarly for $\boldsymbol{b}_{1:k}$, we define $f^{(k)}(\boldsymbol{b}_{1:k},\boldsymbol{a}_{1:k}) := \alpha_{\boldsymbol{a}_{1:k}\boldsymbol{b}_{1:k}}$ for all $k = 1, ..., n$ and take $f^{(0)} := C^{(0)}$. Since taking the partial trace and applying $\diag(\cdot)$ commute, that is,
\begin{align}
\diag\left(\Tr_{B_{k}}\left[C^{(k)}\right]\right) = \Tr_{B_{k}}\left[\diag(C^{(k)})\right]
\end{align}

it follows that 
\begin{align}
\Tr_{B_{k}}\left[\diag(C^{(k)})\right] = \diag(I_{A_{k}} \otimes C^{(k-1)}) = I_{A_{k}} \otimes \diag(C^{(k-1)}).
\end{align}

The left-hand side can be written as
\begin{align}
\sum_{\boldsymbol{a}_{1:k}, \boldsymbol{b}_{1:k-1}} \left(\sum_{b_{k}} \alpha_{\boldsymbol{a}_{1:k}\boldsymbol{b}_{1:k}} \right)\ket{\boldsymbol{a}_{1:k}\boldsymbol{b}_{1:k-1}}\!\!\bra{\boldsymbol{a}_{1:k}\boldsymbol{b}_{1:k-1}}_{\boldsymbol{A}_{1:k},\boldsymbol{B}_{1:k-1}} = \sum_{\boldsymbol{a}_{1:k}, \boldsymbol{b}_{1:k-1}} \left(\sum_{b_{k}} f^{(k)}(\boldsymbol{b}_{1:k},\boldsymbol{a}_{1:k}) \right)\ket{\boldsymbol{a}_{1:k}\boldsymbol{b}_{1:k-1}}\!\!\bra{\boldsymbol{a}_{1:k}\boldsymbol{b}_{1:k-1}}_{\boldsymbol{A}_{1:k},\boldsymbol{B}_{1:k-1}}
\end{align}

and the right-hand side as
\begin{align}
\sum_{\boldsymbol{a}_{1:k}, \boldsymbol{b}_{1:k-1}} \alpha_{\boldsymbol{a}_{1:k-1}\boldsymbol{b}_{1:k-1}}\ket{\boldsymbol{a}_{1:k}\boldsymbol{b}_{1:k-1}}\!\!\bra{\boldsymbol{a}_{1:k}\boldsymbol{b}_{1:k-1}}_{\boldsymbol{A}_{1:k},\boldsymbol{B}_{1:k-1}} = \sum_{\boldsymbol{a}_{1:k}, \boldsymbol{b}_{1:k-1}} f^{(k-1)}(\boldsymbol{b}_{1:k-1},\boldsymbol{a}_{1:k-1})\ket{\boldsymbol{a}_{1:k}\boldsymbol{b}_{1:k-1}}\!\!\bra{\boldsymbol{a}_{1:k}\boldsymbol{b}_{1:k-1}}_{\boldsymbol{A}_{1:k},\boldsymbol{B}_{1:k-1}}
\end{align}

which establishes the required marginalisation condition on the $f^{(k)}$: $\sum_{b_{k}} f^{(k)}(\boldsymbol{b}_{1:k}, \boldsymbol{a}_{1:k}) = f^{(k-1)}(\boldsymbol{b}_{1:k-1},\boldsymbol{a}_{1:k-1})$ for $k = 2, ..., n$ as well as the edge case of $\sum_{b_{1}}f^{(1)}(b_{1}, a_{1}) = f^{(0)}$.  
\end{proof}

We can now give the proof of the proposition.

\begin{proof}[Proof of \Cref{prop:bayesian_classical_comb}] The lower bound is established by showing that any positive semi-definite operator $\Gamma$ on $\bigotimes_{i=1}^{n} \H_{A_{i}} \otimes \H_{B_{i}}$ that satisfies $I_{X} \otimes \Gamma \geq D$ must have trace greater than or equal to $\sum_{\boldsymbol{a}, \boldsymbol{b}} \max_{x}P(x|\boldsymbol{a},\boldsymbol{b}) P(\boldsymbol{b}|\boldsymbol{a})$.

Let $\Gamma$ be any positive semi-definite operator on $\bigotimes_{i = 1}^{n} \H_{A_{i}} \otimes \H_{B_{i}}$ such that $I_{X} \otimes \Gamma \geq D$. By \Cref{lem:diag_comb_inequality}, it follows that $\diag(I_{X} \otimes \Gamma) = I_{X} \otimes \diag(\Gamma) \geq D$. We write
\begin{align*}
I_{X} \otimes \diag(\Gamma) = \sum_{x,\boldsymbol{a}, \boldsymbol{b}} \alpha_{\boldsymbol{a}\boldsymbol{b}} \ket{x\boldsymbol{a}\boldsymbol{b}}\!\!\bra{x\boldsymbol{a}\boldsymbol{b}}_{X,\boldsymbol{A},\boldsymbol{B}}
\end{align*}

where the $\alpha_{\boldsymbol{a}\boldsymbol{b}}$ are all real and non-negative by positive semi-definiteness of $\Gamma$. The condition $I_{X} \otimes \diag(\Gamma) \geq D$ induces a further condition on the $\alpha$ terms: we must have for all $\boldsymbol{a}, \boldsymbol{b}$ that, for all $x$,
\begin{align*}
\alpha_{\boldsymbol{a}\boldsymbol{b}} \geq P(x|\boldsymbol{a},\boldsymbol{b})P(\boldsymbol{b}|\boldsymbol{a})
\end{align*}

which in particular enforces that 
\begin{align*}
\alpha_{\boldsymbol{a}\boldsymbol{b}} \geq \max_{x}P(x|\boldsymbol{a},\boldsymbol{b})P(\boldsymbol{b}|\boldsymbol{a}).
\end{align*}

It thus follows that
\begin{align*}
\Tr[\Gamma] \geq \sum_{\boldsymbol{a},\boldsymbol{b}}\max_{x}P(x|\boldsymbol{a},\boldsymbol{b})P(\boldsymbol{b}|\boldsymbol{a}).
\end{align*}

The upper bound is established by showing that
\begin{align}
\Upsilon := \sum_{\boldsymbol{a},\boldsymbol{b}}\max_{x, \boldsymbol{\tilde{a}}}P(x|\boldsymbol{\tilde{a}},\boldsymbol{b})P(\boldsymbol{b}|\boldsymbol{\tilde{a}})\ket{\boldsymbol{a}\boldsymbol{b}}\!\!\bra{\boldsymbol{a}\boldsymbol{b}}_{\boldsymbol{A},\boldsymbol{B}}
\end{align}

is a valid unnormalised comb, since $I_{X} \otimes \Upsilon \geq D$ clearly holds. Defining $f$ via $f(\boldsymbol{b}, \boldsymbol{a}) = \max_{x, \boldsymbol{\tilde{a}}}P(x|\boldsymbol{\tilde{a}},\boldsymbol{b})P(\boldsymbol{b}|\boldsymbol{\tilde{a}})$, non-negativity is immediate. Moreover, due to the maximum over $\boldsymbol{\tilde{a}}$, $f$ has no dependence on $\boldsymbol{a}$: $f(\boldsymbol{b},\boldsymbol{a}) = f(\boldsymbol{b}, \boldsymbol{a'})$ for all $\boldsymbol{a}, \boldsymbol{a'}$. By defining the functions $f^{(k)}$ via
\begin{align}
f^{(k)}(b_{1},...,b_{k}, a_{1},...,a_{k}) := \sum_{b_{n}, ..., b_{k+1}}f(\boldsymbol{b},\boldsymbol{a})
\end{align}

for all $k = 1, ..., n$, the required conditions (non-negativity and the marginalisation conditions - recall \Cref{eq:marginalisation_independence}) are trivially satisfied due to the non-negativity and independence from $\boldsymbol{a}$ exhibited by $f$. The sum over $b_{1}$ of $f^{(1)}$ so defined is also clearly positive since $\max_{x, \boldsymbol{a}}P(x|\boldsymbol{a},\boldsymbol{b})P(\boldsymbol{b}|\boldsymbol{a})$ must be non-zero for some $\boldsymbol{b}$, and so $\Upsilon$ is indeed an unnormalised classical comb.
\end{proof}


\section{Blind Quantum Computing Protocol Supplementary} \label{sec:app_BQC}

This appendix contains supporting results for the blindness theorems of \Cref{sec:CBQC} and for the accompanying examples.


\subsection{Single Round Theorem Supporting Results} \label{subsec:app_BQC_single_round_proof}

In this subsection, we give the proof of \Cref{thm:single_shot_entropy} which makes use of the bounds from \Cref{prop:bayesian_classical_comb}. The proof of the theorem uses a more compact form of $D_{\client}$ which we now present.

The comb we consider is
\begin{align}
D_{\client} = \sum_{\boldsymbol{\alpha}, O} P(\boldsymbol{\alpha}, O)\ket{\boldsymbol{\alpha},O}\!\!\bra{\boldsymbol{\alpha},O} \otimes \sigma_{\boldsymbol{\alpha}, O}
\end{align}

where
\begin{align}
\sigma_{\boldsymbol{\alpha}, O} = \sum_{\substack{g \sim O,\boldsymbol{r}\\ \boldsymbol{\alpha'}, \boldsymbol{c'} }} P(g|O)P(\boldsymbol{r})P(\boldsymbol{\alpha'}|\boldsymbol{c'},\boldsymbol{\alpha},\boldsymbol{r},g)\ket{\boldsymbol{\alpha'}\boldsymbol{c'}}\!\!\bra{\boldsymbol{\alpha'}\boldsymbol{c'}}_{\boldsymbol{A'},\boldsymbol{C'}}.
\end{align}

Note the labelling is the same as in the main text, namely with $\H_{\boldsymbol{A'}, \boldsymbol{C'}} := \bigotimes_{i=1}^{n} \H_{A'_{i}}\otimes \H_{C'_{i}}$ and $\H_{\boldsymbol{A},\boldsymbol{O}}$ the Hilbert space with basis states given by the values of the joint classical variables $\boldsymbol{A},\boldsymbol{O}$ that specify the computation. Since $\boldsymbol{\alpha'}$ is conditionally independent of $O$ given $g$, we have
\begin{align}
P(\boldsymbol{\alpha'}|\boldsymbol{c'},\boldsymbol{\alpha},\boldsymbol{r},g) &= P(\boldsymbol{\alpha'}|\boldsymbol{c'}, \boldsymbol{\alpha}, \boldsymbol{r}, g, O). \label{eq:alpha_prime_indep}
\end{align}

Furthermore, since $\boldsymbol{r}$ is chosen independently of all variables (\cite[Lemma 4,][]{mantri2017flow}) and $g$ is independent of $\boldsymbol{\alpha}$ and $\boldsymbol{c'}$ (it is a valid gflow for all computations and measurement outcomes), we also have
\begin{align}
P(\boldsymbol{r})P(g| O) = P(\boldsymbol{r}, g|\boldsymbol{c'}, \boldsymbol{\alpha}, O). \label{eq:rg_indep}
\end{align}

From \Cref{eq:alpha_prime_indep} and \Cref{eq:rg_indep} and by summing over $g, \boldsymbol{r}$, we get
\begin{align}
\sigma_{\boldsymbol{\alpha}, O} = \sum_{\boldsymbol{\alpha'}, \boldsymbol{c'}}  P(\boldsymbol{\alpha'}| \boldsymbol{c'}, \boldsymbol{\alpha}, O) \ket{\boldsymbol{\alpha'}\boldsymbol{c'}}\!\!\bra{\boldsymbol{\alpha'}\boldsymbol{c'}}
\end{align}

which allows us to write $D_{\client}$ as 
\begin{align}
D_{\client} = \sum_{\substack{\boldsymbol{\alpha},O\\\boldsymbol{\alpha'},\boldsymbol{c'}}} P(\boldsymbol{\alpha},O)P(\boldsymbol{\alpha'}| \boldsymbol{c'}, \boldsymbol{\alpha}, O) \ket{\boldsymbol{\alpha},O,\boldsymbol{\alpha'},\boldsymbol{c'}}\!\!\bra{\boldsymbol{\alpha},O,\boldsymbol{\alpha'},\boldsymbol{c'}}_{\boldsymbol{A},\boldsymbol{O},\boldsymbol{A'},\boldsymbol{C'}}. \label{eq:pre_Bayesian_D_client}
\end{align}

The Bayesian updated form of $D_{\client}$ is thus
\begin{align}
D_{\client} = \sum_{\substack{\boldsymbol{\alpha},O\\\boldsymbol{\alpha'},\boldsymbol{c'}}} P(\boldsymbol{\alpha}, O | \boldsymbol{\alpha'}, \boldsymbol{c'}) P(\boldsymbol{\alpha'}|\boldsymbol{c'}) \ket{\boldsymbol{\alpha},O,\boldsymbol{\alpha'},\boldsymbol{c'}}\!\!\bra{\boldsymbol{\alpha},O,\boldsymbol{\alpha'},\boldsymbol{c'}}_{\boldsymbol{A},\boldsymbol{O},\boldsymbol{A'},\boldsymbol{C'}}
\end{align}

however, the form in \Cref{eq:pre_Bayesian_D_client} is ultimately the one used in the proof of the theorem. 

The proof of \Cref{thm:single_shot_entropy} makes use of the following lemma which demonstrates that for $\boldsymbol{\alpha'}, \boldsymbol{\alpha}, \boldsymbol{c'}$ and $g$ fixed, there is at most one $\boldsymbol{r}$ such that \Cref{eq:adapt_meas_angles} is satisfied for each $i$.

\begin{Lemma} \label{lem:preimage_gflow} If there exists an $\boldsymbol{r}$ such that $P(\boldsymbol{\alpha'}| \boldsymbol{c'}, \boldsymbol{\alpha}, \boldsymbol{r}, g) = 1$ for all other variables fixed, then it is unique, otherwise $P(\boldsymbol{\alpha'}| \boldsymbol{c'}, \boldsymbol{\alpha}, \boldsymbol{r}, g) = 0$.
\end{Lemma}

\begin{proof} Since $P(\boldsymbol{\alpha'}| \boldsymbol{c'}, \boldsymbol{\alpha}, \boldsymbol{r}, g)$ is a deterministic distribution, it takes values either $0$ or $1$. Let $\boldsymbol{r}$ be such that $P(\boldsymbol{\alpha'}| \boldsymbol{c'}, \boldsymbol{\alpha}, \boldsymbol{r}, g) = 1$. Recalling \Cref{eq:adapt_meas_angles} and the fact that a total order is imposed on communication, we have that for each $i$
\begin{align}
\alpha'_{i} := (-1)^{\bigoplus_{j \in \mathcal{X}_{i}} c'_{j} \oplus r_{j} }\alpha_{i} + \left(r_{i} \oplus \bigoplus_{j \in \mathcal{Z}_{i}} c'_{j} \oplus r_{j} \right)\pi \bmod 2\pi \label{eq:lem_unique_r}
\end{align}

where $\mathcal{X}_{i}$ and $\mathcal{Z}_{i}$ are necessarily subsets of $\{1,..., i-1\}$. In particular, this means that
\begin{align}
\alpha'_{1} = \alpha_{1} + r_{1}\pi \bmod 2\pi
\end{align}

and hence there is a unique value for $r_{1}$ given $\boldsymbol{\alpha}$ and $\boldsymbol{\alpha'}$ are fixed. The result then follows by induction: since $\boldsymbol{c'}$ and $g$ are fixed, for fixed values $\boldsymbol{r}_{1:i-1}$ the quantities $\bigoplus_{j \in \mathcal{X}_{i}} c'_{j} \oplus r_{j}$ and $\bigoplus_{j \in \mathcal{Z}_{i}} c'_{j} \oplus r_{j}$ are determined for each $i$, and thus there is a unique $r_{i}$ for which \Cref{eq:lem_unique_r} holds. 
\end{proof}

We can now give the proof of the theorem.

\begin{proof}[Proof of \Cref{thm:single_shot_entropy}] From \Cref{prop:bayesian_classical_comb}, we know that
\begin{align}
H_{\min}(\boldsymbol{A}, \boldsymbol{O}| \boldsymbol{A'}, \boldsymbol{C'})_{D_{\client}} \geq -\log\left[\sum_{\boldsymbol{\alpha'}}\max_{\boldsymbol{\alpha}, O,\boldsymbol{c'}}P(\boldsymbol{\alpha}, O | \boldsymbol{\alpha'}, \boldsymbol{c'}) P(\boldsymbol{\alpha'}|\boldsymbol{c'}) \right].
\end{align}

Instead of dealing directly with $P(\boldsymbol{\alpha}, O | \boldsymbol{\alpha'}, \boldsymbol{c'})P(\boldsymbol{\alpha'}|\boldsymbol{c'})$, it is easier to revert back to the form before the Bayes' rule was applied. That is, we aim to find
\begin{align}
\sum_{\boldsymbol{\alpha'}}\max_{\boldsymbol{\alpha}, O,\boldsymbol{c'}} P(\boldsymbol{\alpha}, O)P(\boldsymbol{\alpha'}|\boldsymbol{c'}, \boldsymbol{\alpha}, O). \label{eq:non_Bayes_sum}
\end{align} 

Using the uniformity assumptions for choosing the computation, gflow and one-time pads, as well as the definition of $P(\boldsymbol{\alpha'}|\boldsymbol{c'}, \boldsymbol{\alpha}, O)$ in the previous subsection, we have
\begin{align}
P(\boldsymbol{\alpha}, O)P(\boldsymbol{\alpha'}|\boldsymbol{c'}, \boldsymbol{\alpha}, O) = \frac{1}{|\mathcal{A}|^{n}|\mathcal{O}|}\sum_{g \sim O, \boldsymbol{r}} \frac{P(\boldsymbol{\alpha'}| \boldsymbol{c'}, \boldsymbol{\alpha}, \boldsymbol{r}, g)}{|g \sim O|2^{n}}
\end{align}

where $g \sim O$ indicates that the gflow is defined for the output set $O$ and $|g \sim O|$ denotes the number of all such gflows. Recalling that $P(\boldsymbol{\alpha'}| \boldsymbol{c'}, \boldsymbol{\alpha}, \boldsymbol{r}, g)$ is a deterministic distribution, it can be shown (\Cref{lem:preimage_gflow}) that for fixed $g$, $\boldsymbol{\alpha}$ and $\boldsymbol{c'}$ there is at most one $\boldsymbol{r}$ for which $P(\boldsymbol{\alpha'}| \boldsymbol{c'}, \boldsymbol{\alpha}, \boldsymbol{r}, g) = 1$. Thus,
\begin{align}
\max_{\boldsymbol{\alpha}, O, \boldsymbol{c'}} \sum_{g \sim O, \boldsymbol{r}} \frac{P(\boldsymbol{\alpha'}| \boldsymbol{c'}, \boldsymbol{\alpha}, \boldsymbol{r}, g)}{|g \sim O|2^{n}} \label{eq:max_reportable_alpha_O}
\end{align}

can be interpreted as selecting the $\boldsymbol{\alpha}$, $O$ and $\boldsymbol{c'}$ for which $\boldsymbol{\alpha'}$ is reportable from $\boldsymbol{\alpha}$ for the greatest number of pairs $(g, \boldsymbol{r})$ for gflows $g \sim O$. This quantity is clearly upper-bounded by a situation where $\boldsymbol{\alpha'}$ is reportable under all gflows, hence
\begin{align}
\max_{\boldsymbol{\alpha}, O, \boldsymbol{c'}} \sum_{g \sim O, \boldsymbol{r}} \frac{P(\boldsymbol{\alpha'}| \boldsymbol{c'}, \boldsymbol{\alpha}, \boldsymbol{r}, g)}{|g \sim O|2^{n}} &\leq \frac{1}{2^{n}} \label{eq:max_reportable_upper_bound} 
\end{align}

which holds for all $\boldsymbol{\alpha'}$. Thus, returning to \Cref{eq:non_Bayes_sum}:
\begin{align}
\sum_{\boldsymbol{\alpha'}}\max_{\boldsymbol{\alpha}, O,\boldsymbol{c'}} P(\boldsymbol{\alpha}, O)P(\boldsymbol{\alpha'}|\boldsymbol{c'}, \boldsymbol{\alpha}, O) &\leq \sum_{\boldsymbol{\alpha'}}\frac{1}{|\mathcal{A}|^{n}|\mathcal{O}|2^{n}} \\
&= \frac{1}{|\mathcal{O}|2^{n}}
\end{align}

which proves the theorem.
\end{proof}

It is worthwhile making some further comments regarding the interpretation of \Cref{eq:max_reportable_alpha_O} and the related inequality \Cref{eq:max_reportable_upper_bound} in the proof above. Firstly, it is possible to find a simple example, namely that given in the main text (see also below), for which \Cref{eq:max_reportable_upper_bound} is equality for every $\boldsymbol{\alpha'}$. That is, for every $\boldsymbol{\alpha'}$, there exists an $\boldsymbol{\alpha}$, $O$ and $\boldsymbol{c'}$ that $\boldsymbol{\alpha'}$ is reportable from $\boldsymbol{\alpha}$ for every gflow $g \sim O$. Said another way, the pre-images of $\boldsymbol{\alpha'}$ under the gflows $g \sim O$ (for some fixed $\boldsymbol{c'}$ and as $\boldsymbol{r}$ varies) have non-empty mutual intersection. Since the pre-image of each gflow has a fixed size (this follows from \Cref{lem:preimage_gflow}), a larger mutual intersection corresponds to a smaller total set of angles $\boldsymbol{\alpha}$ from which $\boldsymbol{\alpha'}$ can be reported. Since $\boldsymbol{\alpha}$ is one part of the secret information, it is intuitive that a smaller set of possible $\boldsymbol{\alpha}$ given the evidence $\boldsymbol{\alpha'}, \boldsymbol{c'}$ corresponds to a lower min-entropy. This also suggests that, to improve security, considering a graph where the corresponding gflows have smaller mutual intersection is beneficial.

Obtaining equality in \Cref{eq:max_reportable_upper_bound} for all $\boldsymbol{\alpha'}, \boldsymbol{c'}$ accordingly means that equality is also achieved in \Cref{eq:thm_single_round}, however, it is unlikely that this bound is obtained in general. The above discussion highlights how much is dependent on the specific properties of gflows for the graph chosen for the protocol. Since (to the best knowledge of the authors) there is no characterisation of gflows in terms of graph-theoretic properties, and since the given example achieves the lower bound, it is unlikely that a better lower bound that that given in the theorem can be given in general.


\subsection{BQC Minimum Example} \label{subsec:app_BQC_min_example}

This appendix provides further details regarding the minimal example given in the main text (recall \Cref{fig:min_BQC_example}) which obtains the lower bound in \Cref{thm:single_shot_entropy}. We begin by demonstrating that $O = \{2,3\}$ is indeed the only choice of output set for which gflows compatible with the total order exist.

\begin{Lemma} \label{lem:gflows_min_example} For $G$ as in \Cref{fig:min_BQC_example} with partial order $1<2<3$ and the $XY$-plane the only allowed measurement plane, $O = \{2,3\}$ is the only non-trivial output set for which there exist $I \subset \{1,2,3\}$ such that a gflow compatible with the total order exists. Moreover, there are just two possible gflows, $g_{1}$ and $g_{2}$, defined by
\begin{align}
g_{1}: 1 &\mapsto \{2\}, \\
g_{2}:1 &\mapsto \{3\},
\end{align}

where $g_{1}$ is compatible with $(I, O)$ for $I$ equal to $\{1\}$, $\{3\}$ or $\{1,3\}$ and $g_{2}$ is compatible with $(I,O)$ for $I = \{1\}$, $\{2\}$ and $\{1,2\}$, for $O$ as above.
\end{Lemma}

\begin{proof} We begin by demonstrating $O = \{2,3\}$ is the only valid non-trivial output set (note that we do not consider the trivial output set $O = \{1,2,3\}$ since this does not allow for any computation). Due to the total order and measurement plane restriction, $3$ must be in the output set, since if this was not the case, then $3 \in O^{c}$ meaning that any gflow would map $3$ to $\{1\}$, $\{2\}$ or $\{1,2\}$. In any of these cases, we would require that $3<1$ or $3<2$ which contradicts the compatibility with total order. Thus $\{1\}, \{2\}$ and $\{1,2\}$ cannot be output sets. 

Suppose either $\{3\}$ or $\{1,3\}$ was a valid output set. Any gflow compatible with either output must then map $2$ to some subset of $\{1,2,3\}$, however no such set exists for which some contradiction does not arise, as follows. If $2$ maps to $\{1\}$, $2 < 1$ which is a contradiction to the total order. If $2$ maps to $\{2\}$, $\{1,2\}$ or $\{2,3\}$, the gflow requirement $2 \notin g(2)$ is contradicted. If $2$ maps to $\{3\}$, $1 \in \Odd(g(2))$ implying $2<1$, contradicting the total order. If $2$ maps to $\{1,3\}$, then $2 \notin \Odd(g(2))$ which contradicts a gflow requirement.

This leaves $\{2,3\}$ as the only remaining possible output set. We now show it does indeed support gflow and characterise them. For output set $O = \{2,3\}$, the domain of any gflow map is $O^{c} = \{1\}$, so we can begin to characterise valid gflows by where they map $1$. Due to the $XY$-plane restriction, we cannot have $1 \in g(1)$, so a valid gflow cannot map $1$ to $\{1\}, \{1,2\}$ or $\{1,3\}$. Since $\Odd(\{2,3\}) = \{2,3\}$ and since we require $1 \in \Odd(g(1))$, no valid gflow maps $1$ to $\{2,3\}$. This leaves just the options $1 \mapsto \{2\}$ and $1 \mapsto \{3\}$. Both of these are valid gflows since in both cases, $1 \notin g(1)$ is satisfied and the corresponding implications for the partial order, $1<2$ and $1<3$ respectively, are compatible with the total order. Similarly, for both maps $1 \in \Odd(g(1))$ is satisfied ($\Odd(\{2\}) = \{1,3\}$ and $\Odd(\{3\}) = \{1,2\}$) and the again the implications for the partial order are compatible with the total order.

Denote by $g_{1}$ the gflow that maps $1$ to $\{2\}$ and by $g_{2}$ the gflow that maps $1$ to $\{3\}$. Since a gflow is a map from $O^{c}$ to $\mathcal{P}(I^{c})$, $g_{1}$ is compatible with all sets $I$ for which $\{2\} \in \mathcal{P}(I^{c})$, i.e. $I = \{1\}, \{3\}$ and $\{1,3\}$, and similarly, $g_{2}$ is compatible with $I = \{1\}, \{2\}$ and $\{1,2\}$.
\end{proof}

Returning to the discussion of the BQC protocol for the example, it is useful to explicitly write out the correction sets given by $g_{1}$ and $g_{2}$ (only those that are non-empty are shown):
\begin{align*}
\mathcal{X}_{2}^{g_{1}} &= \mathcal{Z}_{3}^{g_{1}} = \{1\}; \\
\mathcal{Z}_{2}^{g_{2}} &= \mathcal{X}_{3}^{g_{2}} = \{1\}.
\end{align*}

Let $\mathcal{A}$ be any agreed upon set of angles that satisfies \Cref{eq:angle_set}. For a single round of the protocol, if $g_{1}$ is chosen by the client, then the reported angles are given by
\begin{align}
\alpha'_{1} &= \alpha_{1} + r_{1}\pi \bmod 2\pi, \label{eq:g_alpha_one} \\
\alpha'_{2} &= (-1)^{r_{1} \oplus c'_{1}}\alpha_{2} + r_{2}\pi \bmod 2\pi, \label{eq:g_alpha_two} \\
\alpha'_{3} &= \alpha_{3} + (r_{3} \oplus r_{1} \oplus c'_{1})\pi \bmod 2\pi, \label{eq:g_alpha_three}
\end{align}

where $\alpha_{1}, \alpha_{2}, \alpha_{3} \in \mathcal{A}$ are the chosen (and secret) angles for the computation, $c'_{1}$ is the classical message reported by the server after the first measurement and $\boldsymbol{r} = r_{3}r_{2}r_{1} \in \{0,1\}^{3}$ is the one-time pad. If $g_{2}$ is used instead, the angles are reported as
\begin{align}
\alpha'_{1} &= \alpha_{1} + r_{1}\pi \bmod 2\pi, \label{eq:g_prime_alpha_one}\\
\alpha'_{2} &= \alpha_{2} +  (r_{2} \oplus r_{1} \oplus c'_{1})\pi \bmod 2\pi, \label{eq:g_prime_alpha_two}\\
\alpha'_{3} &= (-1)^{r_{1} \oplus c'_{1}}\alpha_{3} + r_{3}\pi \bmod 2\pi. \label{eq:g_prime_alpha_three}
\end{align}

Note that the only classical message relevant to this example is $c'_{1}$. Let $\boldsymbol{\alpha'}$ and $\boldsymbol{c'}$ be reported in a single round of the protocol. We can write the pre-image of $\boldsymbol{\alpha'}$ under $g_{1}$ given $\boldsymbol{c'}$ by inverting \Cref{eq:g_alpha_one,eq:g_alpha_two,eq:g_alpha_three} (we drop the $\bmod 2\pi$ notation and leave it implicit):
\begin{align}
\alpha_{1} &= \alpha'_{1} + r_{1}\pi; \\
\alpha_{2} &= (-1)^{r_{1} \oplus c'_{1}}\alpha'_{2} + r_{2}\pi; \\
\alpha_{3} &= \alpha'_{3} + (r_{3} \oplus r_{1} \oplus c'_{1})\pi.
\end{align}

Restricting our focus to a subset of the pre-image defined by $r_{1} = c'_{1}$, that is, the angles
\begin{align}
\alpha_{1} &= \alpha'_{1} + c'_{1}\pi, \\
\alpha_{2} &= \alpha'_{2} + r_{2}\pi, \\
\alpha_{3} &= \alpha'_{3} + r_{3}\pi,
\end{align}

one observes that $\boldsymbol{\alpha'}$ can be reported from any one of these angles under $g_{2}$ given $\boldsymbol{c'}$, namely by the one-time pads $\boldsymbol{\hat{r}} = \hat{r}_{3}\hat{r}_{2}\hat{r}_{1} = r_{3}r_{2}c'_{1}$ (as can be shown via simple substitution into \Cref{eq:g_prime_alpha_one,eq:g_prime_alpha_two,eq:g_prime_alpha_three}). Since this holds for any $\boldsymbol{\alpha'}$ and $\boldsymbol{c'}$, we have thus shown both that the bound \Cref{eq:max_reportable_upper_bound} in the proof of \Cref{thm:single_shot_entropy} above is in fact equality for every $\boldsymbol{\alpha'}$ and moreover that the maximum is obtained for every $\boldsymbol{c'}$. Recalling the discussion after \Cref{prop:bayesian_classical_comb}, this means that we have found the guessing probability (equivalently min-entropy) exactly for this example and that any strategy used by the server is equally informative. These results are corroborated numerically (see \Cref{sec:app_SDP_implementation} for details of the implementation using an SDP solver): the min-entropy is given as $-\log_{2}(0.125) = 3$, which is consistent with $|\mathcal{O}| = 1$ and $n = 3$ for this example (this result was obtained for different choices of $\mathcal{A}$: for $|\mathcal{A}| = 4$ and for $|\mathcal{A}| = 8$). 

We now consider two rounds of the protocol in order to calculate the guessing probability for $D_{\client}^{(2)}$, which, for this example, can be written as
\begin{align}
D_{\client}^{(2)} &= \sum_{\boldsymbol{\alpha}, O} \frac{1}{|\mathcal{A}|^{n}|\mathcal{O}|}\ket{\boldsymbol{\alpha}, O}\!\!\bra{\boldsymbol{\alpha}, O} \otimes \sigma_{\boldsymbol{\alpha}, O}^{(1)} \otimes \sigma_{\boldsymbol{\alpha}, O}^{(2)} \\
&= \sum_{\boldsymbol{\alpha}} \frac{1}{|\mathcal{A}|^{3}}\ket{\boldsymbol{\alpha}}\!\!\bra{\boldsymbol{\alpha}} \otimes \sum_{\substack{\boldsymbol{\alpha'}^{(1)}, \boldsymbol{\alpha'}^{(2)}, \\ \boldsymbol{c'}^{(1)}, \boldsymbol{c'}^{(2)},\\g^{(1)}, g^{(2)}, \\ \boldsymbol{r}^{(1)}, \boldsymbol{r}^{(2)}}}  \frac{1}{2^{8}} P(\boldsymbol{\alpha'}^{(1)}|\boldsymbol{c'}^{(1)}, \boldsymbol{\alpha}, \boldsymbol{r}^{(1)}, g^{(1)})P(\boldsymbol{\alpha'}^{(2)}|\boldsymbol{c'}^{(2)}, \boldsymbol{\alpha}, \boldsymbol{r}^{(2)}, g^{(2)}) \ket{\boldsymbol{\alpha'}^{(1:2)}\boldsymbol{c'}^{(1:2)}}\!\!\bra{\boldsymbol{\alpha'}^{(1:2)}\boldsymbol{c'}^{(1:2)}} \label{eq:two_rounds_example}
\end{align}

where we have used the shorthand notation for $\ket{\boldsymbol{\alpha'}^{(1:2)}\boldsymbol{c'}^{(1:2)}}$ for states on $\H_{\boldsymbol{A'}^(1),\boldsymbol{C'}^{(1)}} \otimes \H_{\boldsymbol{A'}^(2),\boldsymbol{C'}^{(2)}}$, where $g^{(j)}$ take values in $\{g_{1}, g_{2}\}$ with $g_{1}$ and $g_{2}$ are as above, and where we have used $P(g^{(j)}) = \frac{1}{2}$ and $P(\boldsymbol{r}^{(j)}) = \frac{1}{2^{3}}$. For a choice of angle set $\mathcal{A} = \{\frac{\pi}{5}, \frac{\pi}{3}, \frac{-\pi}{3} + \pi, \frac{-\pi}{5} + \pi, \frac{\pi}{5} + \pi, \frac{\pi}{3} + \pi, \frac{-\pi}{3}, \frac{-\pi}{5} \}$, the guessing probability for $D_{\client}^{(2)}$ takes value between $0.140625$ and $0.28125$ as computed via numeric methods (see \cite{Smith_Min_Entropy_and_MBQC}) and as explained in the following. Note that various other choices of angle set give the same or similar results.

These bounds are derived directly from \Cref{prop:bayesian_classical_comb} for the specific comb under consideration. The lower bound (which includes the normalisation over the input spaces) can be written as
\begin{align}
\sum_{\substack{\boldsymbol{\alpha'}^{(1)}, \boldsymbol{\alpha'}^{(2)}, \\ \boldsymbol{c'}^{(1)}, \boldsymbol{c'}^{(2)}}} \frac{1}{|\mathcal{A}|^{3}2^{6}} \max_{\boldsymbol{\alpha}}\sum_{\substack{g^{(1)}, g^{(2)}, \\ \boldsymbol{r}^{(1)}, \boldsymbol{r}^{(2)}}}\frac{1}{2^{8}} P(\boldsymbol{\alpha'}^{(1)}|\boldsymbol{c'}^{(1)}, \boldsymbol{\alpha}, \boldsymbol{r}^{(1)}, g^{(1)})P(\boldsymbol{\alpha'}^{(2)}|\boldsymbol{c'}^{(2)}, \boldsymbol{\alpha}, \boldsymbol{r}^{(2)}, g^{(2)}).
\end{align}

Due to the redundancy of $(c'_{2})^{(j)}$ and $(c'_{3})^{(j)}$ (they don't appear in \Cref{eq:g_alpha_one,eq:g_alpha_two,eq:g_alpha_three,eq:g_prime_alpha_one,eq:g_prime_alpha_two,eq:g_prime_alpha_three}), this reduces to
\begin{align}
\sum_{\substack{\boldsymbol{\alpha'}^{(1)}, \boldsymbol{\alpha'}^{(2)}, \\ (c'_{1})^{(1)}, (c'_{1})^{(2)}}} \frac{1}{|\mathcal{A}|^{3}2^{2}} &\max_{\boldsymbol{\alpha}}\sum_{\substack{g^{(1)}, g^{(2)}, \\ \boldsymbol{r}^{(1)}, \boldsymbol{r}^{(2)}}}\frac{1}{2^{8}} P(\boldsymbol{\alpha'}^{(1)}|(c'_{1})^{(1)}, \boldsymbol{\alpha}, \boldsymbol{r}^{(1)}, g^{(1)})P(\boldsymbol{\alpha'}^{(2)}|(c'_{1})^{(2)}, \boldsymbol{\alpha}, \boldsymbol{r}^{(2)}, g^{(2)}). \label{eq:example_lower_bound}
\end{align}

The upper bound is given by
\begin{align}
\sum_{\boldsymbol{\alpha'}^{(1)}, \boldsymbol{\alpha'}^{(2)}} \frac{1}{|\mathcal{A}|^{3}} &\max_{\boldsymbol{\alpha},(c'_{1})^{(1)}, (c'_{1})^{(2)}}\sum_{\substack{g^{(1)}, g^{(2)}, \\ \boldsymbol{r}^{(1)}, \boldsymbol{r}^{(2)}}}\frac{1}{2^{8}} P(\boldsymbol{\alpha'}^{(1)}|(c'_{1})^{(1)}, \boldsymbol{\alpha}, \boldsymbol{r}^{(1)}, g^{(1)})P(\boldsymbol{\alpha'}^{(2)}|(c'_{1})^{(1)}, \boldsymbol{\alpha}, \boldsymbol{r}^{(2)}, g^{(2)}).
\end{align}

Computing these quantities for $\mathcal{A}$ above via the code in the repository \cite{Smith_Min_Entropy_and_MBQC} gives the results as listed above.

\subsection{Multi-round Blindness Supplementary} \label{subsec:app_multi_round_blind}

For this section, it is useful to have the following notation. Let a graph $G$ on vertices $V$ be given and $O \subset V$ a choice of output set. For a given $\boldsymbol{\alpha} \in \mathcal{A}^{n}$, define the set $\mathcal{A}_{\boldsymbol{\alpha}, O} := \{\boldsymbol{\alpha} + (0,...,0,b_{o_{1}}, ..., b_{o_{|O|}})\pi \bmod 2\pi : b_{o_{i}} \in \{0,1\}\}$ where the zero entries of $(0,...,0,b_{o_{1}}, ..., b_{o_{|O|}})$ correspond to $V \setminus O$ and the labels $o_{i}$ correspond to the elements of $O$. We define an equivalence relation $\sim_{O}$ by $\boldsymbol{\alpha} \sim_{O} \boldsymbol{\widetilde{\alpha}}$ for $\widetilde{\alpha} \in \mathcal{A}_{\boldsymbol{\alpha}, O}$. The set obtained by quotienting $\mathcal{A}^{n}$ by the equivalence relation, i.e. $\mathcal{A}^{n}/\sim_{O}$ has $\frac{|\mathcal{A}|^{n}}{2^{|O|}}$ elements.

\begin{Lemma} \label{lem:output_symmetry} Consider $\sum_{g \sim O, \boldsymbol{r}} P(g|O)p(\boldsymbol{r})\sigma_{\BQC}^{\boldsymbol{\alpha}, \boldsymbol{r}, g}$ as in \Cref{eq:BQC_comb_not_all_comps} with $P(g|O)$ and $P(\boldsymbol{r})$ both uniform. Then, for all $\boldsymbol{\widetilde{\alpha}} \sim_{O} \boldsymbol{\alpha}$:
\begin{align}
\sum_{g \sim O, \boldsymbol{r}} P(g|O)P(\boldsymbol{r})\sigma_{\BQC}^{\boldsymbol{\alpha}, \boldsymbol{r}, g} = \sum_{g \sim O, \boldsymbol{r}} P(g|O)P(\boldsymbol{r})\sigma_{\BQC}^{\boldsymbol{\widetilde{\alpha}}, \boldsymbol{r}, g}
\end{align}
\end{Lemma}

\begin{proof} We begin by noting that, by definition of gflow, for any $g \sim O$, no element of $O$ is in the domain of $g$. This means that, for all $i \in O$, $i \notin \mathcal{X}_{k}$ and $i \notin \mathcal{Z}_{k}$ for all $k \in V$. The consequence of this being that for any $r_{i}$ for $i \in O$, the only equation for $\alpha'_{k}$ (recall \Cref{eq:adapt_meas_angles}) that contains $r_{i}$ is when $k = i$. This results in the following symmetry:
\begin{align}
\alpha'_{i} = (-1)^{\bigoplus_{j \in \mathcal{X}_{i}} c'_{j} \oplus r_{j}}\alpha_{i} +(r_{i} \oplus \bigoplus_{j \in \mathcal{Z}_{i}} c'_{j} \oplus r_{j})\pi \bmod 2\pi = (-1)^{\bigoplus_{j \in \mathcal{X}_{i}} c'_{j} \oplus r_{j}}(\alpha_{i} + \pi) +((r_{i} \oplus 1) \oplus \bigoplus_{j \in \mathcal{Z}_{i}} c'_{j} \oplus r_{j})\pi \bmod 2\pi
\end{align}

Thus, for fixed gflow $g$ and classical messages $\boldsymbol{c'}$, if $\boldsymbol{\alpha'}$ is reportable from $\boldsymbol{\alpha}$ for one-time pads $\boldsymbol{r}$, i.e. $P(\boldsymbol{\alpha'}|\boldsymbol{c'}, \boldsymbol{\alpha}, \boldsymbol{r}, g) = 1$, then $\boldsymbol{\alpha'}$ is reportable from $\boldsymbol{\alpha} + (0,...,0, b_{o_{1}},...,b_{o_{|O|}})\pi$ for one-time pads $\boldsymbol{r} + (0,...,0, b_{o_{1}},...,b_{o_{|O|}})\pi$, i.e. $P(\boldsymbol{\alpha'}|\boldsymbol{c'}, \boldsymbol{\alpha}+ (0,...,0, b_{o_{1}},...,b_{o_{|O|}})\pi, \boldsymbol{r}+ (0,...,0, b_{o_{1}},...,b_{o_{|O|}})\pi, g) = 1$. In light of 
\Cref{lem:preimage_gflow}, and with the assumption that $P(\boldsymbol{r})$ is uniform, this in particular means that
\begin{align}
\sum_{\boldsymbol{r}} P(\boldsymbol{\alpha'}|\boldsymbol{c'}, \boldsymbol{\alpha}, \boldsymbol{r}, g)P(\boldsymbol{r}) = \sum_{\boldsymbol{r}} P(\boldsymbol{\alpha'}|\boldsymbol{c'}. \boldsymbol{\alpha} + (0,...,0, b_{o_{1}},...,b_{o_{|O|}})\pi, \boldsymbol{r}, g)P(\boldsymbol{r}) 
\end{align}

and so
\begin{align}
\sum_{g \sim O, \boldsymbol{r}} P(g|O)P(\boldsymbol{r})\sigma_{\BQC}^{\boldsymbol{\alpha}, \boldsymbol{r}, g} &= \sum_{\substack{g \sim O, \boldsymbol{r}, \\ \boldsymbol{\alpha'}, \boldsymbol{c'}}} P(g|O)P(\boldsymbol{r})P(\boldsymbol{\alpha'}|\boldsymbol{c'}, \boldsymbol{\alpha}, \boldsymbol{r}, g) \ket{\boldsymbol{\alpha'}\boldsymbol{c'}}\!\!\bra{\boldsymbol{\alpha'}\boldsymbol{c'}} \\
&= \sum_{\substack{g \sim O, \boldsymbol{r}, \\ \boldsymbol{\alpha'}, \boldsymbol{c'}}} P(g|O)P(\boldsymbol{r})P(\boldsymbol{\alpha'}|\boldsymbol{c'}, \boldsymbol{\alpha} + (0,...,0, b_{o_{1}},...,b_{o_{|O|}})\pi, \boldsymbol{r}, g) \ket{\boldsymbol{\alpha'}\boldsymbol{c'}}\!\!\bra{\boldsymbol{\alpha'}\boldsymbol{c'}} \\
&= \sum_{g \sim O, \boldsymbol{r}} P(g|O)P(\boldsymbol{r})\sigma_{\BQC}^{\boldsymbol{\alpha} + (0,...,0, b_{o_{1}},...,b_{o_{|O|}})\pi , \boldsymbol{r}, g} 
\end{align}
\end{proof}

The above lemma states that $\sigma_{\boldsymbol{\alpha}, O} = \sigma_{\boldsymbol{\widetilde{\alpha}}, O}$ for all $\widetilde{\alpha} \sim_{O} \boldsymbol{\alpha}$ which is the notation used in the following proof.

\begin{proof}[Proof of \Cref{thm:multi_round}] For this proof, we use the formulation of the comb min-entropy as expressed in \Cref{eq:maximum_pairs}. That is, we consider 
\begin{align}
\max_{\widehat{E}} \Tr_{\boldsymbol{A}, \boldsymbol{O}, \boldsymbol{\widehat{A}}, \boldsymbol{\widehat{O}}}\left[\left(D_{\client}^{(m)}\widehat{E}^{\top}\right)\left(\sum_{\boldsymbol{\alpha}, O} \ket{\boldsymbol{\alpha},O,\boldsymbol{\alpha},O}\!\!\bra{\boldsymbol{\alpha},O,\boldsymbol{\alpha},O}_{\boldsymbol{A},\boldsymbol{O},\boldsymbol{\widehat{A}}, \boldsymbol{\widehat{O}}}\right) \right] 
\end{align}

where $\widehat{E}$ is a probabilistic comb in $\Comb(\mathbb{C} \rightarrow \mathbb{C}, (A'_{1})^{(1)} \rightarrow (C'_{1})^{(1)}, ..., (A'_{n})^{(m)} \rightarrow (C'_{n})^{(m)}, \mathbb{C} \rightarrow \boldsymbol{\widehat{A}}, \boldsymbol{\widehat{O}})$. The transpose $\top$ is over all spaces $\boldsymbol{A'}^{(1)}, \boldsymbol{C'}^{(1)}, ..., \boldsymbol{A'}^{(n)}, \boldsymbol{C'}^{(n)}$. 

Let $\widehat{E}$ be such a probabilistic comb. Writing $D_{\client}^{(m)}$ as $\sum_{\boldsymbol{\alpha}, O}P(\boldsymbol{\alpha}, O)\ket{\boldsymbol{\alpha}, O}\!\!\bra{\boldsymbol{\alpha}, O} \bigotimes_{j = 1}^{m} \sigma_{\boldsymbol{\alpha}, O}^{(j)}$, we get
\begin{align}
\Tr_{\boldsymbol{A}, \boldsymbol{O}, \boldsymbol{\widehat{A}}, \boldsymbol{\widehat{O}}}&\left[\left(D_{\client}^{(m)}\widehat{E}^{\top}\right)\left(\sum_{\boldsymbol{\alpha}, O} \ket{\boldsymbol{\alpha},O,\boldsymbol{\alpha},O}\!\!\bra{\boldsymbol{\alpha},O,\boldsymbol{\alpha},O}_{\boldsymbol{A},\boldsymbol{O},\boldsymbol{\widehat{A}}, \boldsymbol{\widehat{O}}}\right) \right]\nonumber\\ 
&= \sum_{\boldsymbol{\alpha}, O} \bra{\boldsymbol{\alpha}, O}_{\boldsymbol{A}, \boldsymbol{O}}\bra{\boldsymbol{\alpha}, O}_{\boldsymbol{\widehat{A}}, \boldsymbol{\widehat{O}}} \left[\sum_{\boldsymbol{\alpha}, O} P(\boldsymbol{\alpha}, O)\ket{\boldsymbol{\alpha}, O}\!\!\bra{\boldsymbol{\alpha}, O} \otimes \left(\bigotimes_{j = 1}^{m} \sigma_{\boldsymbol{\alpha}, O}^{(j)} \right)\widehat{E}^{\top} \right]\ket{\boldsymbol{\alpha}, O}_{\boldsymbol{A},\boldsymbol{O}}\ket{\boldsymbol{\alpha}, O}_{\boldsymbol{\widehat{A}}, \boldsymbol{\widehat{O}}} \label{eq:trace_dual_comb_multi_round}
\end{align}

By the fact that $\widehat{E}$ is a probabilistic comb dual to the $\bigotimes_{j = 1}^{m} \sigma_{\boldsymbol{\alpha}, O}^{(j)}$, $\left(\bigotimes_{j = 1}^{m} \sigma_{\boldsymbol{\alpha}, O}^{(j)} \right)\widehat{E}^{\top}$ is a (subnormalised) state in $\H_{B}$. Moreover, since $\bigotimes_{j = 1}^{m} \sigma_{\boldsymbol{\alpha}, O}^{(j)} = \bigotimes_{j = 1}^{m} \sigma_{\boldsymbol{\widetilde{\alpha}}, O}^{(j)}$ for all $\boldsymbol{\widetilde{\alpha}} \sim_{O} \boldsymbol{\alpha}$ by \Cref{lem:output_symmetry}, $\left(\bigotimes_{j = 1}^{m} \sigma_{\boldsymbol{\alpha}, O}^{(j)} \right)\widehat{E}^{\top}$ is the \textit{same} state for each $\widetilde{\alpha} \in \mathcal{A}_{\boldsymbol{\alpha}, O}$. For simplicity, let us replace the notation $\left(\bigotimes_{j = 1}^{m} \sigma_{\boldsymbol{\alpha}, O}^{(j)} \right)\widehat{E}^{\top}$ by $\rho_{\sim_{O} \boldsymbol{\alpha}}$. With this notation and invoking the assumption on $P(\boldsymbol{\alpha}, O)$, \Cref{eq:trace_dual_comb_multi_round} becomes
\begin{align}
\sum_{\boldsymbol{\alpha}, O} \bra{\boldsymbol{\alpha}, O}_{\boldsymbol{A}, \boldsymbol{O}}\bra{\boldsymbol{\alpha}, O}_{\boldsymbol{\widehat{A}}, \boldsymbol{\widehat{O}}} \left(\sum_{\boldsymbol{\alpha}, O} \frac{P(O)}{|\mathcal{A}|^{n}}\ket{\boldsymbol{\alpha}, O}\!\!\bra{\boldsymbol{\alpha}, O} \otimes \rho_{\sim_{O} \boldsymbol{\alpha}} \right)&\ket{\boldsymbol{\alpha}, O}_{\boldsymbol{A},\boldsymbol{O}}\ket{\boldsymbol{\alpha}, O}_{\boldsymbol{\widehat{A}}, \boldsymbol{\widehat{O}}} \nonumber \\ 
&= \frac{1}{|\mathcal{A}|^{n}} \sum_{\boldsymbol{\alpha}, O} P(O) \bra{\boldsymbol{\alpha}, O}_{\boldsymbol{\widehat{A}}, \boldsymbol{\widehat{O}}} \rho_{\sim_{O} \boldsymbol{\alpha}}\ket{\boldsymbol{\alpha}, O} \\
&= \frac{1}{|\mathcal{A}|^{n}} \sum_{O} P(O) \sum_{\boldsymbol{\alpha} \in \mathcal{A}^{n}/\sim_{O}} \sum_{\boldsymbol{\widetilde{\alpha}} \sim_{O} \boldsymbol{\alpha}}\bra{\boldsymbol{\widetilde{\alpha}}, O}_{\boldsymbol{\widehat{A}}, \boldsymbol{\widehat{O}}} \rho_{\sim_{O} \boldsymbol{\alpha}}\ket{\boldsymbol{\widetilde{\alpha}}, O}
\end{align}

where we are abusing notation slightly by denoting by $\boldsymbol{\alpha}$ the equivalence class $[\boldsymbol{\alpha}] \in \mathcal{A}^{n}/\sim_{O}$. Since $\rho_{\sim_{O} \boldsymbol{\alpha}}$ is a normalised state for every $\boldsymbol{\alpha}, O$ and the term $\sum_{\boldsymbol{\widetilde{\alpha}} \sim_{O} \boldsymbol{\alpha}}\bra{\boldsymbol{\widetilde{\alpha}}, O}_{\boldsymbol{A'},\boldsymbol{O'}} \rho_{\sim_{O} \boldsymbol{\alpha}, O}\ket{\boldsymbol{\widetilde{\alpha}}, O}$ can be interpreted as part of the trace over $\rho_{\sim_{O} \boldsymbol{\alpha}}$, we thus have
\begin{align}
\sum_{\boldsymbol{\widetilde{\alpha}} \sim_{O} \boldsymbol{\alpha}}\bra{\boldsymbol{\widetilde{\alpha}}, O}_{\boldsymbol{\widehat{A}}, \boldsymbol{\widehat{O}}} \rho_{\sim_{O} \boldsymbol{\alpha}, O}\ket{\boldsymbol{\widetilde{\alpha}}, O} \leq 1
\end{align}

Thus,
\begin{align}
\frac{1}{|\mathcal{A}|^{n}} \sum_{O} P(O) \sum_{\boldsymbol{\alpha} \in \mathcal{A}^{n}/\sim_{O}} \sum_{\boldsymbol{\widetilde{\alpha}} \sim_{O} \boldsymbol{\alpha}}\bra{\boldsymbol{\widetilde{\alpha}}, O}_{\boldsymbol{\widehat{A}}, \boldsymbol{\widehat{O}}} \rho_{\sim_{O} \boldsymbol{\alpha}}\ket{\boldsymbol{\widetilde{\alpha}}, O} &\leq \frac{1}{|\mathcal{A}|^{n}} \sum_{O} P(O) \sum_{\boldsymbol{\alpha} \in \mathcal{A}^{n}/\sim_{O}} 1 \\
&= \frac{1}{|\mathcal{A}|^{n}} \sum_{O} \frac{P(O)|\mathcal{A}|^{n}}{2^{|O|}} \\
&= \frac{1}{|\mathcal{O}|}\sum_{O}\frac{P(O)}{2^{|O|}}
\end{align}

Since this holds for any $\widehat{E}$, it in particular also holds for the one that maximises the trace, so
\begin{align}
\max_{\widehat{E}} \Tr_{\boldsymbol{A}, \boldsymbol{O}, \boldsymbol{\widehat{A}}, \boldsymbol{\widehat{O}}}\left[ D_{\client}^{(m)}\widehat{E}^{\top}\left(\sum_{\boldsymbol{\alpha}, O} \ket{\boldsymbol{\alpha},O,\boldsymbol{\alpha},O}\!\!\bra{\boldsymbol{\alpha},O,\boldsymbol{\alpha},O}_{\boldsymbol{A},\boldsymbol{O},\boldsymbol{\widehat{A}}, \boldsymbol{\widehat{O}}}\right) \right] \leq \sum_{O}\frac{P(O)}{2^{|O|}}
\end{align}

proving the theorem.
\end{proof}


\section{Grey Box MBQC Supplementary} \label{sec:app_grey_box_examples}

This section contains the proofs of the analytical results presented in \Cref{sec:grey_MBQC} as well as the gflow catalogue used in the examples of the same section.


\subsection{Correctness of $\sigma_{\MBQC}^{g}$} \label{subsec:app_sigma_MBQC_correctness}

The following proposition demonstrates that the operator $\sigma_{\MBQC}^{g} \ast_{\boldsymbol{Q'}} \rho_{G}$ correctly produces the right output state for all compatible measurement channels. We denote a measurement channel via its Choi representation as
\begin{align}
\M_{\alpha_{i}, \omega(i)} := \ket{0}\!\!\bra{0}_{C_{i}} \otimes \ket{+_{\alpha_{i}}}\!\!\bra{+_{\alpha_{i}}}_{\omega(i),Q_{i}}^{\top} + \ket{1}\!\!\bra{1}_{C_{i}} \otimes \ket{-_{\alpha_{i}}}\!\!\bra{-_{\alpha_{i}}}_{\omega(i),Q_{i}}^{\top}
\end{align}

where $\omega(i) \in \{XY, XZ, YZ\}$ denotes the measurement plane, $\alpha_{i}$ the angle from one of the positive axes in that plane, and the subscript $i$ indicates which qubits the measurement operator is related to. The partial transposes $\top$ act just on the $A_{i}$ factor of the operator. We denote by $\M_{\alpha_{i}, \omega(i)}^{c_{i}}$ the part of the operator related to outcome $c_{i}$, i.e. $\M_{\alpha_{i}, \omega(i)}^{1} = \ket{1}\!\!\bra{1}_{C_{i}} \otimes \ket{-_{\alpha_{i}}}\!\!\bra{-_{\alpha_{i}}}_{\omega(i),Q_{i}}$.

\begin{Proposition} \label{prop:correctness_sigma_gflow} Let $g$ be a gflow for $(G, I, O, \omega)$ where $G$ is a graph on vertices $V$, let $\sigma_{\MBQC}^{g}$ be defined as in \Cref{eq:sigma_MBQC}, and let $\mathcal{M}_{\alpha_{i}, \omega(i)}^{c_{i}}$ denote the $c_{i}$ measurement outcome of the measurement channels defined above ($c_{i} = 0$ for positive outcome and $c_{i} = 1$ for the negative outcome). Then
\begin{align}
\left(\bigotimes_{i \in V\setminus O} \mathcal{M}_{\alpha_{i}, \omega(i)}^{c_{i}} \right) \ast_{\boldsymbol{C}_{\setminus O},\boldsymbol{Q}_{\setminus O}} \sigma_{\MBQC}^{g} \ast_{\boldsymbol{Q'}} \rho_{G} \label{eq:sigma_gflow_correctness}
\end{align}

is the same state for all $\boldsymbol{c}_{\setminus O} = c_{1}...c_{|V\setminus O|}$. 
\end{Proposition}

\begin{proof} The proposition is a consequence of the fact that $\sigma_{\MBQC}^{g}$ is defined directly from gflow. We write out the details in part to highlight that no issues arise with the choice of ordering of $X$ and $Z$ operators in \Cref{eq:U_corr_i}. We proceed by showing that the state produced by any series of measurement outcomes is the same as that produced by all positive measurement outcomes (which encodes the computation for the MBQC):
\begin{align}
\left(\bigotimes_{i \in V\setminus O} \mathcal{M}_{\alpha_{i}, \omega(i)}^{c_{i}} \right) \ast_{\boldsymbol{C}_{\setminus 0}, \boldsymbol{Q}_{\setminus O}} \sigma_{\MBQC}^{g} \ast_{\boldsymbol{Q'}} \rho_{G} = \left(\bigotimes_{i \in V\setminus O} \mathcal{M}_{\alpha_{i}, \omega(i)}^{0} \right) \ast_{\boldsymbol{C}_{\setminus O}, \boldsymbol{Q}_{\setminus O}} \sigma_{\MBQC}^{g} \ast_{\boldsymbol{Q'}} \rho_{G} \label{eq:tensor_product_meas_equiv}
\end{align}

Consider first only a single measurement channel for qubit $i \in V\setminus O$, which obtains the negative outcome:
\begin{align}
\mathcal{M}_{\alpha_{i}, \omega(i)}^{1}\ast_{C_{i},Q_{i}} \sigma_{\MBQC}^{g} \ast_{\boldsymbol{Q'}} \rho_{G} &= \left(\ket{1}\!\!\bra{1}_{C_{i}} \otimes \ket{-_{\alpha_{i}}}\!\!\bra{-_{\alpha_{i}}}_{\omega(i),Q_{i}} \right) \ast_{C_{i},Q_{i}} \sigma_{\MBQC}^{g} \ast_{\boldsymbol{Q'}} \rho_{G}.
\end{align}

By the definition of gflow, we are guaranteed that
\begin{align}
\ket{-_{\alpha_{i}}}\!\!\bra{-_{\alpha_{i}}}_{\omega(i),Q_{i}}^{\top} \equiv K_{g(i)\vert_{i}}^{\dagger}\ket{+_{\alpha_{i}}}\!\!\bra{+_{\alpha_{i}}}_{\omega(i),Q_{i}}^{\top}K_{g(i)\vert_{i}}
\end{align}

where $K_{g(i)\vert_{i}}$ denotes the operator which appears as the tensor factor of $i$ in the stabiliser $K_{g(i)}$ and we have used that $K_{v}^{\dagger} \equiv K_{v}^{\top}$ for all $v$. Thus, using the properties of the link product, we have
\begin{align}
\left(\ket{1}\!\!\bra{1}_{C_{i}} \otimes \ket{-_{\alpha_{i}}}\!\!\bra{-_{\alpha_{i}}}_{\omega(i),Q_{i}}^{\top} \right) \ast_{C_{i},Q_{i}} \sigma_{\MBQC}^{g} \ast_{\boldsymbol{Q'}} \rho_{G} &= \left(\ket{1}\!\!\bra{1}_{C_{i}} \otimes K_{g(i)\vert_{i}}^{\dagger}\ket{+_{\alpha_{i}}}\!\!\bra{+_{\alpha_{i}}}_{\omega(i),Q_{i}}^{\top}K_{g(i)\vert_{i}} \right) \ast_{C_{i},Q_{i}} \sigma_{\MBQC}^{g} \ast_{\boldsymbol{Q'}} \rho_{G}  \\
&= \left(\ket{1}\!\!\bra{1}_{C_{i}} \otimes \ket{+_{\alpha_{i}}}\!\!\bra{+_{\alpha_{i}}}_{\omega(i),Q_{i}}^{\top} \right) \ast_{C_{i},Q_{i}} K_{g(i)\vert_{i}} \sigma_{\MBQC}^{g} K_{g(i)\vert_{i}}^{\dagger} \ast_{\boldsymbol{Q'}} \rho_{G}.
\end{align}

By contracting the link product over $C_{i}$ and writing $\sigma_{\MBQC}^{g}$ in the form of \Cref{eq:sigma_MBQC}, we get
\begin{align}
\ket{+_{\alpha_{i}}}\!\!\bra{+_{\alpha_{i}}}_{\omega(i),Q_{i}}^{\top} \ast_{Q_{i}} \left(\sum_{\boldsymbol{a}, \boldsymbol{b}, \boldsymbol{c}_{\setminus i}} K_{g(i)\vert_{i}} U_{\text{corr}(1_{i}\boldsymbol{c}_{\setminus i})} \ket{\boldsymbol{a}}\!\!\bra{\boldsymbol{b}}_{Q_{i}} U_{\text{corr}(1_{i}\boldsymbol{c}_{\setminus i})}^{\dagger} K_{g(i)\vert_{i}}^{\dagger} \otimes \ket{\boldsymbol{c}_{\setminus i}\boldsymbol{a}}\!\!\bra{\boldsymbol{c}_{\setminus i}\boldsymbol{b}}_{\boldsymbol{C}_{\setminus i}\boldsymbol{Q'}} \right) \ast_{\boldsymbol{Q'}} \rho_{G} \label{eq:single_meas_neg_outcome}
\end{align}

where the sum over $\boldsymbol{c}_{\setminus i}$ indicates the sum over basis element of all the $\H_{C_{j}}$ for $j \neq i$ and $1_{i}\boldsymbol{c}_{\setminus i}$ denotes the binary string $\boldsymbol{c}$ with the entry corresponding to $C_{i}$ a $1$ and the other entries given by $\boldsymbol{c}_{\setminus i}$. By the definition of the correction sets $\mathcal{X}_{j}$ and $\mathcal{Z}_{j}$, as well as the definition of $U_{\text{corr}(\boldsymbol{c})}$, we can write
\begin{align}
U_{\text{corr}(1_{i}\boldsymbol{c}_{\setminus i})} = (-1)^{f(\boldsymbol{c}_{\setminus i})}U_{\text{corr}(0_{i}\boldsymbol{c}_{\setminus i})}K_{g(i)\vert_{\setminus i}}
\end{align}

where the exponent $f(\boldsymbol{c}_{\setminus i})$ encodes the coefficient that arises from commuting the factors of $K_{g(i)\vert_{\setminus i}}$ through the other factors of $U_{\text{corr}(1_{i}\boldsymbol{c}_{\setminus i})}$ and where $U_{\text{corr}(0_{i}\boldsymbol{c}_{\setminus i})}$ denotes the correction operator with $c_{i} = 0$. In particular, this means that
\begin{align}
K_{g(i)\vert_{i}} U_{\text{corr}(1_{i}\boldsymbol{c}_{\setminus i})} &= (-1)^{f(\boldsymbol{c}_{\setminus i}) + t_{i}}U_{\text{corr}(0_{i}\boldsymbol{c}_{\setminus i})}K_{g(i)\vert_{i}} K_{g(i)\vert_{\setminus i}}  \\
&= (-1)^{f(\boldsymbol{c}_{\setminus i}) + t_{i}}U_{\text{corr}(0_{i}\boldsymbol{c}_{\setminus i})}K_{g(i)}
\end{align}

where $t_{i}$ encodes the coefficient arises from commuting $K_{g(i)\vert_{i}}$ through $U_{\text{corr}(0_{i}\boldsymbol{c}_{\setminus i})}$. Since the same coefficient arises from the conjugate term (i.e. $U_{\text{corr}(1_{i}\boldsymbol{c}_{\setminus i})}^{\dagger}K_{g(i)\vert_{i}}^{\dagger}$), the net result is that no $-1$ factor can appear. That is, we have
\begin{align}
K_{g(i)\vert_{i}} U_{\text{corr}(1_{i}\boldsymbol{c}_{\setminus i})} (\cdot)  U_{\text{corr}(1_{i}\boldsymbol{c}_{\setminus i})}^{\dagger}K_{g(i)\vert_{i}}^{\dagger} = U_{\text{corr}(0_{i}\boldsymbol{c}_{\setminus i})}K_{g(i)}(\cdot)K_{g(i)}^{\dagger}U_{\text{corr}(0_{i}\boldsymbol{c}_{\setminus i})}^{\dagger}
\end{align}

The cancelling of any phase factor arising from commuting $X$ and $Z$ factors in $U_{\text{corr}(\boldsymbol{c})}$ is the reason why the ordering in \Cref{eq:U_corr_i} is justified. Continuing from \Cref{eq:single_meas_neg_outcome}, we have
\begin{align}
\ket{+_{\alpha_{i}}}\!\!\bra{+_{\alpha_{i}}}_{\omega(i),Q_{i}}^{\top} &\ast_{Q_{i}} \left(\sum_{\boldsymbol{a}, \boldsymbol{b}, \boldsymbol{c}_{\setminus i}} U_{\text{corr}(0_{i}\boldsymbol{c}_{\setminus i})}K_{g(i)} \ket{\boldsymbol{a}}\!\!\bra{\boldsymbol{b}}_{\boldsymbol{Q}} K_{g(i)}^{\dagger}U_{\text{corr}(0_{i}\boldsymbol{c}_{\setminus i})}^{\dagger} \otimes \ket{\boldsymbol{c}_{\setminus i}\boldsymbol{a}}\!\!\bra{\boldsymbol{c}_{\setminus i}\boldsymbol{b}}_{\boldsymbol{C}_{\setminus i}\boldsymbol{Q'}} \right) \ast_{\boldsymbol{Q'}} \rho_{G} \\
&= \ket{+_{\alpha_{i}}}\!\!\bra{+_{\alpha_{i}}}_{\omega(i),Q_{i}}^{\top} \ast_{Q_{i}} \left(\sum_{\boldsymbol{a}, \boldsymbol{b}, \boldsymbol{c}_{\setminus i}} U_{\text{corr}(0_{i}\boldsymbol{c}_{\setminus i})} \ket{\boldsymbol{a}}\!\!\bra{\boldsymbol{b}}_{\boldsymbol{Q}} U_{\text{corr}(0_{i}\boldsymbol{c}_{\setminus i})}^{\dagger}  \otimes \ket{\boldsymbol{c}_{\setminus i}}\!\!\bra{\boldsymbol{c}_{\setminus i}}_{\boldsymbol{C}_{\setminus i}} \otimes K_{g(i)}\ket{\boldsymbol{a}}\!\!\bra{\boldsymbol{b}}_{\boldsymbol{Q'}}K_{g(i)}^{\dagger} \right) \ast_{\boldsymbol{Q'}} \rho_{G} \label{eq:properties_of_choi}\\
&\equiv \left(\ket{0}\!\!\bra{0}_{C_{i}} \otimes \ket{+_{\alpha_{i}}}\!\!\bra{+_{\alpha_{i}}}_{\omega(i),Q_{i}}^{\top} \right) \ast_{C_{i},Q_{i}} K_{g(i), \boldsymbol{Q'}} \sigma_{\MBQC}^{g} K_{g(i), \boldsymbol{Q'}}^{\dagger} \ast_{\boldsymbol{Q'}} \rho_{G} \label{eq:K_g_Aprime} \\
&= \mathcal{M}_{\alpha_{i}, \omega(i)}^{0} \ast_{C_{i},Q_{i}}  \sigma_{\MBQC}^{g}  \ast_{\boldsymbol{Q'}} K_{g(i)} \rho_{G} K_{g(i)}^{\dagger}\\
&= \mathcal{M}_{\alpha_{i}, \omega(i)}^{0} \ast_{C_{i},Q_{i}}  \sigma_{\MBQC}^{g}  \ast_{\boldsymbol{Q'}} \rho_{G} 
\end{align}

where the extra subscript $\boldsymbol{Q'}$ on the operators in \Cref{eq:K_g_Aprime} indicates that they act on the appropriate spaces $Q'_{j}$ rather than $Q_{j}$ as before, and can thus be transferred to the $\rho_{G}$ via the properties of the link product (which also uses the fact that $K_{g(i)}^{\dagger} = K_{g(i)}^{T}$). We have also made use of properties of the Choi operator in \Cref{eq:properties_of_choi} to transfer the $K_{g(i)}$ from the $\boldsymbol{Q}$-spaces to the $\boldsymbol{Q'}$-spaces in the first place.

Since the left-hand side of \Cref{eq:tensor_product_meas_equiv} contains a tensor product of terms $\mathcal{M}_{\alpha_{i}, \omega(i)}^{c_{i}}$, we can apply the same reasoning as above to each factor for which $c_{i} = 1$ which thus establishes that \Cref{eq:tensor_product_meas_equiv} does indeed hold for any measurement outcomes $\boldsymbol{c}$. This means that once all the link products are evaluated, which in particular includes the trace over all spaces except for $\bigotimes_{i \in O} \H_{Q_{i}}$, the same state is produced on the output space.
\end{proof}


\subsection{Equivalence of Gflows Operators} \label{subsec:app_gflow_causal_equivalence}

The following is the proof of \Cref{prop:no_causal_learning}, which showed that, in the absence of noise, the gflows for a given measurement plane and partial order are indistinguishable.

\begin{proof}[Proof of \Cref{prop:no_causal_learning}] Consider
\begin{align}
\left(\bigotimes_{j = 1}^{|V \setminus O|} \ket{c_{j}}\!\!\bra{c_{j}}_{C_{j}} \right) \ast_{\boldsymbol{C}_{\setminus O}} \sigma_{\MBQC}^{g} \ast_{\boldsymbol{Q'}} \rho_{G}
\end{align}

for some $g \sim (G, I, O, \omega)$. Evaluating the left-hand link product gives
\begin{align}
\left(\sum_{\boldsymbol{a}, \boldsymbol{b}} U_{\text{corr}(\boldsymbol{c}_{\setminus O})}^{g}\ket{\boldsymbol{a}}\!\!\bra{\boldsymbol{b}}_{\boldsymbol{Q}}(U_{\text{corr}(\boldsymbol{c}_{\setminus O})}^{g})^{\dagger} \otimes \ket{\boldsymbol{a}}\!\!\bra{\boldsymbol{b}}_{\boldsymbol{Q'}} \right) \ast_{\boldsymbol{Q'}} \rho_{G} = \left(\sum_{\boldsymbol{a}, \boldsymbol{b}}\ket{\boldsymbol{a}}\!\!\bra{\boldsymbol{b}}_{\boldsymbol{Q}} \otimes \ket{\boldsymbol{a}}\!\!\bra{\boldsymbol{b}}_{\boldsymbol{Q'}} \right) \ast_{\boldsymbol{Q'}}  U_{\text{corr}(\boldsymbol{c}_{\setminus O})}^{g} \rho_{G} (U_{\text{corr}(\boldsymbol{c}_{\setminus O})}^{g})^{\dagger}
\end{align}

where $\boldsymbol{c}_{\setminus O} = c_{1}...c_{|V \setminus O|}$. To prove the proposition, it suffices to show that all the $U_{\text{corr}(\boldsymbol{c}_{\setminus O})}^{g}$ as $g$ varies are mutually related by stabilisers of $G$ for each $\boldsymbol{c}_{\setminus O}$, up to a phase (which is cancelled by the corresponding conjugate phase from the adjoint $(U_{\text{corr}(\boldsymbol{c}_{\setminus O})}^{g})^{\dagger}$). 

Let $g, g' \sim (G, I, O, \omega)$ be arbitrary. By definition of $U_{\text{corr}(\boldsymbol{c}_{\setminus O})}$, we have
\begin{align}
U_{\text{corr}(\boldsymbol{c}_{\setminus O})}^{g} &\propto \prod_{j = 1}^{|V \setminus O|} K_{g(j)}^{c_{j}}\vert_{\setminus j} \\
U_{\text{corr}(\boldsymbol{c}_{\setminus O})}^{g'} &\propto \prod_{j = 1}^{|V \setminus O|} K_{g'(j)}^{c_{j}}\vert_{\setminus j} 
\end{align}

where $\propto$ denotes that a $-1$ phase may arise from commuting $X$ and $Z$ operators to arrive at the above form of the operator in terms of $K_{g(j)}$ from the canonical form of $U_{\text{corr}(\boldsymbol{c}_{\setminus O})}^{g}$ (respectively $U_{\text{corr}(\boldsymbol{c}_{\setminus O})}^{g'}$) as in \Cref{eq:U_corr} and \Cref{eq:U_corr_i}. Despite that $K_{g(j)}$ and $K_{g'(j)}$ may be different stabilisers, by the fact that both $g$ and $g'$ observe the same measurement planes, the $j$th tensor factor of each is the same. Thus,
\begin{align}
U_{\text{corr}(\boldsymbol{c}_{\setminus O})}^{g} &\propto \left(\bigotimes_{l = 1}^{|V \setminus O|} K_{g(l)}^{c_{l}}\vert_{l} \right) \prod_{j = 1}^{|V \setminus O|} K_{g(j)} \\
&= \left(\bigotimes_{l = 1}^{|V \setminus O|} K_{g(l)}^{c_{l}}\vert_{l} \right) \prod_{j = 1}^{|V \setminus O|} K_{g'(j)} K_{g'(j)} K_{g(j)}  \\
&=  \left(\bigotimes_{l = 1}^{|V \setminus O|} K_{g'(l)}^{c_{l}}\vert_{l} \right) \left( \prod_{j = 1}^{|V \setminus O|} K_{g'(j)} \right) \left(\prod_{i = 1}^{|V \setminus O|}  K_{g'(i)} K_{g(i)} \right) \\
&\propto U_{\text{corr}(\boldsymbol{c}_{\setminus O})}^{g'}\left(\prod_{i = 1}^{|V \setminus O|} K_{g(i)} K_{g'(i)}  \right).
\end{align}

The restriction to only those gflows $g \sim (G, I, O, \omega)$ that have mutually compatible partial orders that is made in the statement of the proposition is required when considering $\sigma_{\MBQC}^{g}$ with a total order on input and output spaces.
\end{proof}


\subsection{Gflow-Induced Quantum Causal Models} \label{subsec:app_gflow_dag}

This subsection provides some further details regarding the the quantum causal model (QCM) induced by gflow for MBQC and thus supports \Cref{subsec:caus_inf} in the main text. Firstly, we confirm here that gflow does indeed define a DAG, which is a requirement for showing \Cref{prop:MBQC_QCM}. Thereafter we provide a table that elucidates the comparisons between the components of the QCM defined for MBQC and the components of classical causal models as presented eg., by Pearl in \cite{pearl2009causality}.

\begin{Proposition} \label{prop:acyclic} Let $(G, I, O, \omega)$ be such that a gflow exists and let $(g, <)$ be a choice of such a gflow. The directed graph $\overline{G}$ on vertex set $V$ and edge set defined from the gflow via 
\begin{align*}
E &:= \left\{(i, j) \in V \times V| j\in V, i \in \mathcal{X}_{j} \cup \mathcal{Z}_{j} \right\}
\end{align*}

is acyclic.
\end{Proposition}

\begin{proof} Suppose for a contradiction that $\overline{G}$ is not acyclic, that is, there exists a directed path $v_{0} \rightarrow v_{1} \rightarrow ... \rightarrow v_{k} = v_{0}$ for some sequence of vertices $v_{0}, ..., v_{k}$. By the definition of the edge set of $\overline{G}$, it follows that either $v_{i+1} \in g(v_{i})\setminus\{v_{i}\}$ or $v_{i} \neq v_{i+1} \in \Od(g(v_{i}))$. In either case $v_{i} < v_{i+1}$ in the ordering of the given gflow. Thus, $v_{0} < v_{k} = v_{0}$ by transitivity of the order, which gives the desired contradiction as the order is strict by definition.
\end{proof}

Classical causal models are typically presented as consisting of a set of observed variables $V = \{V_{1}, ..., V_{n}\}$ (those within the model), a set of unobserved variables $U = \{U_{1}, ..., U_{n}\}$ (those determined by factors outside of the model), a DAG that specifies which variables (both unobserved and observed) that have a causal influence on a given observed variable, and set of functions $F = \{f_{i}: \Pa(V_{i}) \cup U_{i} \rightarrow V_{i} \}_{i=1}^{n}$ that map from the parents of a variable to the variable itself. These functions are often called structural equations and uniquely specify the values for the observed variables given values of the unobserved variables. A probabilistic causal model further includes a distribution $P(U)$ over the unobserved variables. Using this terminology, the following table outlines the correspondence between classical causal models and the QCM $\sigma_{\MBQC}^{g}$ defined in \Cref{subsec:caus_inf}.

\begin{center}
\begin{tabular}{|c|c|c|c|c|}
\hline
&\multicolumn{4}{c|}{Comparing the QCM $\sigma_{\MBQC}^{g}$ to Classical Causal Models}\\
\hline
\textbf{Classical Causal Model} & Unobserved Variable $U_{i}$ & Observed Variable $V_{i}$ & $f_{i}: \Pa(V_{i}) \cup U_{i} \rightarrow V_{i}$ & DAG \\
\textbf{MBQC QCM $\sigma_{\MBQC}^{g}$} & Input Space $\H_{A'_{i}}$ & Quantum Node $\H_{A_{i}} \otimes \H_{C_{i}}$ & $\rho_{A_{i}|C_{j:j\in\Pa(i)}, A'_{i}}$  & $\overline{G}$\\
\hline
\end{tabular}
\end{center}


\subsection{Gflow Catalogue} \label{subsec:app_gflow_catalogue}

\Cref{fig:XY_gflows}, \Cref{fig:XZ_gflows}, and \Cref{fig:YZ_gflows} depict the DAGs corresponding to the $15$ different gflows for the four-vertex graph considered in \Cref{sec:grey_MBQC} and depicted in \Cref{fig:grey_box_graph}, grouped by the assigned measurement plane for the second qubit (the first is always measured in the $XY$-plane). The details of the gflow and corresponding $U_{\text{corr}(\boldsymbol{c})}$ for each DAG are gives in the caption.


\section{Details of SDP Implementation} \label{sec:app_SDP_implementation}

The numerical calculations of the guessing probabilities throughout this work made use of the convex optimisation library CVXPY \cite{diamond2016cvxpy, agrawal2018rewriting}. We primarily used the Splitting Conic Solver (SCS) \cite{scs}. Our code is provided at \cite{Smith_Min_Entropy_and_MBQC}. Calculations were run on an HP $Z4$ $G4$ $9980XE$ workstation. 

To provide some indication of the limitations of the numerical approach in its current form, to complement those detailed in \Cref{subsec:limitations}, we document here some of hurdles we faced when calculating guessing probabilities. For the BQC examples, we generated the classical combs as $1$-dimensional objects representing the corresponding diagonals. Despite this, size issues played a role when storing the combs even before calling the SDP solver: $D_{\client}$ for a single round and for any $\mathcal{A}$ of size greater than $12$ and $D_{\client}^{(2)}$ for $\mathcal{A}$ of size greater than $4$ caused problems. The guessing probability for $D_{\client}^{(1)}$ for $|\mathcal{A}| = 4$ was calculated within a day, whereas for $|\mathcal{A}| = 8$, it took on the order of a week and required approximately $110 GB$ of RAM. Attempting to calculate the guessing probability for $D_{\client}^{(2)}$ with $|\mathcal{A}| = 4$ exceeded the available RAM of the machine.

For the Grey Box MBQC examples, the guessing probabilities for all single round combs could be calculated quickly (on the order of minutes) however all multi-round combs again caused size problems. Typically, the primary bottleneck occurred when enforcing the comb constraints (i.e. sequential partial trace constraints) on the variable in the SDP solver, which tended to dominate the runtime of the algorithm. Otherwise, some size error would occur during the Cone Matrix Stuffing reduction phase of the solver (see the CVXPY documentation for details).

\begin{figure}
\captionsetup{justification=centering}
\begin{subfigure}{0.3\textwidth}
\captionsetup{justification=centering}
\centering
\includegraphics[width=0.7\textwidth]{Figures/g1}
\caption{$g_{1}: 1 \mapsto \{2\}, 2 \mapsto \{3,4\}$, \newline $U_{\text{corr}(\boldsymbol{c})}^{g_{1}} = X_{2}^{c_{1}} \otimes X_{3}^{c_{2}}Z_{3}^{c_{2}} \otimes X_{4}^{c_{2}}Z_{4}^{c_{1}c_{2}}$}
\label{fig:g1}
\end{subfigure}
\hfill
\begin{subfigure}{0.3\textwidth}
\captionsetup{justification=centering}
\centering
\includegraphics[width=0.7\textwidth]{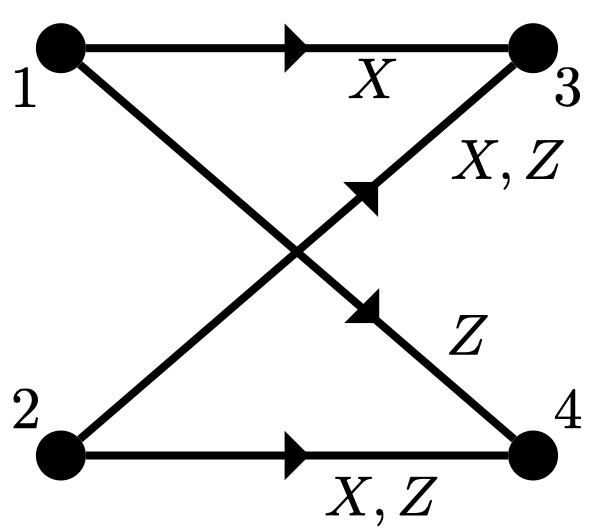}
\caption{$g_{2}: 1 \mapsto \{3\}, 2 \mapsto \{3,4\}$, \newline $U_{\text{corr}(\boldsymbol{c})}^{g_{2}} = X_{3}^{c_{1}c_{2}}Z_{3}^{c_{2}} \otimes X_{4}^{c_{2}}Z_{4}^{c_{1}c_{2}}$}
\label{fig:g2}
\end{subfigure}
\hfill
\begin{subfigure}{0.3\textwidth}
\captionsetup{justification=centering}
\centering
\includegraphics[width=0.7\textwidth]{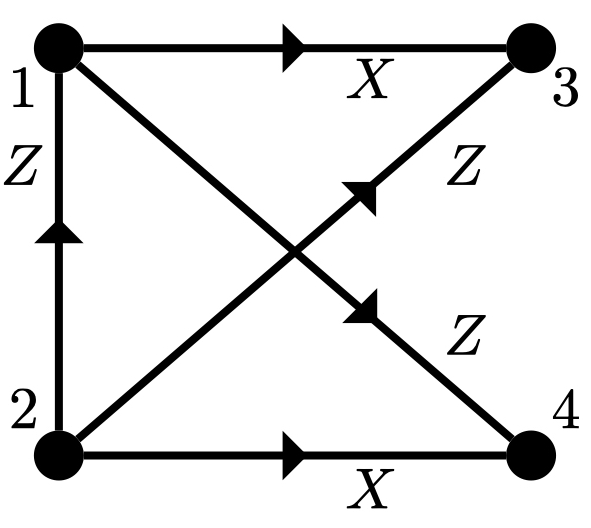}
\caption{$g_{3}: 1 \mapsto \{3\}, 2 \mapsto \{4\}$, \newline $U_{\text{corr}(\boldsymbol{c})}^{g_{3}} = Z_{1}^{c_{2}} \otimes X_{3}^{c_{1}}Z_{3}^{c_{2}} \otimes X_{4}^{c_{2}}Z_{4}^{c_{1}}$}
\label{fig:g3}
\end{subfigure} \vspace{15mm}\\

\begin{subfigure}{0.5\textwidth}
\captionsetup{justification=centering}
\centering
\includegraphics[width=0.425\textwidth]{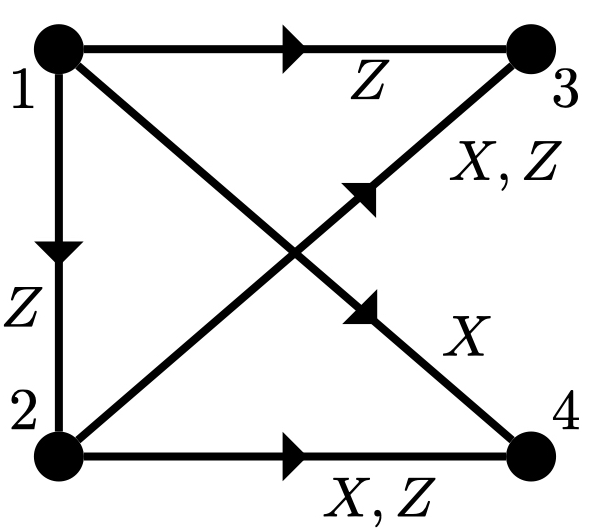}
\caption{$g_{4}: 1 \mapsto \{4\}, 2 \mapsto \{3,4\}$, \newline $U_{\text{corr}(\boldsymbol{c})}^{g_{4}} =  Z_{2}^{c_{1}}\otimes X_{3}^{c_{2}}Z_{3}^{c_{1}c_{2}} \otimes X_{4}^{c_{1}c_{2}}Z_{4}^{c_{2}}$}
\label{fig:g4}
\end{subfigure}
\hfill
\begin{subfigure}{0.5\textwidth}
\captionsetup{justification=centering}
\centering
\includegraphics[width=0.425\textwidth]{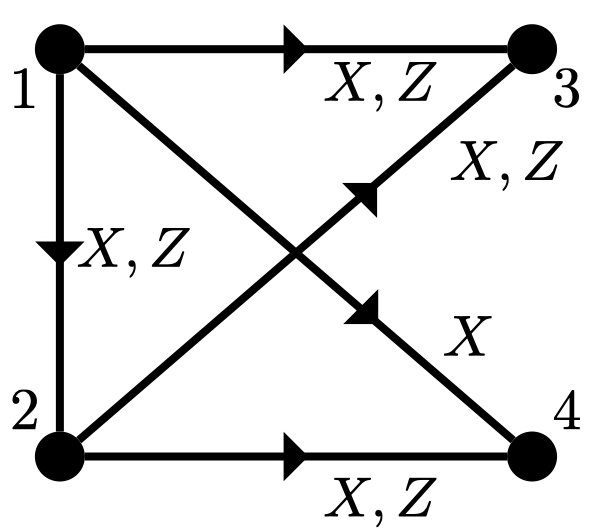}
\caption{$g_{5}: 1 \mapsto \{2, 3, 4\}, 2 \mapsto \{3,4\}$, \newline $U_{\text{corr}(\boldsymbol{c})}^{g_{5}} = X_{2}^{c_{1}}Z_{2}^{c_{1}}\otimes X_{3}^{c_{1}c_{2}}Z_{3}^{c_{1}c_{2}} \otimes X_{4}^{c_{1}c_{2}}Z_{4}^{c_{2}}$}
\label{fig:g5}
\end{subfigure}
\caption{(a) - (e) display the directed acyclic graphs and correction operators for the gflows $g_{1}, ..., g_{5}$ respectively which are compatible with $(G,I,O, \omega)$ where $\omega(1) = \omega(2) = XY$. The corresponding captions detail the gflows themselves and the associated $U_{\text{corr}(\boldsymbol{c})}$. The partial order for $g_{1}, g_{2}, g_{4}$ and $g_{5}$ is given by $1 < 2$ , and that for $g_{3}$ is $2<1$.}
\label{fig:XY_gflows}
\end{figure}

\begin{figure}
\captionsetup{justification=centering}
\begin{subfigure}{0.3\textwidth}
\captionsetup{justification=centering}
\centering
\includegraphics[width=0.7\textwidth]{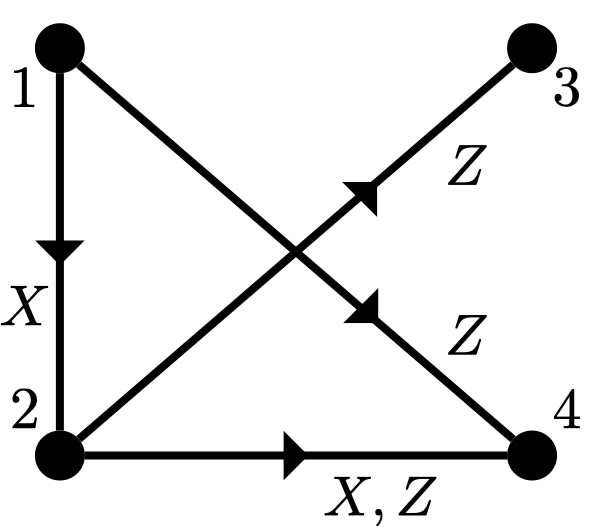}
\caption{$g_{6}: 1 \mapsto \{2\}, 2 \mapsto \{2,4\}$, \newline $U_{\text{corr}(\boldsymbol{c})}^{g_{6}} =  X_{2}^{c_{1}}\otimes Z_{3}^{c_{2}} \otimes X_{4}^{c_{2}}Z_{4}^{c_{1}c_{2}}$}
\label{fig:g6}
\end{subfigure}
\hfill
\begin{subfigure}{0.3\textwidth}
\captionsetup{justification=centering}
\centering
\includegraphics[width=0.7\textwidth]{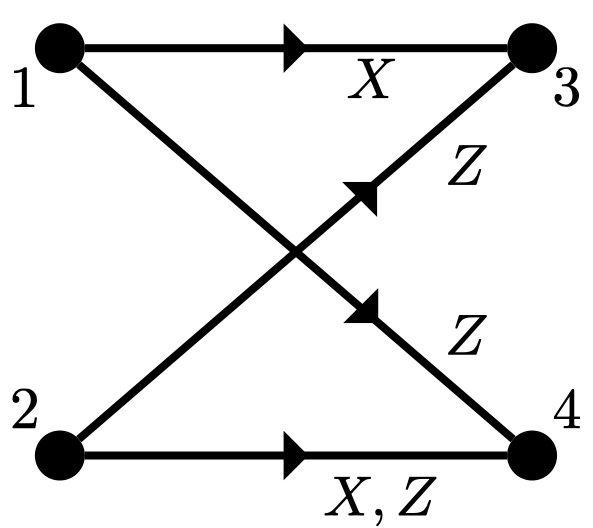}
\caption{$g_{7}: 1 \mapsto \{3\}, 2 \mapsto \{2,4\}$, \newline $U_{\text{corr}(\boldsymbol{c})}^{g_{7}} =  X_{3}^{c_{1}}Z_{3}^{c_{2}} \otimes X_{4}^{c_{2}}Z_{4}^{c_{1}c_{2}}$}
\label{fig:g7}
\end{subfigure} 
\hfill
\begin{subfigure}{0.3\textwidth}
\captionsetup{justification=centering}
\centering
\includegraphics[width=0.7\textwidth]{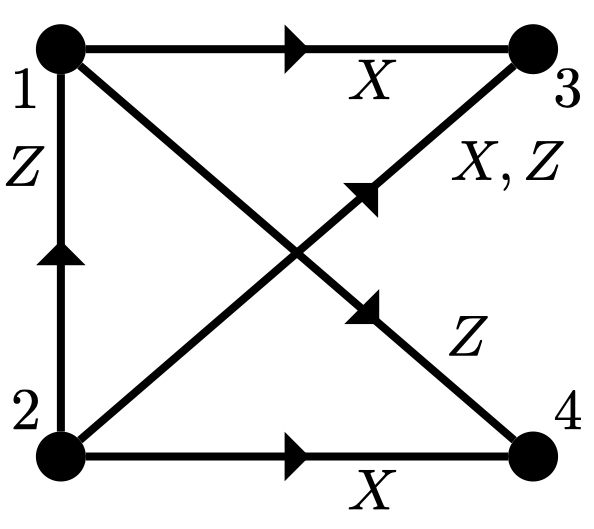}
\caption{$g_{8}: 1 \mapsto \{3\}, 2 \mapsto \{2,3,4\}$, \newline $U_{\text{corr}(\boldsymbol{c})}^{g_{8}} = Z_{1}^{c_{2}} \otimes X_{3}^{c_{1}c_{2}}Z_{3}^{c_{2}} \otimes X_{4}^{c_{2}}Z_{4}^{c_{1}}$}
\label{fig:g8}
\end{subfigure} \vspace{15mm} \\
\begin{subfigure}{0.5\textwidth}
\captionsetup{justification=centering}
\centering
\includegraphics[width=0.425\textwidth]{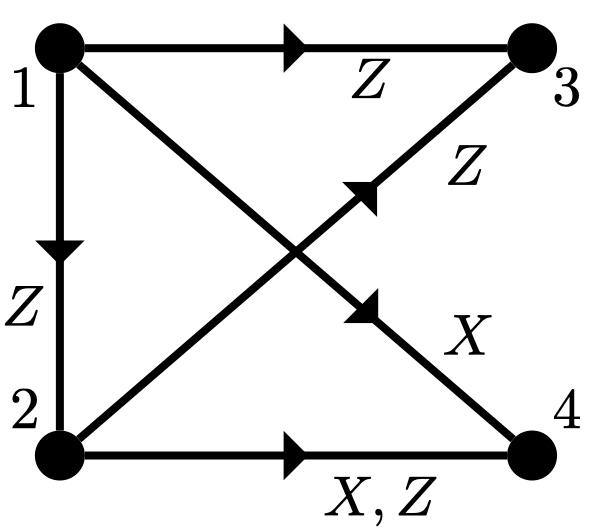}
\caption{$g_{9}: 1 \mapsto \{4\}, 2 \mapsto \{2,4\}$, \newline $U_{\text{corr}(\boldsymbol{c})}^{g_{9}} = Z_{2}^{c_{1}}\otimes Z_{3}^{c_{1}c_{2}} \otimes X_{4}^{c_{1}c_{2}}Z_{4}^{c_{2}}$}
\label{fig:g9}
\end{subfigure} 
\hfill
\begin{subfigure}{0.5\textwidth}
\captionsetup{justification=centering}
\centering
\includegraphics[width=0.425\textwidth]{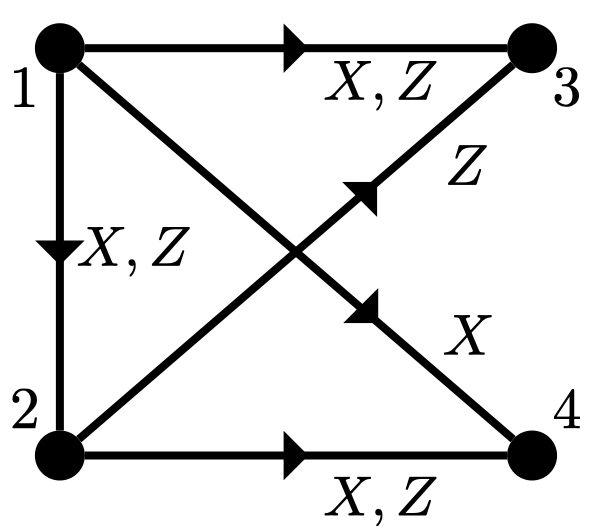}
\caption{$g_{10}: 1 \mapsto \{2,3,4\}, 2 \mapsto \{2,4\}$, \newline $U_{\text{corr}(\boldsymbol{c})}^{g_{10}} =  X_{2}^{c_{1}}Z_{2}^{c_{1}}\otimes X_{3}^{c_{1}}Z_{3}^{c_{1}c_{2}} \otimes X_{4}^{c_{1}c_{2}}Z_{4}^{c_{2}}$}
\label{fig:g10}
\end{subfigure}
\caption{(a) - (e) display the directed acyclic graphs and correction operators for the gflows $g_{6}, ..., g_{10}$ respectively which are compatible with $(G,I,O, \omega)$ where $\omega(1) = XY$ and $\omega(2) = XZ$. The corresponding captions detail the gflows themselves and the associated $U_{\text{corr}(\boldsymbol{c})}$. The partial order for $g_{6}, g_{7}, g_{9}$ and $g_{10}$ is given by $1 < 2$ , and that for $g_{8}$ is $2<1$.}
\label{fig:XZ_gflows}
\end{figure}

\begin{figure}
\captionsetup{justification=centering}
\begin{subfigure}{0.3\textwidth}
\captionsetup{justification=centering}
\centering
\includegraphics[width=0.7\textwidth]{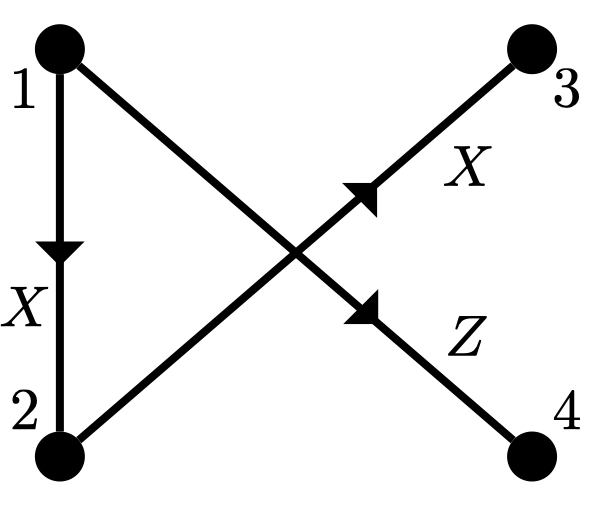}
\caption{$g_{11}: 1 \mapsto \{2\}, 2 \mapsto \{2,3\}$, \newline $U_{\text{corr}(\boldsymbol{c})}^{g_{11}} =  X_{2}^{c_{1}}\otimes X_{3}^{c_{2}} \otimes Z_{4}^{c_{1}}$}
\label{fig:g11}
\end{subfigure}
\hfill
\begin{subfigure}{0.3\textwidth}
\captionsetup{justification=centering}
\centering
\includegraphics[width=0.7\textwidth]{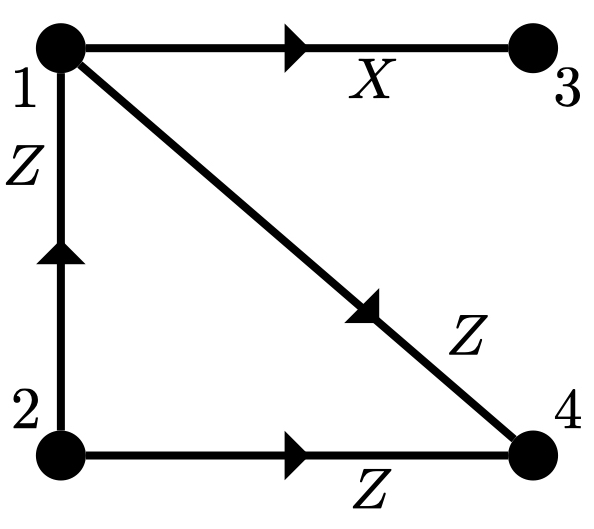}
\caption{$g_{12}: 1 \mapsto \{3\}, 2 \mapsto \{2\}$, \newline $U_{\text{corr}(\boldsymbol{c})}^{g_{12}} = Z_{1}^{c_{2}} \otimes X_{3}^{c_{1}} \otimes Z_{4}^{c_{1}c_{2}}$}
\label{fig:g12}
\end{subfigure}
\hfill
\begin{subfigure}{0.3\textwidth}
\captionsetup{justification=centering}
\centering
\includegraphics[width=0.7\textwidth]{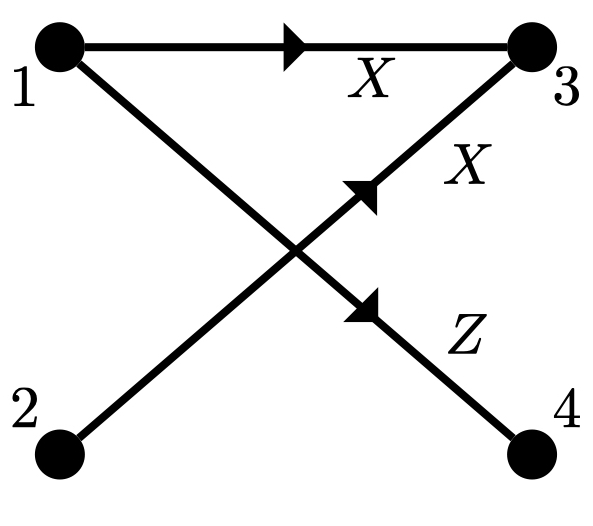}
\caption{$g_{13}: 1 \mapsto \{3\}, 2 \mapsto \{2,3\}$, \newline $U_{\text{corr}(\boldsymbol{c})}^{g_{13}} = X_{3}^{c_{1}c_{2}} \otimes Z_{4}^{c_{1}}$}
\label{fig:g13}
\end{subfigure} \vspace{15mm} \\

\begin{subfigure}{0.5\textwidth}
\captionsetup{justification=centering}
\centering
\includegraphics[width=0.425\textwidth]{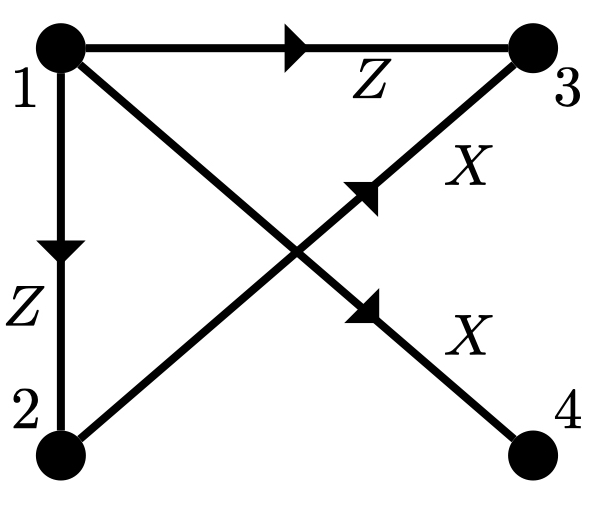}
\caption{$g_{14}: 1 \mapsto \{4\}, 2 \mapsto \{2,3\}$, \newline $U_{\text{corr}(\boldsymbol{c})}^{g_{14}} = Z_{2}^{c_{1}}\otimes X_{3}^{c_{2}}Z_{3}^{c_{1}} \otimes X_{4}^{c_{1}}$}
\label{fig:g14}
\end{subfigure}
\hfill
\begin{subfigure}{0.5\textwidth}
\captionsetup{justification=centering}
\centering
\includegraphics[width=0.425\textwidth]{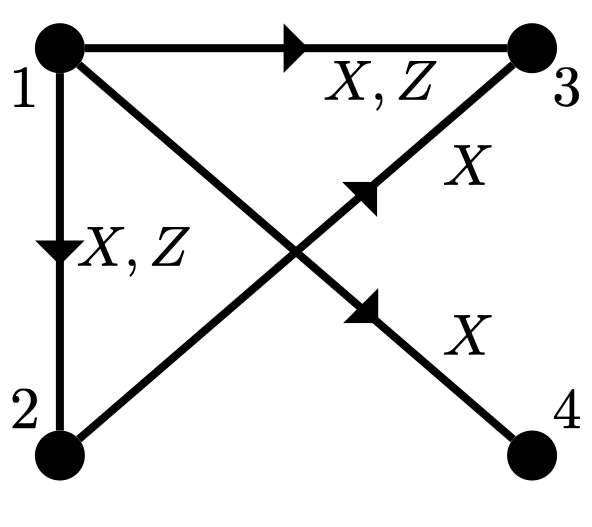}
\caption{$g_{15}: 1 \mapsto \{2,3,4\}, 2 \mapsto \{2,3\}$, \newline $U_{\text{corr}(\boldsymbol{c})}^{g_{15}} = X_{2}^{c_{1}}Z_{2}^{c_{1}}\otimes X_{3}^{c_{1}c_{2}}Z_{3}^{c_{1}} \otimes X_{4}^{c_{1}}$}
\label{fig:g15}
\end{subfigure}
\caption{(a) - (e) display the directed acyclic graphs and correction operators for the gflows $g_{11}, ..., g_{15}$ respectively which are compatible with $(G,I,O, \omega)$ where $\omega(1) = XY$ and $\omega(2) = YZ$. The corresponding captions detail the gflows themselves and the associated $U_{\text{corr}(\boldsymbol{c})}$. The partial order for $g_{11}, g_{13}, g_{14}$ and $g_{15}$ is given by $1 < 2$ , and that for $g_{12}$ is $2<1$.}
\label{fig:YZ_gflows}
\end{figure}

\end{document}